\documentclass[a4paper, 11pt]{book}

\usepackage{graphicx}
\usepackage{amsmath}
\usepackage{amssymb}
\usepackage{theorem}
\usepackage{fancyhdr}
\usepackage{setspace}
\usepackage[margin=4.5cm]{geometry}
\usepackage{makeidx}
\usepackage{array}
\usepackage{longtable}
\theorembodyfont{\rmfamily}    
\newtheorem{theorem}{Theorem}[section]
\newtheorem{lemma}[theorem]{Lemma}
\newtheorem{definition}[theorem]{Definition}
\newtheorem{corollary}[theorem]{Corollary}

\newtheorem{proposition}[theorem]{Proposition}
\newtheorem{remark}[theorem]{Remark}
\newcommand{\qed}{\hfill$\square$}

\newenvironment{proof}{%
  \noindent{\em Proof.\ }}{%
  \hspace*{\fill}\qed \\
  \vspace{2ex}}
\newcommand{\braket}[2]{\langle #1 | #2 \rangle} 
\newcommand{\ket}[1]{| #1 \rangle} 
\newcommand{\bra}[1]{\langle #1 |} 
\newcommand{\bol}[1]{\mathbf{#1}}
\newcommand{\rom}[1]{\mathrm{#1}}
\newcommand{\san}[1]{\mathsf{#1}}
\newcommand{\mymid}{:~}
\newcommand{\mix}{\rom{mix}}
\newcommand{\argmax}{\mathop{\rm argmax}\limits}
\newcommand{\argmin}{\mathop{\rm argmin}\limits}

\allowdisplaybreaks
\makeindex

\begin{document}
\begin{titlepage}
\begin{center}
\vspace*{1in}
{\Large A Study of Channel Estimation and Postprocessing  \\ in Quantum Key
 Distribution Protocols}
\par
\vspace{1.0in}
{\Large Shun Watanabe} 
\par
\vspace{1.0in}
{\Large Supervisor: Prof.~Ryutaroh Matsumoto}
\par
\vspace{1.0in}
{\large A Thesis submitted for the degree of Doctor of Philosophy}
\par
\vspace{0.5in}
{\large Department of Communications and Integrated Systems}
\par
\vspace{0.5in}
{\large Tokyo Institute of Technology}
\par
\vspace{0.5in}
{\large 2009}
\end{center}
\end{titlepage}

\pagestyle{empty}

\chapter*{Acknowledgments}
\pagenumbering{roman}
\setcounter{page}{1}

First of all, I wish to express my sincere gratitude and special
thanks to my supervisor, Prof.~Ryutaroh Matsumoto for
his constant guidance and close supervision during all
the phases of this work. Without his guidance
and valuable advise, I could not accomplish my research.

I would also wish to express my gratitude to 
Prof.~Tomohiko Uyematsu for valuable advice and support.

It is my pleasure to deeply thank Prof.~Masahito Hayashi
for fruitful discussions and for being a
co-examiner of this thesis.

Constructive comments and suggestions given at 
conferences and seminars have significantly 
improved the presentation of my results.
Especially, I would like to thank
Dr.~Jean-Christian Boileau, Prof.~Akio Fujiwara,
Dr.~Manabu Hagiwara, Dr.~Kentaro Imafuku, Prof.~ Hideki Imai,
Prof.~Mitsugu Iwamoto, Dr.~Yasuhito Kawano,
Dr.~Akisato Kimura, 
Prof.~Shigeaki Kuzuoka,
Prof.~Hoi-Kwong Lo and members of his group,
Dr.~Takayuki Miyadera, Dr.~Jun Muramatsu,
Prof.~Hiroshi Nagaoka, Prof.~Tomohiro Ogawa,
Prof.~Renato Renner and members of his group,
Mr.~Yutaka Shikano, Prof.~Tadashi Wadayama,
and Prof.~Stefan Wolf and members of his group
for valuable comments.

I would also like to thank Prof.~Wakaha Ogata,
Prof.~Kohichi Sakaniwa, and Prof.~Isao Yamada
for valuable advice and for being co-examiner
of this thesis.

My deep thanks are also addressed to all my colleagues at
Uyematsu and Matsumoto laboratory for their support and
helpful comments during regular seminars.
I would also like to thank our group secretaries,
Ms.~Junko Goto and Ms.~Kumiko Iriya for their care
and kindness.

Finally, I would like to deeply thank my parents for
their support and encouragement.

This research was also partially supported by the Japan
Society for the Promotion of Science
under Grant-in-Aid No.~00197137.
\chapter*{Abstract}

Quantum key distribution (QKD) has attracted great attention
as an unconditionally secure key distribution scheme.
The fundamental feature of QKD protocols is that
the amount of information gained by an eavesdropper,
usually referred to as Eve, can be estimated
from the channel between the legitimate sender and the receiver,
usually referred to as Alice and Bob respectively.
Such a task cannot be conducted in classical key
distribution schemes.
If the estimated amount is lower than a threshold,
then Alice and Bob determine the length of a secret key
from the estimated amount of Eve's information,
and can share a secret key by performing the 
postprocessing.
One of the most important criteria for the efficiency
of the QKD protocols is the key generation rate,
which is the length of securely sharable key per channel use.

In this thesis, we investigate the channel estimation procedure
and the postprocessing procedure of the QKD protocols
in order to improve the key generation rates of the
QKD protocols. 
Conventionally in the channel estimation procedure,
we only use the statistics of matched measurement
outcomes, which are bit sequences
transmitted and received by the same basis, to
estimate the channel; mismatched measurement outcomes,
which are bit sequences transmitted and received 
by different bases, are discarded in the conventional
estimation procedure. In this thesis, we propose
a channel estimation procedure in which we use
the mismatched measurement outcomes in addition
to the matched measurement outcomes. Then, we clarify
that the key generation rates of the QKD protocols
with our channel estimation procedure is higher 
than that with the conventional channel estimation procedure.

In the conventional postprocessing procedure, 
which is known as the {\em advantage distillation},
we transmit a message over the public channel redundantly, which is
unnecessary divulging of information to Eve.
In this thesis, we propose a postprocessing in which
the above mentioned divulging of information
is reduced by using the distributed data compression.
We clarify that the key generation rate of 
the QKD protocol with our proposed
postprocessing is higher than that
with the conventionally known postprocessings.

\tableofcontents

\chapter{Introduction}
\label{ch:introduction}
\pagestyle{fancy}
\fancyhead{}
\fancyfoot{}
\fancyhead[LE,RO]{\thepage}
\fancyhead[RE]{\nouppercase{\leftmark}}
\fancyhead[LO]{\nouppercase{\rightmark}}
\setcounter{page}{0}
\pagenumbering{arabic}


\section{Background}
\label{sec:background}

Key distribution is one of the most important and
challenging problem in cryptology.
When a sender wants to transmit a confidential message to
a receiver, the sender usually encipher the message
by using a secret key that is only available to
the sender and the receiver.
For a long time, many methods have been proposed 
to solve the key distribution problem.
One of the most broadly used method in the present day is
a method whose
security is based on difficulties to solve
some mathematical problems, such as
factorization into prime numbers.
Such kind of method is believed to be
practically secure,
but it has not been proved to be unconditionally secure;
there might exist some clever algorithm to
solve those mathematical problems efficiently.
On the other hand, quantum key distribution (QKD),
which is the main theme of this thesis, has attracted
the attention of many researchers,
for the reason that its security is based on principles of
the quantum mechanics.
In other word, the QKD is secure as long as 
the quantum mechanics is correct.

The concept of the quantum cryptography was proposed
by Wiesner in 1970s. Unfortunately, his paper was
rejected by a journal, and was not published until
1983 \cite{wiesner:83}\footnote{For more detailed history on
the quantum cryptography, see Brassard's review article
\cite{brassard:05}.}. 
In 1980s, the quantum cryptography was revived by
Bennett {\em et al}.~in a series of 
papers \cite{bennett:82, bennett:83, bennett:84b, bennett:84}.
Especially, the quantum key distribution first appeared
in Bennett and Brassard's one page proceedings paper \cite{bennett:83}
presented at a conference,
although it is more commonly known as BB84 from its
1984 full publication \cite{bennett:84}. 

At first, the security of the BB84 protocol was
guaranteed only in the ideal situation such that
the channel between the sender and receiver is noiseless.
Later, Bennett {\em et al}.~proposed modified protocols
to handle the case in which the channel between
the sender and the receiver is not necessarily 
noiseless \cite{bennett:89, bennett:92b}.
During the course of their struggle against the problem,
many important concepts such as the information
reconciliation and the privacy amplification, which
are explained in detail later, 
were proposed \cite{bennett:85, bennett:88}.
Finally , Mayers 
proposed his version of the BB84 protocol, and showed 
its unconditional 
security \cite{mayers:01} (preliminary versions of his proof
were published in \cite{mayers:95, mayers:96}).
Biham {\em et al}.~also proposed
their version of the BB84 protocol and showed its unconditional
security \cite{biham:00, biham:06}.

In 2000, Shor an Preskill made a remarkable observation
on Mayer's security proof of the BB84 protocol \cite{shor:00}.
They observed that the entanglement distillation protocol (EDP)
\cite{bennett:96,lo:99} with the CSS code, one of the quantum
error correcting codes proposed by Calderbank, Shor,
and Stean \cite{calderbank:96, steane:96}, is implicitly
used in Mayer's  version of the BB84 protocol, and 
presented a simple
proof of Mayer's version of the BB84 protocol.
Their proof technique based on the CSS code is further
extended to some directions. For example, 
Lo \cite{lo:01} proved the security of another QKD protocol,
the six state protocol proposed by 
Bru\ss~\cite{bruss:98}, by using the technique based on the
CSS code. 

Recently, Renner {\em et al}.~\cite{renner:05, renner:05b, kraus:05}
developed information theoretical techniques to prove the security
of the QKD protocols including the BB84 protocol and the six-state
protocol\footnote{Throughout this thesis, we only treat the
BB84 protocol and the six-state protocol, and we mean these two
protocols by the QKD protocols.}.
Their proof method provides important insight into
the security proof of the QKD protocols. More precisely,
they proved the security of the QKD protocols by extending
the key agreement in the information theory
\cite{maurer:93,ahlswede:93}, which will be explained in the next section, 
to the context of the QKD protocols.

In this thesis, we employ Renner {\em et al}.'s approach
for the security proof of the QKD protocols instead of
Shor and Preskill's approach. 
Then, we investigate two important phases,
the channel estimation and the postprocessing,
of the QKD protocols.

The QKD protocol roughly consists of three phases: the bit
transmission phase, the channel estimation phase, and the
postprocessing phase. In the bit transmission phase, the
legitimate sender, usually referred to as Alice
sends a bit sequence to the legitimate receiver, usually
referred to as Bob, by encoding them into quantum carrier
(eg. polarizations of photons).
The channel estimation phase will be 
explained in Section \ref{sec:the-role}. In the postprocessing phase,
Alice and Bob share a secret key based on their bit
sequences obtained in the bit transmission phase.
The postprocessing phase can be essentially regarded as
the key agreement problem in the information theory,
which will be explained in the next section.

\section{Key Agreement in Information Theory}
\label{sec:key-agreement-in-information-theory} 

Following Shannon's mathematical formulation of the 
cryptography \cite{shannon:49} and the studies on
confidential message transmissions over
noisy channels by Wyner \cite{wyner:75}
and Csis\'{z}ar and K\"orner \cite{csiszar:78},
the problem of the key agreement in the information theory
was formulated by Maurer \cite{maurer:93},
and was also studied 
by Ahlswede and Csis\'{z}ar \cite{ahlswede:93}.

In Maurer's formulation Alice and Bob have
sequences of independently identically
distributed (i.i.d.) correlated binary\footnote{Actually, the formulation
in \cite{maurer:93, ahlswede:93} is not restricted to
binary random variables. However, we restrict our attention
to the binary case because Alice and Bob obtain
binary sequences in the QKD protocols (refer to Section \ref{sec:the-role}).} 
random variables
$\bol{X} = (X_1,\ldots,X_n)$ and 
$\bol{Y} = (Y_1,\ldots,Y_n)$ respectively, and the eavesdropper,
usually referred to as Eve, has a sequence of i.i.d. random
variables $\bol{E} = (E_1,\ldots,E_n)$, which are regarded as
the information she obtained by eavesdropping $\bol{X}$ and $\bol{Y}$.
They conduct a postprocessing\footnote{The postprocessing
is a QKD jargon that means a procedure to distill a secret
key from Alice and Bob's bit sequences.} procedure and share a
secret key  by using the pair of
bit sequence $(\bol{X}, \bol{Y})$ as a seed.

In the postprocessing procedure, Alice and Bob are
allowed to exchange messages over the authenticated
public channel, that is, Eve can know every message
transmitted over this channel but she cannot tamper
or forge a message. Actually, the authenticated
public channel can be realized if Alice and Bob
initially share a short secret key \cite{stinson:91}\footnote{For this reason,
it might be more appropriate to call the procedure
the {\em key expansion} rather than the {\em key agreement}.}.
In the rest of this thesis, we assume that the public channel
is always authenticated though we do not mention it
explicitly. 

The communication over the public channel in
the postprocessing procedure may be one-way 
(from Alice to Bob\footnote{The message transmission
can be from Bob to Alice, which case will be treated in
Chapter \ref{ch:channel-estimation}.}) or
two-way. The most elementary postprocessing
procedure is a procedure with one-way public communication, and it
consists of two procedures, the information reconciliation procedure
and the privacy amplification procedure. 

The purpose of the information reconciliation procedure for Alice
and Bob is to agree on a bit sequence from their
correlated bit sequences. This procedure is nothing but
the Slepian-Wolf coding scheme \cite{slepian:73}\footnote{Actually,
the procedures proposed in \cite{maurer:93, ahlswede:93} do
not use the Slepian-Wolf coding scheme. The Slepian-Wolf
coding scheme in the context of the key agreement
was first used by Muramatsu \cite{muramatsu:06b} explicitly, 
although it was already used in cryptography 
community implicitly (for example in \cite{maurer:00}).}. 
In this scheme, Alice sends the compressed version 
$C$ (say $k$ bit data) of $\bol{X}$ to Bob. Then,
Bob reproduce $\hat{\bol{X}}$ by using his bit sequence $\bol{Y}$
and the received data $C$. It is well known that
Bob can reproduce Alice's bit sequence with negligible error
probability if Alice sends appropriate $k \simeq n H(X|Y)$
bits data.

The purpose of the privacy amplification procedure for Alice
and Bob is to distill secret keys from their
bit sequences shared in the information reconciliation procedure.
More specifically, Alice and Bob distill 
$\ell$ bits (usually much shorter than $n$ bit) 
secret key by using appropriate function from $n$ bit
to $\ell$ bit. 
We require the secret keys to be information
theoretically secure, i.e.,
the distilled key is uniformly
distributed and statistically independent from 
Eve's available information $C$ and $\bol{E}$.

Since the pair of bit sequences initially shared by Alice and Bob
are considered as a precious resource\footnote{Actually, Alice and 
Bob's initial bit sequences are shared by transmitting photons
in the QKD protocols, and the transmission rate of the photon
is usually very slow compared to the transmission rate of
the public channel.}, 
we desire the {\em key generation
rate} $\ell / n$ to be as large as possible.
Especially in this paper, we investigate the 
asymptotic behavior of the key generation rate,
{\em asymptotic key generation rate}, such that
the secure key agreement is possible.
Roughly speaking\footnote{If Alice conducts a
preprocessing before the information reconciliation
procedure, then the condition in
Eq.~(\ref{eq:heuristic-key-generation-rate}) can be 
slightly generalized as 
\begin{eqnarray*}
\frac{\ell}{n} \overset{<}{\sim}
H(U|EV) - H(U|YV),
\end{eqnarray*} where $U$ and $V$ are
auxiliary random variables such that
$V$, $U$, $X$, and $(Y,E)$ form a Markov chain in this order.
Although the meaning of the auxiliary random variables
have been unclear for a long time, recently Renner {\em et al}.~clarified
the meaning of $U$ as the noisy preprocessing 
in the context of QKD protocol \cite{renner:05}
(see also Remark \ref{remark:noisy-preprocessing}).}, 
the secure key can be distilled
if the key generation rate is smaller than Eve's
ambiguity (per bit) about the bit sequence after the
information reconciliation, that is,
\begin{eqnarray}
\frac{\ell}{n} \overset{<}{\sim}  
 H(X|E) - H(X|Y).
\label{eq:heuristic-key-generation-rate}
\end{eqnarray}

In \cite{maurer:93}, Maurer also proposed a
postprocessing procedure with two-way public communication.
More specifically, he proposed a preprocessing called
{\em advantage distillation} that is conducted
before the information reconciliation procedure. 
In the advantage distillation,
Alice divides her bit sequence into blocks
of length $2$, and sends the parity $X_{2i-1} \oplus X_{2i}$
of each block to Bob. Bob also divides his bit sequence
into blocks of length $2$, and tells Alice whether the received
parity of the $i$th block coincides with Bob's corresponding parity 
$Y_{2i-1} \oplus Y_{2i}$. If their corresponding parities coincide,
they keep the second bits of those blocks, which are 
regarded to have strong correlation. Otherwise, they
discard those blocks, which are regarded to have
weak correlation. Maurer showed that the key generation rate
of the postprocessing with the advantage distillation
can be strictly higher than the right hand side of 
Eq.~(\ref{eq:heuristic-key-generation-rate}) in an example.

In the context of the QKD protocol, the postprocessing
procedure with both one-way and two-way public communication
were considered. Actually, the postprocessing procedure
with one-way public communication were first studied
\cite{mayers:01, shor:00}. Later, the postprocessing
with the advantage distillation in the 
context of QKD protocol was proposed by 
Gottesman and Lo \cite{gottesman:03}.
The postprocessing with the advantage distillation was
extensively studied by Bae and Ac\'in \cite{bae:07}.

In Chapter \ref{ch:postprocessing}, 
we propose a new kind of postprocessing procedure
with two-way public communication in the context
of QKD protocol. 
The purpose of the advantage distillation
was to divide the blocks into highly correlated ones
and weakly correlated ones by exchanging the parities.
The key idea of our proposed postprocessing is that 
the parities in the conventional advantage distillation is redundantly
transmitted over the public channel, and should be compressed
by the Slepian-Wolf coding because Bob's bits 
$(Y_{2i-1}, Y_{2i})$ is correlated to Alice's parity 
$X_{2i-1} \oplus X_{2i}$.
In our proposed postprocessing, Alice does not
sends the parities itself, 
but she sends the compressed version of the parities
by regarding Bob's sequence $\bol{Y}$ as the side-information at the decoder.
This enables Alice and Bob to extract a secret key
also from the parity sequence, and improves the key
generation rate. Actually, the key generation rate of the
QKD protocols with our proposed postprocessing procedure 
is as high as that with conventional one-way
or two-way postprocessing procedures. We also 
clarify that the former is strictly higher than the latter in
some cases.

\section{Unique Property of Quantum Key Distribution}
\label{sec:the-role}

In the previous section, we have explained 
the mathematical formulation of the key
agreement in the information theory. 
Then, we have explained the fact that Alice and
Bob have to set the key generation rate according
Eve's ambiguity about the bit sequence after the information
reconciliation procedure 
(Eq.~(\ref{eq:heuristic-key-generation-rate}))\footnote{When
Alice and Bob conduct the postprocessing with two-way
public communication, they have to set the key generation
rate according to more complicated formula (for more detail,
see Chapter \ref{ch:postprocessing}).}
in order to share an information theoretically secure key.
However, Alice and Bob cannot calculate the amount of Eve's
ambiguity about the bit sequence if they do not know 
the probability distribution $P_{XYE}$ of their initial
bit sequence and Eve's available information. 
Therefore, they have to estimate the probability distribution itself, 
or at least they have to estimate a lower bound on the
quantity $H(X|E)$\footnote{Since the quantity $H(X|Y)$
only involves the marginal distribution $P_{XY}$, 
Alice and Bob can easily estimate it by sacrificing
a part of their bit sequence as samples. Therefore, we restrict
our attention to the quantity $H(X|E)$.}.  
If Alice and Bob's bit sequences $(\bol{X},\bol{Y})$
are distributed by using a classical channel,
for example the standard telephone line or the Internet, 
then a valid estimate will be the trivial one, $0$, because
Eve can eavesdrop as much as she want without being
detected. The QKD protocols provide a way to estimate
a non-trivial lower bound on $H(X|E)$ by using the axioms of the
quantum mechanics.

In the BB84 protocol, Alice randomly chooses a bit sequence
and send it by encoding each bit into a polarization of a photon.
When she encodes each bit into a polarization of a photon,
she chooses one of two encoding rules at random.
In the first encoding rule, she encodes $0$ into
the vertical polarization, and $1$ into the
horizontal polarization. In the second encoding rule,
she encodes $0$ into the $45$ degree polarization,
and $1$ into the $135$ degree polarization.

On the other hand, Bob measures the received photons
by using one of two measurement device at random.
The first measurement device discriminate between
the vertical and the horizontal polarizations,
and the measurement outcome is decoded into 
the corresponding bit value. The second measurement
device discriminate between the $45$ degree
and the $135$ degree polarizations, and 
the measurement outcome is decoded into the
corresponding bit value. 

After the reception of the photons,
Alice and Bob announce over the public channel
which encoding rule and which measurement device they
have used for each bit.
Then, they keep those bits if their encoding rule
and measurement device are compatible, i.e.,
Alice uses the first (the second) encoding
rule and Bob uses the first (the second)
measurement device. We call such bit sequences
the {\em matched measurement outcomes}. 
On the other hand, they discard those bits
if their encoding rule and measurement
device are incompatible, i.e.,
Alice uses the first (the second) encoding
rule and Bob uses the second (the first)
measurement device. We call such bit sequences
the {\em mismatched measurement outcomes}.
Furthermore, Alice and Bob announce a part of
their matched measurement outcomes to estimate candidates of
the quantum channel over which the photons were transmitted.
The rest of the matched measurement outcomes
are used as a seed for sharing a secret key.

The most important feature of the QKD protocols 
is that we can calculate the quantity
$H(X|E)$\footnote{It should be noted that we have
to use the conditional von Neumann entropy instead
of the conditional Shannon entropy in the case
of the QKD protocols (for more detail, 
see Chapter \ref{ch:channel-estimation}).} 
by using the axioms of the quantum mechanics
if they know the quantum channel exactly. 
Therefore, we can estimate a lower bound on
$H(X|E)$ via estimating the candidates of the
quantum channel. Actually, we employ the 
quantity $H(X|E)$ minimized over the estimated candidates
of the quantum channel as an estimate of true $H(X|E)$.

As we explained above, in the conventional BB84 protocol
we discard the mismatched measurement outcomes and we
estimate the candidates of the quantum channel
by using only the samples from the matched measurement
outcomes. 
In Chapter \ref{ch:channel-estimation},
we propose a channel estimation procedure in which
we use the mismatched measurement outcomes in addition
to the samples from the matched measurement outcomes.
The use of the mismatched measurement outcomes enables us to
reduce candidates of the quantum channel, and then enables
us to estimate tighter lower bounds on the quantity $H(X|E)$. 
Actually, we clarify that the key generation rate decided according
to our proposed channel estimation procedure is at least as
high as the key generation rate decided according to
the conventional channel estimation procedure.
We also clarify that the former is strictly higher than
the latter in some cases. In Chapter \ref{ch:postprocessing},
we also apply our proposed channel estimation procedure
to the protocol with the two-way postprocessing proposed
in Chapter \ref{ch:postprocessing}. 

It should be noted that the use of the mismatched 
measurement outcomes was already considered in 
literatures. In early 90s, Barnett et al.~\cite{barnett:93} 
showed that the use of mismatched measurement
outcomes enables Alice and Bob to detect the presence of Eve with 
higher probability for the so-called intercept and resend attack.
Furthermore, some literatures use the mismatched measurement outcomes
to ensure the quantum channel to be a Pauli channel
\cite{bruss:03,liang:03, kaszlikowski:05, kaszlikowski:05b},
where a Pauli channel is a channel over which four kinds of 
Pauli errors (including the identity) occur probabilistically.
However the quantum channel is not necessarily a Pauli channel
in general. One of the aims of this thesis is to convince
the readers that the non-Pauli channels deserve consideration
in the research of the QKD protocols
as well as the Pauli channel.

\section{Summary}
\label{sec:summary}

The QKD protocols consists of three phases: the bit
transmission phase, the channel estimation phase, and the
postprocessing phase.
The role of the channel estimation phase is to
estimate the amount of Eve's ambiguity about
the bit sequence transmitted in the bit transmission
phase. According to the estimated amount of
Eve's ambiguity, we decide the key generation rate
and conduct the postprocessing to share a secret key.

In the conventional estimation procedure, 
we do not use the mismatched measurement outcomes.
By using the mismatched measurement outcomes in addition
to the samples from the matched measurement outcomes, we 
can improve the key generation rate of the QKD protocols. 
This topic is investigated 
in Chapter \ref{ch:channel-estimation}.

In the conventional (two-way)
postprocessing procedure, we transmit a 
message over the public channel redundantly,
which is unnecessary divulging of  information to Eve. 
By transmitting the compressed version of the 
redundantly transmitted message, we can improve
the key generation rate of the QKD protocols.
This topic is investigated 
in Chapter \ref{ch:postprocessing}.



\chapter{Preliminaries}
\label{ch:preliminaries}

In this chapter, we introduce some terminologies and  notations,
and give a brief review of the known results that are used
throughout this thesis.
The first section is devoted to a review of the 
classical information theory \cite{cover} and
the quantum information theory \cite{nielsen-chuang:00,
hayashi-book:06}. 
In the second section, we review the known results on
the privacy amplification, which is the most important 
tool for the security of the QKD protocols.

\section{Elements of Classical and Quantum Information Theory}
\label{sec:basic}

\subsection{Probability Distribution and Density Operator}
\label{subsec:probability-distribution}

\index{probability}
For a finite set ${\cal X}$, let ${\cal P}({\cal X})$ be 
the set of all probability distributions $P$ on ${\cal X}$, i.e.,
$P(x) \ge 0$ for all $x \in {\cal X}$ and 
$\sum_{x \in {\cal X}} P(x) = 1$.
For a sequence $\bol{x} = (x_1, \ldots, x_n) \in {\cal X}^n$,
the type of $\bol{x}$ is the 
empirical probability distribution $P_{\bol{x}} \in {\cal P}({\cal X})$ defined by
\begin{eqnarray*}
P_{\bol{x}}(a) := \frac{ | \{ i \mid x_i = a \} | }{n}~~~~~~\mbox{for }
a \in {\cal X},
\end{eqnarray*}
where $|A|$ is the cardinality of a set $A$.

For a finite-dimensional Hilbert space ${\cal H}$, let 
${\cal P}({\cal H})$ be the set of all density  operators
$\rho$ on ${\cal H}$, i.e., $\rho$ is non-negative and
normalized, $\rom{Tr} \rho = 1$.
Mathematically, a state of a quantum mechanical system with
$d$-degree of freedom is represented by a density operator on ${\cal H}$
with $\dim {\cal H} = d$.
Throughout the thesis, we occasionally call $\rho$ a state
and ${\cal H}$ a system.
For Hilbert spaces ${\cal H}_A$ and ${\cal H}_B$, the set of all density 
operators ${\cal P}({\cal H}_A \otimes {\cal H}_B)$ on the tensor product space
${\cal H}_A \otimes {\cal H}_B$ is defined in a similar manner.
In Section \ref{sec:privacy-amplification}, we occasionally treat 
non-normalized non-negative operators. For this reason,
we denote the set of all non-negative operators
on a system ${\cal H}$ (and a composite system ${\cal H}_A \otimes {\cal
H}_B$)
by ${\cal P}^\prime({\cal H})$ (and 
${\cal P}^\prime({\cal H}_A \otimes {\cal H}_B)$).

The classical random variables can be regarded as a special case of
the quantum states. For a random variable $X$ with a distribution
$P_X \in {\cal P}({\cal X})$, let 
\begin{eqnarray*}
\rho_X := \sum_{x \in {\cal X}} P_X(x) \ket{x}\bra{x},
\end{eqnarray*}
where $\{ \ket{x} \}_{x \in {\cal X}}$ is an orthonormal basis of ${\cal H}_X$.
We call $\rho_X$ the operator representation of the classical
distribution $P_X$.

When a quantum system ${\cal H}_A$ is prepared in a state $\rho^x_A$ according
to a realization $x$ of a random variable $X$ with a probability distribution $P_X$, it is
convenient to describe this situation
by a density operator
\begin{eqnarray}
	\label{eq-definition-cq-state}
\rho_{XA} := \sum_{x \in {\cal X}} P_X(x) \ket{x}\bra{x} \otimes
 \rho_A^x \in {\cal P}({\cal H}_X \otimes {\cal H}_A),
\end{eqnarray}
where $\{ \ket{x} \}_{x \in {\cal X}}$ is an orthonormal basis of 
${\cal H}_X$.
We call the density operator $\rho_{XA}$ a $\{cq\}$-state \cite{devetak:04},
or we say $\rho_{XA}$ is classical on ${\cal H}_X$
with respect to the orthonormal basis $\{ \ket{x} \}_{x \in {\cal X}}$.
We call $\rho_A^x$ a conditional operator. 
When a quantum system ${\cal H}_A$ is prepared in a state $\rho_A^{x,y}$
according to a joint random variable $(X, Y)$ with a probability
distribution $P_{XY}$, a state $\rho_{XYA}$ is defined in a similar
manner, and the state $\rho_{XYA}$ is called a $\{ ccq \}$-state. 
For non-normalized operator 
$\rho_{XA} \in {\cal P}^\prime({\cal H}_X \otimes {\cal H}_A)$,
if we can write $\rho_{XA}$ as in Eq.~(\ref{eq-definition-cq-state}),
we say that $\rho_{XA}$ is classical on ${\cal H}_{X}$
with respect to the orthonormal basis $\{ \ket{x} \}_{x \in {\cal X}}$.
However, it should be noted that the distribution $P_X$ or
conditional operators $\rho_A^x$ are not necessarily normalized
for a non-normalized $\rho_{XA}$.

For a $\{cq\}$-state $\rho_{XA} \in {\cal P}({\cal H}_X \otimes {\cal
H}_A)$, we occasionally consider a density operator such 
that the classical system ${\cal H}_X$ is mapped by
a function $f:{\cal X} \to {\cal Y}$. By setting the
distribution
\begin{eqnarray*}
P_Y(y) = \sum_{x \in {\cal X} \atop f(x) = y} P_X(x)
\end{eqnarray*}
and the density operator
\begin{eqnarray*}
\rho_A^y = \sum_{x \in {\cal X} \atop f(x) = y} P_X(x) \rho_A^x /P_Y(y),
\end{eqnarray*}
we can describe the resulting $\{cq\}$-state as
\begin{eqnarray}
\label{eq:cq-state-mapping}
\rho_{YE} := \sum_{y \in {\cal Y}} P_Y(y) \ket{y}\bra{y}
  \otimes \rho_A^y. 
\end{eqnarray}

In the quantum mechanics, the most general measurement is
described by the positive operator valued measure
(POVM). A POVM for a system ${\cal H}$
consists of the set ${\cal A}$ of
measurement outcomes, and the set ${\cal M} = \{M_a \}_{a \in {\cal A}}$
of positive operators indexed by the set ${\cal A}$.
For a state $\rho \in {\cal P}({\cal H})$, the probability
distribution of the measurement outcomes is given by
\begin{eqnarray*}
P(a) = \rom{Tr}[ \rho M_a].
\end{eqnarray*}

In the quantum mechanics, the most general state evolution of a
quantum mechanical system is described by a
completely positive (CP) map. It can be shown that
any CP map ${\cal E}$ can be written as
\begin{eqnarray}
	\label{eq-kraus-operator}
{\cal E}(\rho) = \sum_{a \in {\cal A}} E_a \rho E_a^*
\end{eqnarray}
for a family of linear operators
$\{ E_a \}_{a \in {\cal A}}$ from
the initial system ${\cal H}$ to the destination system 
${\cal H}^\prime$, where ${\cal A}$ is the index set.
We usually require the map to be trace preserving (TP),
i.e., $\sum_{a \in {\cal A}} E_a^* E_a = \rom{id}_{\cal H}$,
but if a state evolution involves a selection of states by a measurement,
then the corresponding CP map is not necessarily trace preserving, i.e.,
$\sum_{a \in {\cal A}} E_a^* E_a \le \rom{id}_{\cal H}$.

\subsection{Distance and Fidelity}
\label{subsec:distance}

In this thesis, we use two kinds of distances.
One is the variational distance of ${\cal P}({\cal X})$.
For non-negative functions 
$P, P^\prime \in {\cal P}({\cal X})$, the variational distance between
$P$ and $P^\prime$ is defined by
\begin{eqnarray*}
\| P - P^\prime \| := \sum_{x \in {\cal X}} | P(x) - P^\prime(x) |.
\end{eqnarray*}
The other distance used in this paper is 
the trace distance of ${\cal P}^\prime({\cal H})$.
For non-negative operators $\rho, \sigma \in {\cal P}^\prime({\cal H})$, the trace distance
between $\rho$ and $\sigma$ is defined by
\begin{eqnarray*}
\| \rho - \sigma \| := \rom{Tr} | \rho - \sigma |,
\end{eqnarray*}
where $|A| := \sqrt{A^* A}$ for a operator on ${\cal H}$, and 
$A^*$ is the adjoint operator of $A$.
The following lemma states that the trace distance 
between (not necessarily normalized operators) does not 
increase by applying a CP map, and it is used several times 
in this paper.
\begin{lemma}
\label{lemma:monotonicity-of-trace-distance}
\cite[Lemma A.2.1]{renner:05b}
Let $\rho, \rho^\prime \in {\cal P}^\prime({\cal H})$ and let 
${\cal E}$ be a trace-non-increasing CP map, i.e., ${\cal E}$ satisfies
$\rom{Tr}{\cal E}(\sigma) \le \rom{Tr} \sigma$ for any
$\sigma \in {\cal P}^\prime({\cal H})$. Then we have
\begin{eqnarray*}
\| {\cal E}(\rho) - {\cal E}(\rho^\prime) \| \le \| \rho - \rho^\prime \|.
\end{eqnarray*}
\end{lemma}

The following lemma states that, for a $\{cq\}$-state 
$\rho_{XB}$, if two classical messages $v$ and $\bar{v}$ are
computed from $x$ and they are equal with high probability,
then the $\{ccq\}$ state $\rho_{XVB}$ and $\rho_{X\bar{V}B}$ that
involve computed classical messages $v$ and $\bar{v}$ 
are close with respect to the trace distance.
\begin{lemma}
	\label{lemma-error-prob-continuity}
Let 
\begin{eqnarray*}
\rho_{XB} := \sum_{x \in {\cal X}} P_X(x) \ket{x}\bra{x} \otimes \rho_B^x
\end{eqnarray*}
be a $\{cq\}$-state, and let $V := f(X)$ for a function $f$ and
$\bar{V} := g(X)$ for a function $g$. Assume that 
\begin{eqnarray*}
\Pr\{ V \neq \overline{V} \} =
\sum_{\scriptstyle x \in {\cal X}  \atop
f(x) \neq g(x)} P_X(x) 
\le \varepsilon.
\end{eqnarray*}
Then, for $\{ ccq\}$-states
\begin{eqnarray*}
\rho_{XVB} := \sum_{x \in {\cal X} }
P_X(x) \ket{x}\bra{x}  \otimes \ket{f(x)}\bra{f(x)}
\otimes \rho_B^x
\end{eqnarray*}
and 
\begin{eqnarray*}
\rho_{X\overline{V}B} := \sum_{x \in {\cal X} }
P_X(x) \ket{x}\bra{x}  \otimes \ket{g(x)}\bra{g(x)}
\otimes \rho_B^x,
\end{eqnarray*}
we have
\begin{eqnarray*}
\| \rho_{XVB} - \rho_{X\overline{V}B} \| \le 2 \varepsilon.
\end{eqnarray*}
\end{lemma}
\begin{proof}
We have
\begin{eqnarray*}
\lefteqn{ \| \rho_{XVB} - \rho_{X\overline{V}B} \| } \\
&=& \sum_{x \in {\cal X} } P_X(x)
\| \ket{x}\bra{x} \| 
\cdot \| \ket{f(x)}\bra{f(x)} - \ket{g(x)}\bra{g(x)} \|
\cdot \| \rho_B^x \| \\
&=& \sum_{x \in {\cal X} } P_X(x)
\cdot 2(1- \delta_{f(x), g(x)}) \\
&\le& 2 \varepsilon,
\end{eqnarray*}
where $\delta_{a,b} = 1$ if $a=b$ and $\delta_{a,b} = 0$ if $a \neq b$. 
\end{proof}

The fidelity between two (not necessarily normalized) operators 
$\rho, \sigma \in {\cal P}^\prime({\cal H})$
is defined by
\begin{eqnarray*}
F(\rho, \sigma) := \rom{Tr} \sqrt{\sqrt{\rho}\sigma \sqrt{\rho}}.
\end{eqnarray*}
The following lemma is an extension of Uhlmann's theorem to 
non-normalized operators $\rho$ and $\sigma$.
\begin{lemma}
\cite[Theorem A.1.2]{renner:05b}
	\label{uhlman-theorem}
Let $\rho, \sigma \in {\cal P}^\prime({\cal H})$,
and let $\ket{\psi} \in {\cal H}_R \otimes {\cal H}$ be a purification of $\rho$. Then
\begin{eqnarray*}
F(\rho, \sigma) = \max_{\ket{\phi}\bra{\phi}} F(\ket{\psi}\bra{\psi}, \ket{\phi}\bra{\phi}),
\end{eqnarray*}
where the maximum is taken over all purifications 
$\ket{\phi} \in {\cal H}_R \otimes {\cal H}$ of $\sigma$.
\end{lemma}
The trace distance and the fidelity have close relationship.
If the trace distance
between two non-negative operators $\rho$ and $\sigma$ is close to $0$,
then the fidelity between $\rho$ and $\sigma$ is close to $1$,
and vise versa.
\begin{lemma}
\cite[Lemma A.2.4]{renner:05b}
	\label{upper-bound-by-fidelity}
Let $\rho, \sigma \in {\cal P}^\prime({\cal H})$. Then, we have
\begin{eqnarray*}
\| \rho - \sigma \| \le 
\sqrt{(\rom{Tr}\rho + \rom{Tr} \sigma)^2 - 4 F(\rho, \sigma)^2 }.
\end{eqnarray*}
\end{lemma}
\begin{lemma}
\cite[Lemma A.2.6]{renner:05b}
	\label{lower-bound-by-fidelity}
Let $\rho, \sigma \in {\cal P}^\prime({\cal H})$. Then, we have
\begin{eqnarray*}
\rom{Tr}\rho + \rom{Tr} \sigma - 2 F(\rho, \sigma) 
\le \| \rho - \sigma \|.
\end{eqnarray*}
\end{lemma}

\subsection{Entropy and its Related Quantities}
\label{subsec:entropic-measures}

For a random variable $X$ on ${\cal X}$ with
a probability distribution $P_X \in {\cal P}({\cal X})$, the entropy
of $X$ is defined by
\begin{eqnarray*}
H(X) = H(P_X) := - \sum_{x \in {\cal X}} P_X(x) \log P_X(x),
\end{eqnarray*}
where we assume the base of $\log$ is $2$
throughout the thesis.
Especially for a real number $0 \le p \le 1$,
the binary entropy function is defined by 
\begin{eqnarray*}
h(p) := - p \log p - (1-p) \log (1-p).
\end{eqnarray*}
Similarly, for a joint random variables $X$ and $Y$ with
a joint  probability distribution $P_{XY} \in {\cal P}({\cal X} \times {\cal Y})$,
the joint entropy of $X$ and $Y$ is
\begin{eqnarray*}
H(XY) &=& H(P_{XY}) \\
&:=& - \sum_{(x,y) \in {\cal X} \times {\cal Y}}
P_{XY}(x,y) \log P_{XY}(x,y).
\end{eqnarray*}
The conditional entropy of $X$ given $Y$ is defined by
\begin{eqnarray*}
H(X|Y) := H(XY) - H(Y).
\end{eqnarray*}
The mutual information between the joint random variables
$X$ and $Y$ is defined by
\begin{eqnarray*}
I(X;Y) := H(X) + H(Y) - H(XY).
\end{eqnarray*}

For a quantum state $\rho \in {\cal P}({\cal H})$, the von Neumann entropy
of the system is defined by
\begin{eqnarray*}
H(\rho) := - \rom{Tr} \rho \log \rho.
\end{eqnarray*}
For a quantum state $\rho_{AB} \in {\cal P}({\cal H}_A \otimes {\cal H}_B)$
of the composite system, the von Neumann entropy of the composite system
is $H(\rho_{AB})$.
The conditional von Neaumann entropy of the system $A$ given the system $B$
is defined by
\begin{eqnarray*}
H_{\rho}(A|B) := H(\rho_{AB}) - H(\rho_B),
\end{eqnarray*}
where $\rho_B = \rom{Tr}_A [ \rho_{AB} ]$ is the partial trace of
$\rho_{AB}$ over the system $A$.
The quantum mutual information between the system $A$ and $B$
is defined by
\begin{eqnarray*}
I_{\rho}(A;B) := H(\rho_A) + H(\rho_B) - H(\rho_{AB}).
\end{eqnarray*}
It should be noted that, for $\{cq\}$-state $\rho_{XA}$, 
the quantum mutual information coincides with the
Holevo information, i.e.,
\begin{eqnarray*}
I_{\rho}(X;A) = H(\rho_A) - \sum_{x \in {\cal X}} P_X(x) H(\rho_A^x).
\end{eqnarray*}

\begin{remark}
    \label{remark-convention-2}
In this paper, we denote $\rho_A$ for  $\rom{Tr}_{B}[ \rho_{AB}]$ or
$\rho_{B}$ for $\rom{Tr}_{AC}[ \rho_{ABC}]$ e.t.c. without
declaring them if they are obvious from the context. 
\end{remark}

\subsection{Bloch Sphere, Choi Operator, and Stokes Parameterization}
\label{subsec:bloch-stokes}

In this section, we first introduce the Bloch sphere, which is a 
parameterization of the set ${\cal P}({\cal H})$ of density operators
on two-dimensional space (qubit). Then, we 
introduce the Choi operator for the qubit channel and
its Stokes parameterization.

Let 
\begin{eqnarray*}
\sigma_\san{x} := \left[\begin{array}{cc}
  0 & 1 \\ 1 & 0 \end{array} \right],~~~ 
\sigma_\san{y} := \left[\begin{array}{cc}
  0 & - \bol{i} \\ \bol{i} & 0 \end{array} \right],~~~
\sigma_\san{z} := \left[\begin{array}{cc}
 1 & 0 \\ 0 & -1 \end{array}\right] 
\end{eqnarray*} 
be the Pauli operators, and let $\sigma_\san{i} = I$
be the identity operator on the qubit. 
Then, the set $\{ \sigma_\san{i}, \sigma_\san{x}, \sigma_\san{y},
\sigma_\san{z}\}$
form a basis of the set ${\cal L}({\cal H})$ of all operators
on ${\cal H}$. Furthermore, we have
\begin{eqnarray}
\label{eq:bloch-sphere}
{\cal P}({\cal H}) = \left\{
\frac{1}{2} \left[ \begin{array}{cc}
1 + \theta_\san{z} & \theta_\san{x} - \bol{i} \theta_\san{y} \\
\theta_\san{x} + \bol{i} \theta_\san{y} & 1 - \theta_\san{z}
\end{array}\right]
\mymid  \theta_\san{x}^2 + \theta_\san{y}^2 + \theta_\san{z}^2 \le 1
\right\},
\end{eqnarray}
that is, there is one-to-one correspondence between a qubit density
operator and a (column) vector\footnote{For a reason 
clarified in Section \ref{sec:example},
we denote the coordinate in this order.} 
$\theta = [\theta_\san{z}, \theta_\san{x}, \theta_\san{y}]^T$
within the unit sphere, which is called the 
Bloch sphere \cite{nielsen-chuang:00}.
By a straightforward calculation, 
we can find that the von Neumann entropy of
the density operator $\rho$ that corresponds to
the vector $\theta = [\theta_\san{z},\theta_\san{x},\theta_\san{y}]^T$
is 
\begin{eqnarray}
\label{eq:entropy-of-bloch}
H(\rho) = h\left( \frac{1 + \| \theta \| }{2} \right),
\end{eqnarray}
where $\| \theta\|$ is the Euclidian norm of the vector $\theta$.

Let ${\cal W}({\cal H}_A, {\cal H}_B)$ be the set of all
TPCP maps (see Section \ref{subsec:probability-distribution})
from ${\cal P}({\cal H}_A)$ 
to ${\cal P}({\cal H}_B)$, where we set ${\cal H}_A = {\cal H}_B$ 
as qubit. Let
\begin{eqnarray}
\label{eq:definition-of-maximally-entangled}
\ket{\psi} := \frac{\ket{00} + \ket{11}}{\sqrt{2}}
\end{eqnarray}
be a maximally entangled state on the composite system
${\cal H}_A \otimes {\cal H}_B$. 
Then, we define the set 
${\cal P}_c \subset {\cal P}({\cal H}_A
\otimes {\cal H}_B)$ such as any element 
$\rho \in {\cal P}_c$ 
satisfies $\rom{Tr}_B[\rho] = I/2$.
It is well known that \cite{choi:75, fujiwara:99}
there is one-to-one correspondence between the set 
${\cal W}({\cal H}_A, {\cal H}_B)$  
and the set ${\cal P}_c$ via
the map
\begin{eqnarray*}
{\cal W}({\cal H}_A, {\cal H}_B) \ni {\cal E} \mapsto
  \rho_{AB} := (\rom{id} \otimes {\cal E})(\psi) \in 
  {\cal P}_c.
\end{eqnarray*}
The operator $\rho_{AB}$ is also known as the (normalized) Choi
operator \cite{choi:75}.

For a Choi operator $\rho_{AB} \in {\cal P}_c$,
let
\begin{eqnarray}
\label{eq:def-stokes-R}
R_{\san{ba}} &:=& \rom{Tr}[ \rho_{AB} (\bar{\sigma}_\san{a} \otimes
 \sigma_\san{b})] 
\end{eqnarray}
and 
\begin{eqnarray}
\label{eq:def-stokes-t}
t_{\san{b}} &:=& \rom{Tr}[ \rho_{AB} (I \otimes \sigma_{\san{b}})]
\end{eqnarray}
for $\san{a},\san{b} \in \{\san{z},\san{x},\san{y} \}$,
where $\bar{\sigma}_{\san{a}}$ is the complex 
conjugate of $\sigma_{\san{a}}$.
The pair 
\begin{eqnarray*}
(R, t) := \left(
\left[\begin{array}{ccc}
R_{\san{zz}} & R_{\san{zx}} & R_{\san{zy}} \\
R_{\san{xz}} & R_{\san{xx}} & R_{\san{xy}} \\
R_{\san{yz}} & R_{\san{yx}} & R_{\san{yy}}
\end{array}\right],
\left[\begin{array}{c}
t_{\san{z}} \\ t_{\san{x}} \\ t_{\san{y}}
\end{array}\right] 
\right) 
\end{eqnarray*}
of the matrix and the vector
is called the Stokes parameterization of
the channel ${\cal E}$ and the Choi operator $\rho_{AB}$
\cite{fujiwara:98,fujiwara:99}. 
By a straightforward calculation,
we can find that the channel ${\cal E}$ is equivalent to 
the affine map
\begin{eqnarray*}
\left[ \begin{array}{c}
\theta_{\san{z}} \\ \theta_{\san{x}} \\ \theta_{\san{y}}
\end{array} \right] 
\mapsto
\left[ \begin{array}{ccc}
R_{\san{zz}} & R_{\san{zx}} & R_{\san{zy}} \\
R_{\san{xz}} & R_{\san{xx}} & R_{\san{xy}} \\
R_{\san{yz}} & R_{\san{yx}} & R_{\san{yy}}
\end{array} \right]
\left[ \begin{array}{c}
\theta_{\san{z}} \\ \theta_{\san{x}} \\ \theta_{\san{y}}
\end{array} \right]
+ \left[ \begin{array}{c}
t_{\san{z}} \\ t_{\san{x}} \\ t_{\san{y}} 
\end{array} \right]
\end{eqnarray*}
from the Bloch sphere to itself.

In the rest of this thesis, we identify a Choi operator
and its Stokes parameterization if it is obvious from the
context. For example, $(R,t) \in {\cal A} \subset {\cal P}_c$ means that the Choi operator $\rho_{AB}$ 
corresponding to $(R,t)$ is included in the subset ${\cal A}$.

\section{Privacy Amplification} 
\label{sec:privacy-amplification} 

In this section, we review the privacy amplification. 
First, we review notions of the (smooth) min-entropy and 
the (smooth) max-entropy. The (smooth) min-entropy and the (smooth)
max-entropy are 
useful tool to prove the security of QKD protocols 
\cite{kraus:05, renner:05, renner:05b}. 
Especially, (smooth) min-entropy is much more important,
because it is related to the length of the securely distillable 
key by the privacy amplification.
The privacy amplification \cite{bennett:85,bennett:88, bennett:95} is a technique 
to  distill a secret key from partially secret data, on 
which an adversary might have some information. 
Later, the privacy amplification was extended to the case that 
an adversary have information encoded into a state of a quantum system
\cite{christandl:04, konig:05, renner:05c, renner:05b}.
Most of the following results can be found in \cite[Sections 3 and
5]{renner:05b},
but lemmas without citations are additionally proved 
in the appendix of \cite{watanabe:07}. 
We need Lemma \ref{lemma-weak-monotonicity}
to apply the results in \cite{renner:05b} 
to the QKD protocols with two-way postprocessing
in Chapter \ref{ch:postprocessing}. 
More specifically,
Eq.~(3.22) in \cite[Theorem 3.2.12]{renner:05b} plays an important
role to show a statement similar as Corollary
\ref{lemma-weak-chain-rule} in the case of 
the QKD protocols with one-way postprocessing.
However, the condition of Eq.~(3.22) in \cite[Theorem
3.2.12]{renner:05b}
is too restricted, and cannot be applied 
to the case of the two-way postprocessing
proposed in Chapter \ref{ch:postprocessing}.
Thus, we show Corollary \ref{lemma-weak-chain-rule}
via Lemma \ref{lemma-weak-monotonicity}.
Lemmas \ref{neighbor-purification} and 
\ref{neighbor-purification-classical} are needed
 to prove Lemma \ref{lemma-weak-monotonicity}.

\subsection{Min- and Max- Entropy}

The (smooth) min-entropy and (smooth) max-entropy are formally defined as follows.
\begin{definition}
\cite[Definition 3.1.1]{renner:05b}
Let $\rho_{AB} \in {\cal P}^\prime({\cal H}_A \otimes {\cal H}_B)$ and
$\sigma_B \in {\cal P}({\cal H}_B)$. The min-entropy of 
$\rho_{AB}$ relative to $\sigma_B$ is defined by
\begin{eqnarray*}
H_{\min}(\rho_{AB}|\sigma_B) := - \log \lambda,
\end{eqnarray*}
where $\lambda$ is the minimum real number such that 
$\lambda \cdot \rom{id}_A \otimes \sigma_B - \rho_{AB} \ge 0$,
where $\rom{id}_A$ is the identity operator on ${\cal H}_A$.
When the condition $\rom{supp}(\rho_B) \subset \rom{supp}(\sigma_B)$ does not hold,
there is no $\lambda$ satisfying the condition 
$\lambda \cdot \rom{id}_A \otimes \sigma_B - \rho_{AB} \ge 0$, thus we define
$H_{\min}(\rho_{AB}|\sigma_B) := - \infty$.

The max-entropy of $\rho_{AB}$ relative to $\sigma_B$ is defined by
\begin{eqnarray*}
H_{\max}(\rho_{AB}|\sigma_B) := \log \rom{Tr} \left( (\rom{id}_A \otimes \sigma_B) 
\rho_{AB}^0 \right),
\end{eqnarray*}
where $\rho_{AB}^0$ denotes the projector onto the support of $\rho_{AB}$.

The min-entropy and the max-entropy of $\rho_{AB}$ given ${\cal H}_B$ are defined by
\begin{eqnarray*}
H_{\min}(\rho_{AB}|B) &:=& \sup_{\sigma_B} H_{\min}(\rho_{AB}|\sigma_B) \\
H_{\max}(\rho_{AB}|B) &:=& \sup_{\sigma_B} H_{\max}(\rho_{AB}|\sigma_B),
\end{eqnarray*}
where the supremum ranges over all $\sigma_B \in {\cal P}({\cal H}_B)$.
\end{definition}

When ${\cal H}_B$ is the trivial space $\mathbb{C}$, the min-entropy and
the max-entropy of $\rho_A$ is 
\begin{eqnarray*}
H_{\min}(\rho_A) &=& - \log \lambda_{\max}( \rho_A) \\
H_{\max}(\rho_A) &=& \log \rom{rank}(\rho_A),
\end{eqnarray*}
where $\lambda_{\max}(\cdot)$ denotes the maximum eigenvalue of the
argument.

\begin{definition}
\cite[Definitions 3.2.1 and 3.2.2]{renner:05b}
Let $\rho_{AB} \in {\cal P}^\prime({\cal H}_A \otimes {\cal H}_B)$,
$\sigma_B \in {\cal P}({\cal H}_B)$, and $\varepsilon \ge 0$.
The $\varepsilon$-smooth min-entropy and the $\varepsilon$-smooth max-entropy
of $\rho_{AB}$ relative to $\sigma_B$ are defined by
\begin{eqnarray*}
H_{\min}^{\varepsilon}(\rho_{AB}|\sigma_B) &:=&
\sup_{\overline{\rho}_{AB}} H_{\min}(\overline{\rho}_{AB}|\sigma_B) \\
H_{\max}^{\varepsilon}(\rho_{AB}|\sigma_B) &:=&
\inf_{\overline{\rho}_{AB}} H_{\max}(\overline{\rho}_{AB}|\sigma_B),
\end{eqnarray*}
where the supremum and infimum ranges over the set
${\cal B}^{\varepsilon}(\rho_{AB})$ of all operators
$\overline{\rho}_{AB} \in {\cal P}^\prime({\cal H}_A \otimes {\cal H}_B)$ such
that $\| \overline{\rho}_{AB} - \rho_{AB} \| \le (\rom{Tr} \rho_{AB}) \varepsilon$.

The conditional 
$\varepsilon$-smooth min-entropy and the $\varepsilon$-smooth max-entropy
of $\rho_{AB}$ given ${\cal H}_B$ are defined by
\begin{eqnarray*}
H_{\min}^{\varepsilon}(\rho_{AB}|B) &:=& \sup_{\sigma_B} H_{\min}^{\varepsilon}(\rho_{AB}|\sigma_B) \\
H_{\max}^{\varepsilon}(\rho_{AB}|B) &:=& \sup_{\sigma_B} H_{\max}^{\varepsilon}(\rho_{AB}|\sigma_B),
\end{eqnarray*}
where the supremum ranges over all $\sigma_B \in {\cal P}({\cal H}_B)$.
\end{definition}

The following lemma is a kind of chain rule for the smooth min-entropy.
\begin{lemma}
     \label{lemma-chain-rule}
\cite[Theorem 3.2.12]{renner:05b}
For a tripartite operator 
$\rho_{ABC} \in {\cal P}^\prime({\cal H}_A \otimes {\cal H}_B \otimes {\cal
 H}_C)$, we have
\begin{eqnarray}
   \label{chain-rule}
H_{\min}^{\varepsilon}(\rho_{ABC}|C)
\le H_{\min}^{\varepsilon}(\rho_{ABC}|BC) + H_{\max}(\rho_B).
\end{eqnarray}
\end{lemma}

The following lemma states that removing the classical system 
only decreases the min-entropy.
\begin{lemma}
\label{monotonicity-of-min-entropy}
\cite[Lemma 3.1.9]{renner:05b}
(monotonicity of min-entropy)
Let $\rho_{XBC} \in {\cal P}^\prime({\cal H}_X \otimes {\cal H}_B \otimes {\cal H}_C)$ be
classical on ${\cal H}_X$, and let $\sigma_C \in {\cal P}({\cal H}_C)$.
Then, we have
\begin{eqnarray*}
H_{\min}(\rho_{XBC}|\sigma_C) \ge H_{\min}(\rho_{BC}|\sigma_C).
\end{eqnarray*}
\end{lemma}
In order to extend Lemma \ref{monotonicity-of-min-entropy} to the
smooth min-entropy, we need Lemmas \ref{neighbor-purification} 
and \ref{neighbor-purification-classical}.
\begin{lemma}
\label{neighbor-purification}
Let $\rho_{AB} \in {\cal P}({\cal H}_A \otimes {\cal H}_B)$ be a density operator.
For $\varepsilon \ge 0$, let 
$\hat{\rho}_B \in {\cal B}^\varepsilon(\rho_B)$. Then, there exists
a operator $\hat{\rho}_{AB} \in {\cal B}^{\bar{\varepsilon}}(\rho_{AB})$ such
that $\rom{Tr}_A[\hat{\rho}_{AB}] = \hat{\rho}_B$,
where $\bar{\varepsilon} := \sqrt{8 \varepsilon}$.
\end{lemma}
\begin{proof}
Since $\hat{\rho}_B \in {\cal B}^\varepsilon(\rho_B)$, we have
\begin{eqnarray*}
\| \hat{\rho}_B \| \ge \| \rho_B \| - \| \rho_B - \hat{\rho}_B \| \ge 1 - \varepsilon.
\end{eqnarray*}
Then, from Lemma \ref{lower-bound-by-fidelity}, we have
\begin{eqnarray*}
F(\rho_B, \hat{\rho}_B) &\ge& \frac{1}{2} (
\rom{Tr}\rho_B + \rom{Tr} \hat{\rho}_B - \| \rho_B - \hat{\rho}_B \| ) \\
&\ge& 1 - \varepsilon.
\end{eqnarray*}
Let $\ket{\Psi} \in {\cal H}_R \otimes {\cal H}_A \otimes {\cal H}_B$ be a purification
of $\rho_{AB}$. Then, from Theorem \ref{uhlman-theorem}, there exists a purification
$\ket{\Phi} \in {\cal H}_R \otimes {\cal H}_A \otimes {\cal H}_B$ of
$\hat{\rho}_B$ such that
\begin{eqnarray*}
F(\ket{\Psi}, \ket{\Phi}) = F(\rho_B, \hat{\rho}_B) \ge 1 - \varepsilon.
\end{eqnarray*}
By noting that $F(\ket{\Psi}, \ket{\Phi})^2 \ge 1 - 2\varepsilon$, from
Lemma \ref{upper-bound-by-fidelity}, we have
\begin{eqnarray*}
\| \ket{\Psi}\bra{\Psi} - \ket{\Phi}\bra{\Phi} \| \le
\sqrt{8 \varepsilon}.
\end{eqnarray*}
Let $\hat{\rho}_{AB} := \rom{Tr}_R[ \ket{\Phi}\bra{\Phi} ]$.
Then, since the trace distance does not increase by the partial trace, we have
\begin{eqnarray*}
\| \rho_{AB} - \hat{\rho}_{AB} \| \le \sqrt{8 \varepsilon}.
\end{eqnarray*}
\end{proof}
\begin{remark}
In Lemma \ref{neighbor-purification}, if the density operator
$\rho_{AB}$ is classical with respect to both systems 
${\cal H}_A  \otimes {\cal H}_B$, then we can easily replace
$\bar{\varepsilon}$ by $\varepsilon$.
Then, $\bar{\varepsilon}$ in Lemma \ref{neighbor-purification-classical},
\ref{lemma-weak-monotonicity} and Corollary
\ref{lemma-weak-chain-rule} can also be replaced by
$\varepsilon$.
\end{remark}
\begin{lemma}
\label{neighbor-purification-classical}
Let $\rho_{XB} \in {\cal P}({\cal H}_X \otimes {\cal H}_B)$ be a density operator
that is classical on ${\cal H}_X$.
For $\varepsilon \ge 0$, let $\hat{\rho}_B \in {\cal B}^\varepsilon(\rho_B)$.
Then, there exists a operator $\hat{\rho}_{XB} \in {\cal B}^{\bar{\varepsilon}}(\rho_{XB})$
such that $\rom{Tr}_X[ \hat{\rho}_{XB} ] = \hat{\rho}_B$ and $\hat{\rho}_{XB}$ is classical on
${\cal H}_X$, where $\bar{\varepsilon} := \sqrt{8 \varepsilon}$.
\end{lemma}
\begin{proof}
From Lemma \ref{neighbor-purification}, there exists a operator 
$\rho^\prime_{XB} \in {\cal B}^{\bar{\varepsilon}}(\rho_{XB})$ such
that $\rom{Tr}_X[ \rho_{XB}^\prime ] = \hat{\rho}_B$.
Let ${\cal E}_X$ be a projection measurement CP map on ${\cal H}_X$, i.e.,
\begin{eqnarray*}
{\cal E}_X(\rho) := \sum_{x \in {\cal X}} \ket{x}\bra{x} \rho \ket{x}\bra{x},
\end{eqnarray*}
where $\{ \ket{x} \}_{x \in {\cal X}}$ is an orthonormal basis of ${\cal H}_X$.
Let $\hat{\rho}_{XB} := ({\cal E}_X \otimes \rom{id}_B)(\rho^\prime_{XB})$.
Then, since the trace distance does not increase by the CP map, and 
$({\cal E}_X \otimes \rom{id}_B)(\rho_{XB}) = \rho_{XB}$, we have
\begin{eqnarray*}
\lefteqn{
\| \hat{\rho}_{XB} - \rho_{XB} \| } \\
&=&
\| ({\cal E}_X \otimes \rom{id}_B)(\rho_{XB}^\prime) -
({\cal E}_X \otimes \rom{id}_B)(\rho_{XB}) \| \\
&\le&
\| \rho_{XB}^\prime - \rho_{XB}
\| \\
&\le& \bar{\varepsilon},
\end{eqnarray*}
where the first inequality follows from 
Lemma \ref{lemma:monotonicity-of-trace-distance}.
Furthermore, we have
$\rom{Tr}_X [ \hat{\rho}_{XB} ] = \rom{Tr}_X [ \rho_{XB}^\prime ] = \hat{\rho}_B$,
and $\hat{\rho}_{XB}$ is classical on ${\cal H}_X$.
\end{proof}

The following lemma states that
the monotonicity of the min-entropy (Lemma \ref{monotonicity-of-min-entropy})
can be extended to the smooth min-entropy 
by adjusting the smoothness $\varepsilon$. 
\begin{lemma}
\label{lemma-weak-monotonicity}
Let $\rho_{XBC} \in {\cal P}({\cal H}_X \otimes {\cal H}_B \otimes {\cal H}_C)$
be a density operator that is classical on ${\cal H}_X$. 
Then, for any $\varepsilon \ge 0$, we have
\begin{eqnarray*}
H_{\min}^{\bar{\varepsilon}}(\rho_{XBC} | C) \ge 
H_{\min}^{\varepsilon}(\rho_{BC} | C),
\end{eqnarray*}
where $\bar{\varepsilon} := \sqrt{8 \varepsilon}$.
\end{lemma}
\begin{proof}
We will prove that
\begin{eqnarray*}
H_{\min}^{\bar{\varepsilon}}(\rho_{XBC}|\sigma_C) \ge 
H_{\min}^{\varepsilon}(\rho_{BC} | \sigma_C)
\end{eqnarray*}
holds for any $\sigma_C \in {\cal P}({\cal H}_C)$.
From the definition of the smooth min-entropy,
for any $\nu > 0$, there exists $\hat{\rho}_{BC} \in {\cal B}^{\varepsilon}(\rho_{BC})$
such that
\begin{eqnarray}
\label{eq-proof-of-weak-monotonicity-1}
H_{\min}(\hat{\rho}_{BC} | \sigma_C) \ge H_{\min}^{\varepsilon}(\rho_{BC}|\sigma_C) -\nu.
\end{eqnarray}
From Lemma \ref{neighbor-purification-classical}, there exists a operator
$\hat{\rho}_{XBC} \in {\cal B}^{\bar{\varepsilon}}(\rho_{XBC})$ such that
$\rom{Tr}_X [ \hat{\rho}_{XBC} ] = \hat{\rho}_{BC}$, and
$\hat{\rho}_{XBC}$ is classical on ${\cal H}_X$.
Then, from Lemma \ref{monotonicity-of-min-entropy}, we have
\begin{eqnarray}
\label{eq-proof-of-weak-monotonicity-2}
H_{\min}(\hat{\rho}_{XBC}|\sigma_C) \ge H_{\min}(\hat{\rho}_{BC}|\sigma_C).
\end{eqnarray}
Furthermore, from the definition of smooth min-entropy, we have
\begin{eqnarray}
\label{eq-proof-of-weak-monotonicity-3}
H_{\min}^{\bar{\varepsilon}}(\rho_{XBC} | \sigma_C) \ge 
H_{\min}(\hat{\rho}_{XBC} | \sigma_C).
\end{eqnarray}
Since $\nu > 0$ is arbitrary, combining 
Eqs.~(\ref{eq-proof-of-weak-monotonicity-1})--(\ref{eq-proof-of-weak-monotonicity-3}), we 
have the assertion of the lemma.
\end{proof}

Combining Eq.~(\ref{chain-rule}) of Lemma \ref{lemma-chain-rule}
and Lemma \ref{lemma-weak-monotonicity}, we have the following corollary,
which states that the condition decreases 
the smooth min-entropy by at most the amount of
the max-entropy of the condition, and 
plays an important role to prove the security of the QKD protocols.
\begin{corollary}
\label{lemma-weak-chain-rule}
Let $\rho_{XBC} \in {\cal P}({\cal H}_X \otimes {\cal H}_B \otimes {\cal H}_C)$
be a density operator that is classical on ${\cal H}_X$. 
Then, for any $\varepsilon \ge 0$, we have
\begin{eqnarray*}
H_{\min}^{\bar{\varepsilon}}(\rho_{XBC} | XC) \ge 
H_{\min}^{\varepsilon}(\rho_{BC} | C) - H_{\max}(\rho_X),
\end{eqnarray*}
where $\bar{\varepsilon} := \sqrt{8 \varepsilon}$.
\end{corollary}

For a product $\{cq\}$-state $\rho_{XB}^{\otimes n}$,
the smooth min-entropy can be evaluated by using the
von Neumann entropy.
\begin{lemma}
\label{lemma:min-entropy-of-product}
\cite[Corollary 3.3.7]{renner:05b}\footnote{See also Ref.~[22] of \cite{scarani:07}}
Let $\rho_{XB} \in {\cal P}({\cal H}_X \otimes {\cal H}_B)$ be a 
density operator which is classical on ${\cal H}_X$.
Then for $\varepsilon \ge 0$, we have
\begin{eqnarray*}
\frac{1}{n}H_{\min}^{\varepsilon}(\rho_{XB}^{\otimes n}|B^n) \ge 
  H(\rho_{XB}) - H(\rho_B) - \delta,
\end{eqnarray*}
where $\delta := (2 H_{\max}(\rho_X) + 3) \sqrt{\frac{\log(2/\varepsilon)}{n}}$.
\end{lemma}
\subsection{Privacy Amplification}
\label{subsec:pa}

The following definition is used to state the 
security of the distilled key by the privacy amplification.

\begin{definition}
\cite[Definition 5.2.1]{renner:05b}
\label{definition-distance-from-uniform}
Let $\rho_{AB} \in {\cal P}^\prime({\cal H}_A \otimes {\cal H}_B)$.
Then the trace distance from the uniform of $\rho_{AB}$ given $B$ is
defined by
\begin{eqnarray*}
d(\rho_{AB}|B) := \| \rho_{AB} - \rho_A^{\rom{mix}} \otimes \rho_B \|,
\end{eqnarray*}
where $\rho_A^{\rom{mix}} := \frac{1}{\dim {\cal H}_A} \rom{id}_A$ is the fully
mixed state on ${\cal H}_A$ and 
$\rho_B := \rom{Tr}_A[\rho_{AB}]$.
\end{definition}
\begin{definition}
\cite{carter:79}
\label{definition-two-universal-hash}
Let ${\cal F}$ be a set of functions from ${\cal X}$ to
${\cal S}$, and let $P_F$ be the uniform probability distribution on ${\cal F}$.
The set ${\cal F}$ is called {\em universal hash family} if
$\Pr \{ f(x) = f(x^\prime) \} \le \frac{1}{|{\cal Z}|}$ for
any distinct $x, x^\prime \in {\cal X}$.
\end{definition}

Consider an operator 
$\rho_{XE} \in {\cal P}^\prime({\cal H}_X \otimes {\cal H}_E)$ 
that is classical with respect to
an orthonormal basis $\{ \ket{x} \}_{x \in {\cal X}}$ of ${\cal H}_X$,
and assume that $f$ is a function from ${\cal X}$ to ${\cal S}$.
The operator describing the classical function output together
with the quantum system ${\cal H}_E$ is then given by
\begin{eqnarray}
\label{eq-state-key}
\rho_{f(X)E} := \sum_{s \in {\cal S}} \ket{s}\bra{s} \otimes \rho_E^s
~\mbox{for } \rho_E^s := \sum_{x \in f^{-1}(z)} \rho_E^x,
\end{eqnarray}
where $\{ \ket{s} \}_{s \in {\cal S}}$ is an orthonormal basis of ${\cal H}_S$.

Assume now that the function $f$ is randomly chosen from a set ${\cal
F}$
of function according to the uniform probability distribution $P_F$.
Then the output $f(x)$, the state of the quantum system, and the
choice of the function $f$ is  described by the operator
\begin{eqnarray}
\label{state-distilled-key}
\rho_{F(X)EF} := \sum_{f \in {\cal F}} P_F(f) \rho_{f(X)E} \otimes
\ket{f}\bra{f}
\end{eqnarray}
on ${\cal H}_S \otimes {\cal H}_E \otimes {\cal H}_F$, where ${\cal H}_F$
is a Hilbert space with orthonormal basis $\{ \ket{f} \}_{f \in {\cal F}}$.
The system ${\cal H}_S$ describes the distilled key, and 
the system ${\cal H}_E$ and  ${\cal H}_F$ describe the information
which an adversary Eve can access. 
The following lemma states that the length of securely distillable key
is given by the conditional smooth min-entropy
$H_{\min}^{\varepsilon}(\rho_{XE}|E)$.

\begin{lemma}
\cite[Corollary 5.6.1]{renner:05b}
      \label{lemma-privacy-amplification}
Let $\rho_{XE} \in {\cal P}({\cal H}_X \otimes {\cal H}_E)$ be a density
operator which is classical with respect to an orthonormal basis 
$\{ \ket{x} \}_{x \in {\cal X}}$ of ${\cal H}_X$.
Let ${\cal F}$ be a universal hash family of functions from
${\cal X}$ to $\{ 0,1\}^\ell$, and let $\varepsilon > 0$.
Then we have
\begin{eqnarray*}
d(\rho_{F(X)EF}|EF) \le 2 \varepsilon + 
2^{-\frac{1}{2} ( H_{\min}^{\varepsilon}(\rho_{XE}|E) - \ell)}
\end{eqnarray*}
for $\rho_{F(X)EF} \in {\cal P}({\cal H}_S \otimes {\cal H}_E \otimes {\cal H}_F)$
defined by Eq.~(\ref{state-distilled-key}).
\end{lemma}

By using Corollary \ref{lemma-weak-chain-rule} and 
Lemma \ref{lemma-privacy-amplification}, we can 
derive the following corollary, which gives the length
of the securely distillable key when Eve can access 
classical information in addition to the quantum information.
\begin{corollary}
\label{corollary-privacy-amplification}
Let $\rho_{XCE}$ be a density operator 
on ${\cal P}({\cal H}_X \otimes {\cal H}_C \otimes {\cal H}_E)$
that is classical with respect to the systems $X$ and $C$.
Let ${\cal F}$ be a universal family of hash functions from
${\cal X}$ to $\{0,1\}^\ell$, and let $\varepsilon > 0$.
If 
\begin{eqnarray*}
\ell < H_{\min}^{\bar{\varepsilon}}(\rho_{XE}|E) - \log \dim {\cal H}_C
 - 2 \log(1/\varepsilon),
\end{eqnarray*}
then we have
\begin{eqnarray*}
d(\rho_{F(X)CEF}|CEF) \le 3 \varepsilon,
\end{eqnarray*}
where $\bar{\varepsilon} = \varepsilon^2/8$.
\end{corollary}
\begin{remark}
\label{remark:privacy-amplification}
When the density operator $\rho_{XCE}$ is such that
the system $C$ only depends on $X$, then $\bar{\varepsilon}$
in Corollary \ref{corollary-privacy-amplification} can be
replaced by $\varepsilon$ \cite[Lemma 6.4.1]{renner:05b}.
\end{remark}



\chapter{Channel Estimation}
\label{ch:channel-estimation}

\section{Background}

As we have mentioned in 
Chapter \ref{ch:introduction},
the QKD protocols consists of three phases:
the bit transmission phase, the channel estimation phase,
and the postprocessing phases.
The postprocessing is a procedure in which Alice and
Bob generate a secret key from their bit sequences
obtained in the bit transmission phase, and the 
key generation rate (the length of the generated key
divided by the length of their initial bit sequences)
is decided according to the amount of Eve's ambiguity
about their bit sequence estimated in the channel 
estimation phase. The channel estimation
phase is the main topic investigated in 
this chapter.

Mathematically, quantum channels are described by 
trace preserving completely positive (TPCP) maps
\cite{nielsen-chuang:00}.
Conventionally in the QKD protocols, we only use 
the statistics of matched measurement outcomes,
which are transmitted and received by the same basis,
to estimate the TPCP map describing the quantum channel;
mismatched measurement outcomes, which are
transmitted and received by different bases, are discarded
in the conventionally used channel estimation methods.
By using the statistics of mismatched measurement outcomes
in addition to that of matched measurement outcomes, we can estimate 
the TPCP map more accurately than the conventional 
estimation method. Such an accurate channel estimation method
is also known as the quantum tomography \cite{chuang:97, poyatos:97}.
In early 90s, Barnett et al.~\cite{barnett:93} 
showed that the use of mismatched measurement
outcomes enables Alice and Bob to detect the presence of Eve with 
higher probability for the so-called intercept and resend attack.
Furthermore, some literatures use the accurate estimation method to ensure
the channel to be a Pauli channel
\cite{bruss:03,liang:03, kaszlikowski:05, kaszlikowski:05b},
where a Pauli channel is a channel over which four kinds of 
Pauli errors (including the identity) occur probabilistically.
However the channel is not necessarily a Pauli channel.

The use of the accurate channel estimation method
has a potential to improve the key generation rates of 
the QKD protocols. For this purpose, we have to
construct a postprocessing that fully utilize the
accurate channel estimation results.
However, there was no proposed practically implementable postprocessing that 
can fully utilizes the accurate estimation method.
Recently, Renner et al.~\cite{renner:05, renner:05b, kraus:05}
developed information theoretical techniques to
prove the security of the QKD protocols. 
Their proof techniques can be used to prove the security
of the QKD protocols with a postprocessing that fully utilizes
the accurate estimation method.
However they only considered Pauli channels or
partial twirled  
channels\footnote{By the partial twirling (discrete twirling)
\cite{bennett:96b}, any channel becomes a Pauli channel.}.
For Pauli channels, the accurate estimation method
and the conventional estimation method make no difference.

In this chapter, we propose 
a channel estimation procedure
in which we use the mismatched measurement outcomes
in addition to the matched measurement outcomes,
and also propose a postprocessing that fully 
utilize our channel estimation procedure.
We use the Slepian-Wolf 
coding \cite{slepian:73} with the linear code
(linear Slepian-Wolf coding) 
in our information reconciliation (IR) procedure.

The use of the linear Slepian-Wolf coding
in the IR procedure has the following advantage over the 
IR procedures in the literatures
\cite{renner:05, renner:05b, kraus:05, devetak:04}.
In \cite{devetak:04}, the authors constructed their IR procedure
by the so-called random coding method. Therefore, their
IR procedure is not practically implementable.
In \cite{renner:05, renner:05b, kraus:05}, the authors
constructed their IR procedure by randomly choosing
an encoder from a universal hash family\footnote{See Definition \ref{definition-two-universal-hash}
for the definition of the universal hash family.}.
Their IR procedure is essentially equivalent to 
the Slepian-Wolf coding. However, the ensemble
the encoder of the low density parity check (LDPC) code, which
is one of the practical linear codes, is not
a universal hash family.
On the other hand, 
the linear code in our IR procedure can be 
a LDPC code.

The rest of this chapter is organized as follows:
In Section \ref{sec:bb84-and-six-state-protocol},
we explain the bit transmission phase of the 
QKD protocols with some technical terminologies.
Then, we formally describe the problem setting of
the QKD protocols.
In Section \ref{sec:one-way-IR},
we show our IR procedure.
In Section \ref{subsec:key-generation-rate},
we show our proposed channel estimation procedure, and
then clarify a sufficient condition such that Alice
and Bob can share a secure key
(Theorem \ref{theorem:one-way-security}).
Then, we derive the asymptotic key generation rate
formulae.
In Section \ref{sec:relation-to-conventional},
we clarify the relation between our proposed
channel estimation procedure and the conventional
channel estimation procedure. 
In Section \ref{sec:example}, we investigate 
the asymptotic key generation rates for some representative examples
of channels.

It should be noted that most of the results in this chapter
first appeared in \cite{watanabe:08}. However, some of the results
in Section \ref{subsec:amplitude-damping} and 
Section \ref{sec:condition-strict-improvement} are newly obtained
in this thesis.

\section{BB84 and Six-State Protocol}
\label{sec:bb84-and-six-state-protocol}

In the six-state protocol,
Alice randomly sends bit $0$ or $1$ to Bob by modulating it
into a transmission basis that is randomly chosen from
the $\san{z}$-basis $\{ \ket{0_\san{z}}, \ket{1_\san{z}} \}$,
the $\san{x}$-basis $\{ \ket{0_\san{x}}, \ket{1_\san{x}} \}$, 
or the $\san{y}$-basis $\{ \ket{0_\san{y}}, \ket{1_\san{y}} \}$,
where $\ket{0_\san{a}}, \ket{1_\san{a}} $ are eigenstates of the Pauli 
operator $\sigma_\san{a}$ for $\san{a} \in \{\san{x},\san{y},\san{z}\}$ respectively.  
Then Bob randomly chooses one of measurement observables
$\sigma_\san{x}$, $\sigma_\san{y}$, and $\sigma_\san{z}$, and converts
a measurement result $+1$ or $-1$ into
a bit $0$ or $1$ respectively.
After a sufficient number of transmissions, Alice
and Bob publicly announce their transmission bases and
measurement observables. 
They also announce a part of their bit sequences 
as sample bit sequences for
estimating channel between Alice and Bob.

In the BB84 protocol, Alice only uses $\san{z}$-basis and
$\san{x}$-basis to transmit the bit sequence, 
and Bob only uses observables $\sigma_\san{z}$ and 
$\sigma_\san{x}$ to receive the bit sequence. 

For simplicity we assume that Eve's attack is 
the collective attack, i.e., the channel connecting
Alice and Bob is given by tensor products of 
a channel $\mathcal{E}_B$
from a qubit density operator to itself.
This assumption is not a restriction for
Eve's attack by the following reason.
Suppose that Alice and Bob perform a random permutation to
their bit sequence. By performing this random
permutation, the channel between Alice and Bob becomes
permutation invariant. Then, we can 
asymptotically reduce the security of
the QKD protocols for the most general attack,
the coherent attack, to the security
of the collective attack by using the 
(quantum) de Finetti representation 
theorem \cite{renner:05b, renner:07, christandl:09}.
Roughly speaking, the de Finetti representation theorem
says that (randomly permuted) general attack can be
approximated by a convex mixture of collective attacks. 

So far we have explained the so-called prepare and
measure scheme of the QKD protocols. There is the so-called
entanglement based scheme of the QKD protocols \cite{ekert:91}. 
In the entanglement
based scheme, Alice prepares the Bell state 
\begin{eqnarray*}
\ket{\psi} = \frac{\ket{00} + \ket{11}}{\sqrt{2}}, 
\end{eqnarray*}
and sends the second system to Bob over the quantum 
channel ${\cal E}_B$. Then, Alice and Bob conduct 
measurements for the shared state
\begin{eqnarray*}
\rho_{AB} := (\rom{id} \otimes {\cal E}_B)(\psi)
\end{eqnarray*}
by using randomly chosen observables $\sigma_\san{a}$
and $\sigma_\san{b}$ respectively. Although the 
entangled based  scheme 
is essentially equivalent to the prepare and measure
scheme \cite{bennett:92}, the latter is more practical
in the present day technology because Alice and Bob
do not need the quantum memory to store qubits.
However, the former is more convenient to mathematically 
treat the BB84 protocol and the six-state protocol
in a unified manner. Therefore in the rest of this thesis,
we employ the entanglement based scheme of the QKD
protocols, and consider the following situation.

Suppose that Alice and Bob share the bipartite 
(qubits) system $({\cal H}_A \otimes {\cal H}_B)^{\otimes N}$
whose state is $\rho_{AB}^{\otimes N}$. Alice and 
Bob conduct measurements for the first $n$ (out of $N$) bipartite systems
by $\san{z}$-basis respectively\footnote{In this thesis, we mainly
consider a secret key generated from Alice and Bob's measurement
outcomes by $\san{z}$-basis. Therefore, we occasionally omit
the subscripts $\{\san{x}, \san{y}, \san{z}\}$ of bases, and
the basis $\{ \ket{0}, \ket{1} \}$ is regarded as $\san{z}$-basis
unless otherwise stated.}. They also conduct measurements for the latter
$m$ (out of $N$) bipartite systems by randomly chosen bases from the set
${\cal J}_b := \{\san{x},\san{z}\}$ in the BB84 protocol and 
${\cal J}_s := \{\san{x}, \san{z}, \san{y}\}$ in the 
six-state protocol. Formally, the measurement for the latter $m$
systems can be described by the bipartite POVM 
${\cal M} := \{ M_z \}_{z \in {\cal Z}}$ on 
the bipartite system ${\cal H}_A \otimes {\cal H}_B$, 
where ${\cal Z} := \mathbb{F}_2 \times {\cal J}_b \times \mathbb{F}_2
\times {\cal J}_b$
for the BB84 protocol and 
${\cal Z} := \mathbb{F}_2 \times {\cal J}_s \times \mathbb{F}_2 \times
{\cal J}_s$ for the six-state protocol.
Note that Alice and Bob generate a secret key from the first $n$
measurement outcomes 
$(\bol{x}, \bol{y}) \in \mathbb{F}_2^n \times \mathbb{F}_2^n$,
and they estimate an unknown density operator $\rho_{AB}$
by using the latter measurement outcomes 
$\bol{z} \in {\cal Z}^m$, which we call the sample sequence.
When we do not have to discriminate between the BB84 protocol
and the six-state protocol, we omit the subscripts of
${\cal J}_b$ and ${\cal J}_s$, and denote them by ${\cal J}$.

As is usual in QKD literatures, we assume\footnote{By
this assumption, we are considering the worst case, that is,
the security under this assumption implies the security for
the situation in which Eve can conduct a measurement for 
a subsystem ${\cal H}_{E^\prime}$ of ${\cal H}_E$. This fact
can be formally proved by using the monotonicity of the
trace distance, because the security is defined by 
using the trace distance in this thesis 
(see Section \ref{subsec:key-generation-rate}).}
that
Eve can obtain her information
by conducting a measurement for an environment system ${\cal H}_E$
such that a purification $\psi_{ABE}$ of $\rho_{AB}$ is
a density operator of joint system 
${\cal H}_A \otimes {\cal H}_B \otimes {\cal H}_E$.
Therefore, Alice's bit sequence $\bol{x} = (x_1,\ldots, x_n)$,
Bob's bit sequence $\bol{y} = (y_1,\ldots, y_n)$,
and the state in Eve's system can be described by 
the $\{ccq\}$-state 
\begin{eqnarray*}
\rho_{\bol{X}\bol{Y}\bol{E}} = 
  \sum_{(\bol{x}, \bol{y}) \in \mathbb{F}_2^n \times \mathbb{F}_2^n}
P_{XY}^n(\bol{x}, \bol{y}) \ket{\bol{x}, \bol{y}}\bra{\bol{x},\bol{y}}
\otimes \rho_{\bol{E}}^{\bol{x}, \bol{y}},
\end{eqnarray*}
where $P_{XY}^n$ is the product distribution of
$P_{XY}(x,y) := \rom{Tr}[ \ket{x,y}\bra{x,y} \rho_{AB}]$,
and $\rho_{\bol{E}}^{\bol{x}, \bol{y}} := \rho_E^{x_1,y_1} \otimes
\cdots \otimes \rho_E^{x_n, y_n}$ for the normalized 
density operator $\rho_E^{x,y}$ of
$\rom{Tr}_{AB}[(\ket{x,y}\bra{x,y} \otimes I_E) \psi_{ABE}]$.

\section{One-Way Information Reconciliation}
\label{sec:one-way-IR}

When Alice and Bob have correlated classical sequences,
$\bol{x}, \bol{y} \in \mathbb{F}_2^n$, the purpose of the IR 
procedure for Alice and Bob is to share the same classical sequence by
exchanging messages over the public authenticated channel, where
$\mathbb{F}_2$ is the field of order $2$.
Then, the purpose of the PA procedure is to extract a secret
key from the shared bit sequence.
In this section, we present the most basic IR procedure,
the one-way IR procedure. In the one-way IR procedure,
only Alice (resp.~Bob) transmit messages to Bob (resp.~Alice)
over the public channel.

Before describing our IR procedure, we should review the basic
facts of linear codes.
An $[n,n-k]$ classical linear code $\mathcal{C}$ is an 
$(n-k)$-dimensional linear subspace of $\mathbb{F}_2^n$, and its
parity check matrix $M$ is an $k \times n$ matrix
of rank $k$ with $0,1$ entries such that $M \mathbf{c} = \mathbf{0}$
for any codeword $\mathbf{c} \in \mathcal{C}$.
By using these preparations, our procedure is described as follows:
\begin{enumerate}
\renewcommand{\theenumi}{\roman{enumi}}
\renewcommand{\labelenumi}{(\theenumi)}
\item \label{one-way-IR-step1}
Alice calculates the syndrome $t = t(\bol{x}) := M \mathbf{x}$,
and sends it to Bob over the public channel.

\item \label{one-way-IR-step2}
Bob decodes $(\mathbf{y}, t)$ into an estimate of $\mathbf{x}$ by 
a decoder $\hat{\bol{x}}: \mathbb{F}_2^n \times \mathbb{F}_2^k \to \mathbb{F}_2^n$.
\end{enumerate}

In the QKD protocols, Alice and Bob do not
know the probability distribution $P_{XY}$ in advance,
and they estimate  candidates 
$\{ P_{XY, \theta} \mymid \theta \in \Theta\}$ of the
actual probability distribution $P_{XY}$. 
In order to use the above IR procedure in the QKD protocols,
the decoding error probability have to be universally small
for any candidate of the probability distribution.
For this reason, we introduce the concept that
an IR procedure is $\delta$-universally-correct\footnote{Early papers
of QKD protocols did not consider the universality of the IR procedure.
The need for the universality was first pointed out by
Hamada \cite{hamada:04} as long as the author's knowledge.} as follows.
\begin{definition}
\label{def:one-way-correct}
We define an IR procedure to be $\delta$-universally-correct
for the class $\{ P_{XY, \theta} \mymid \theta \in \Theta\}$
of probability distributions if
\begin{eqnarray*}
P_{XY,\theta}^n(\{ (\bol{x},\bol{y}) \mymid \bol{x} \neq
 \hat{\bol{x}}(\bol{y},t(\bol{x})) \}) \le \delta
\end{eqnarray*}
for every $\theta \in \Theta$.
\end{definition}

An example of a decoder that fulfils the universality
is the minimum entropy decoder defined by
\begin{eqnarray*}
\hat{\bol{x}}(\bol{y}, t) := \argmin_{\bol{x}: M \bol{x} = t} H(P_{\bol{x}\bol{y}}).
\end{eqnarray*}
\begin{theorem}
\label{theorem:universal-coding}
\cite[Theorem 1]{csiszar:82}
Let $r$ be a real number that satisfies
\begin{eqnarray*}
r > \min_{\theta \in \Theta} H(X_\theta|Y_\theta),
\end{eqnarray*}
where the random variables $(X_\theta,Y_\theta)$ are distributed
according to $P_{XY,\theta}$.
Then, for every sufficiently large $n$, there exists a $k \times n$ parity check 
matrix $M$ such that $\frac{k}{n} \le r$ and a constant
$E > 0$ that does not depends on $n$,
and then the decoding error probability by the minimum
entropy decoding satisfies
\begin{eqnarray*}
P_{XY,\theta}^n(\{ (\bol{x},\bol{y}) \mymid \bol{x} \neq
 \hat{\bol{x}}(\bol{y},t(\bol{x})) \}) \le e^{-n E}
\end{eqnarray*}
for every $\theta \in \Theta$.
\end{theorem}

\begin{remark}
\label{remark:ir-with-code}
Conventionally, we used the error correcting code
instead of the Slepian-Wolf coding in the 
IR procedure (e.g.~\cite{shor:00}). 
In this remark, we show that the leakage of
information to Eve in the above IR procedure is
as small as that in the IR procedure with the 
error correcting code. Furthermore, we show
the sufficient and necessary condition for
that the former equals to the latter.
 
For appropriately chosen
linear code ${\cal C} \subset \mathbb{F}_2^n$,
the IR procedure with the error correcting (linear) code
is conducted as follows.
\begin{enumerate}
\renewcommand{\theenumi}{\roman{enumi}}
\renewcommand{\labelenumi}{(\theenumi)}
\item \label{step1-ir-ec}
Alice randomly choose a code word $\bol{c} \in {\cal C}$,
and sends $\bol{c} + \bol{x}$ to Bob over the public channel.

\item Bob decodes $\bol{c} + \bol{x} + \bol{y}$ into an estimate
$\hat{\bol{c}}$ of the code word $\bol{c}$ by a decoder
from $\mathbb{F}_2^n$ to ${\cal C}$. Then, he obtains  an estimate
$\hat{\bol{x}}$ of $\bol{x}$ by subtracting $\hat{\bol{c}}$ from
the received public message $\bol{c} + \bol{x}$.
\end{enumerate}
Note that Step (\ref{step1-ir-ec}) is equivalent to
sending the syndrome $M \bol{x} \in \mathbb{F}_2^k$ to Bob
from the view point of Eve, because
Eve can know to which coset of $\mathbb{F}_2^n/ {\cal C}$
Alice's sequence $\bol{x}$ belongs by knowing $\bol{c} + \bol{x}$.
However, the length $k$ of the syndrome have to be larger
than that in the IR procedure with the Slepian-Wolf coding
by the following reason.

Define a probability distribution\footnote{For simplicity, we assume
that there exists only one candidate of distribution $P_{XY}$,
and omit $\theta$ in this remark.} on $\mathbb{F}_2$ as
\begin{eqnarray}
\label{eq-pw}
P_W(w) := \sum_{y \in \mathbb{F}_2} P_Y(y) P_{X|Y}(y+w|y).
\end{eqnarray}
Then the error $\mathbf{w} := \mathbf{x} + \mathbf{y}$
between Alice and Bob's sequences is distributed according to $P_W^n$. 
Since we can regard that the code word $\bol{c}$ is transmitted over
the binary symmetric channel (BSC) with the crossover probability $P_W(1)$,
the converse of the channel coding theorem \cite{cover}
implies that $\dim {\cal C}/n = 1 - k/n$ have to be smaller than
$1 - H(W)$. By using the log-sum inequality \cite{cover}
and Eq.~(\ref{eq-pw}), we have
\begin{eqnarray*}
\lefteqn{ H(X|Y) } \\
&=& \sum_{x, y \in \mathbb{F}_2}  P_Y(y) 
    P_{X|Y}(x|y) \log \frac{1}{P_{X|Y}(x|y)} \\
&=& \sum_{w ,y \in \mathbb{F}_2} 
    P_{Y}(y) P_{X|Y}(y+w|y) \log 
       \frac{P_{Y}(y)}{P_{Y}(y) P_{X|Y}(y+w|y)} \\
&\le& \sum_{w \in \mathbb{F}_2} P_W(w) \log
   \frac{1}{P_W(w)} \\
&=& H(W),
\end{eqnarray*}
and the equality holds if and only if
$P_{X|Y}(w|0)$ equals $ P_{X|Y}(1 + w|1)$ for any $w \in \mathbb{F}_2$.
\end{remark}

\begin{remark}
\label{remark:ldpc}
When we implement the above IR procedure, we should use
a parity check matrix with an efficient decoding algorithm.
For example, we may use the low density parity check (LDPC)
matrix \cite{gallager:63} with the sum-product algorithm.

For a given sequence $\mathbf{y} \in \mathbb{F}_2^n$,
and a syndrome $t \in \mathbb{F}_2^k$, define
a function
\begin{eqnarray}
\label{eq-sum-product}
P^*(\hat{\mathbf{x}}) := \prod_{j=1}^n P_{X|Y}(\hat{x}_j|y_j)
 \prod_{i=1}^k \mathbf{1}\left[ \sum_{l \in N(i)} 
 \hat{x}_l = t_i \right],
\end{eqnarray}
where $N(i) := \{ j \mid M_{ij} = 1\}$ for 
the parity check matrix $M$, and $\mathbf{1}[\cdot]$ is the
indicator function.
The function $P^*(\hat{\mathbf{x}})$ is the non-normalized
a posteriori probability distribution on $\mathbb{F}_2^n$
given $\mathbf{y}$ and $t$. The sum-product
algorithm is a method to (approximately) calculate
the marginal a posteriori probability, i.e.,
\begin{eqnarray*}
P^*_j(\hat{x}_j) := \sum_{\hat{x}_l,l \neq j}
  P^*(\hat{\mathbf{x}}).
\end{eqnarray*}
The definition of a posteriori probability in Eq.~(\ref{eq-sum-product})
is the only difference between the decoding for the Slepian-Wolf source coding
and that for the channel coding.
More precisely, we replace \cite[Eq.~(47.6)]{mackay-book} with
Eq.~(\ref{eq-sum-product}), and use the algorithm in 
\cite[Section 47.3]{mackay-book}.
The above procedure is a generalization of \cite{liveris:02},
and a special case of \cite{coleman:06}.

In QKD protocols we should 
minimize the block error probability 
rather than the bit error probability, because a bit error might
propagate to other bits after the privacy amplification.
Although the sum-product algorithm is designed to minimize
the bit error probability, it is known by computer simulations
that the algorithm makes the block error 
probability small \cite{mackay-book}. 

Unfortunately, it has not been shown analytically that
the LDPC matrix with the sum-product algorithm can satisfy
the condition in Definition \ref{def:one-way-correct}.
However, it has been shown that 
the LDPC matrix can
satisfy the condition in Definition \ref{def:one-way-correct}
if we use the maximum a posteriori probability (MAP) 
decoding with an estimated probability distribution
\cite{yamasaki:09}\footnote{In \cite{muramatsu:05}, Muramatsu
 {\em et.~al.}~has proposed to use the LDPC code and the MAP decoding for the Slepian-Wolf
code sysmtem. However, their result cannot be used in the context of
the QKD protocol, because there is an estimation error of 
the distribution $P_{XY}$.}.
Since the sum-product algorithm is a approximation of
the MAP decoding, we expect that the LDPC matrix with the sum-product algorithm can satisfy
the condition in Definition \ref{def:one-way-correct} as well.
\end{remark}
\section{Channel Estimation and Asymptotic Key Generation Rate}

\subsection{Channel Estimation Procedure}
\label{subsec:key-generation-rate}

In this section, we show the channel estimation procedure.
The purpose of the channel estimation procedure
is to estimate an unknown Choi operator
$\rho = \rho_{AB} \in {\cal P}_c$
from the sample sequence $\bol{z} \in {\cal Z}^m$.
By using the estimate of the Choi operator,
we show a condition on the parameters (the rate of the syndrome
and the key generation rate) in the postprocessing 
such that Alice and Bob can share a secure key
(Theorem \ref{theorem:one-way-security}).

Let us start with the channel estimation procedure of
the six-state protocol.
In this thesis, we employ the maximum likelihood
(ML) estimator:
\begin{eqnarray*}
\hat{\rho}(\bol{z}) := 
\argmax_{\rho \in {\cal P}_c}
P_\rho^m(\bol{z}),
\end{eqnarray*}
where $P_\rho^m$ is $m$ products of 
the probability distribution 
$P_\rho$ of the sample symbol $z \in {\cal Z}$
defined by $P_\rho(z) := \rom{Tr}[M_z \rho]$.

As we have seen in Section
\ref{sec:key-agreement-in-information-theory},
the conditional von Neumann entropy 
\begin{eqnarray*}
H_\rho(X|E) := H(\rho_{XE}) - H(\rho_E)
\end{eqnarray*}
plays an important role to decide the key generation
rate in the postprocessing, where 
\begin{eqnarray*}
\rho_{XE} := \rom{Tr}_B\left[ \left(\sum_{x \in \mathbb{F}_2}
                             \ket{x}\bra{x} \otimes I_{BE} \right)
\psi_{ABE} \left( \sum_{x \in \mathbb{F}_2}
                             \ket{x}\bra{x} \otimes I_{BE}\right)\right]
\end{eqnarray*}
for a purification $\ket{\psi_{ABE}}$ of $\rho = \rho_{AB}$.
Therefore, we have to estimate 
this quantity, $H_\rho(X|E)$. Actually, the estimator
\begin{eqnarray*}
\hat{H}_{\bol{z}}(X|E) := H_{\hat{\rho}(\bol{z})}(X|E)
\end{eqnarray*}
is the ML estimator of $H_\rho(X|E)$
\cite[Theorem 7.2.10]{casella-berger:02}.

Next, we consider the channel estimation procedure of
the BB84 protocol. 
Although the Choi operator $\rho$ is described 
by $12$ real parameters (in the Stokes parameterization), from 
Eqs.~(\ref{eq:def-stokes-R}) and (\ref{eq:def-stokes-t}),
we find that the distribution $P_\rho$ only depends on the parameters
$\omega := (R_\san{zz},R_\san{zx}, R_\san{xz}, R_\san{xx}, t_\san{z},
t_\san{x})$,
and does not depend on the parameters
$\tau := (R_\san{zy}, R_\san{xy}, R_\san{yz}, R_\san{yx},
R_\san{yy},t_\san{y})$.
Therefore, we regard
the set 
\begin{eqnarray*}
\Omega := \{ \omega \in \mathbb{R}^6 \mymid \exists \tau \in
 \mathbb{R}^6 ~(\omega,\tau) \in {\cal P}_c \}
\end{eqnarray*}
as the parameter space, and denote $P_\rho$ by $P_\omega$.
Then, we estimate the parameters $\omega$ by the ML estimator:
\begin{eqnarray*}
\hat{\omega}(\bol{z}) := \argmax_{\omega \in \Omega} P_\omega^m(\bol{z}),
\end{eqnarray*}

Since we cannot estimate the parameters $\tau$, we have to
consider the worst case, and estimate the quantity
\begin{eqnarray}
\label{eq:definition-of-worst-case}
\min_{\varrho \in {\cal P}_c(\omega)} H_{\varrho}(X|E)
\end{eqnarray}
for a given $\omega \in \Omega$, where the set
\begin{eqnarray*}
{\cal P}_c(\omega) := \{ \varrho = (\omega^\prime, \tau^\prime) \in
 {\cal P}_c  \mymid \omega^\prime = \omega \}
\end{eqnarray*}
is the candidates of Choi operators for a given $\omega \in \Omega$.
Actually, 
\begin{eqnarray*}
\hat{H}_{\bol{z}}(X|E) := \min_{\varrho \in {\cal
 P}_c(\hat{\omega}(\bol{z}))} H_\varrho(X|E)
\end{eqnarray*}
is the ML estimator of the quantity in 
Eq.~(\ref{eq:definition-of-worst-case}).

It is known that the ML estimator is a consistent
estimator
(with certain conditions, which are satisfied in our
case \cite{wald:49}), that is, the quantities
\begin{eqnarray}
\label{eq:definition-mu-1}
\mu_s(\alpha,m) := P_{\rho}^m(\{ \bol{z} \mymid \|\hat{\rho}(\bol{z}) -
 \rho\| > \alpha\})
\end{eqnarray}
for the six-state protocol and
\begin{eqnarray}
\label{eq:definition-mu-2}
\mu_b(\alpha,m) := P_\omega^m(\{ \bol{z} \mymid \| \hat{\omega}(\bol{z}) -
 \omega \| > \alpha\})
\end{eqnarray}
for the BB84 protocol converge to $0$ for any $\alpha > 0$ as 
$m$ goes to infinity.
In the rest of this thesis, we omit the subscripts of $\mu_s(\alpha,m)$
and $\mu_b(\alpha,m)$, and denote them by $\mu(\alpha,m)$. 

Since $H_\rho(X|E)$ is a continuous function of $\rho$,
which follows from the continuity of the von Neumann entropy,
there exists a function
$\eta_s(\cdot)$ such that
\begin{eqnarray}
\label{eq:consistency-bound-1}
| \hat{H}_\bol{z}(X|E) - H_\rho(X|E) | \le \eta_s(\alpha) 
\end{eqnarray}
for $\| \hat{\rho}(\bol{z}) - \rho\| \le \alpha$ and
$\eta_s(\alpha) \to 0$ as $\alpha \to 0$. 
Similarly, since Eq.~(\ref{eq:definition-of-worst-case})
is a continuous function of $\omega$, which will be
proved in Lemma \ref{lemma:continuity-of-min},
there exists a function 
$\eta_b(\cdot)$ such that
\begin{eqnarray}
\label{eq:consistency-bound-2}
| \hat{H}_\bol{z}(X|E) - \min_{\varrho \in
 {\cal P}_c(\omega)} H_\varrho(X|E) | \le \eta_b(\alpha) 
\end{eqnarray}
for $\|\hat{\omega}(\bol{z}) - \omega\| \le \alpha$
and $\eta_b(\alpha) \to 0$ as $\alpha \to 0$. 
In the rest of this thesis, 
we omit the subscripts of $\eta_s(\cdot)$ and $\eta_b(\cdot)$,
and denote them by $\eta(\cdot)$.

\subsection{Sufficient Condition on Key Generation Rates for Secure Key Agreement}
\label{subsec:sufficient-codition}

In this section, we explain how Alice and Bob decides the parameters
of the postprocessing and conduct it. Then, we show a sufficient
conditions on the parameters such that Alice and Bob can
share a secure key.

If the sample sequence is not contained in a prescribed
acceptable region ${\cal Q} \subset {\cal Z}^m$
(see Remark \ref{remark:abort} for the definition),
then Alice and Bob abort the protocol.
Otherwise, they decide the rate $\frac{k(\bol{z})}{n}$ of the linear code 
used in the IR procedure according to the sample bit sequence $\bol{z}$. 
Furthermore, they also
decide the length $\ell(\bol{z})$ of the finally distilled
key according to the sample sequence $\bol{z}$.
Then, they conduct the postprocessing as follows.
\begin{enumerate}
\renewcommand{\theenumi}{\roman{enumi}}
\renewcommand{\labelenumi}{(\theenumi)}

\item \label{one-way-postprocessing-step1}
Alice and Bob undertake the IR procedure
of Section \ref{sec:one-way-IR}, and 
Bob obtains the estimate $\hat{\bol{x}}$
of Alice's raw key $\bol{x}$.

\item \label{one-way-postprocessing-step2}
Alice and Bob carry out the
privacy amplification (PA) procedure 
to distill a key pair $(s_A, s_B)$ such 
that Eve has little information about it.
Alice first randomly chooses a function,
$f: \mathbb{F}_2^{n} \to \{ 0,1\}^{\ell(\bol{z})}$,
from a universal hash family 
(see Definition \ref{definition-two-universal-hash}),
and sends the choice of $f$ to Bob over the public channel.
Then, Alice's distilled key is
$s_A = f(\bol{x})$ and 
Bos's distilled key is $s_B = f(\hat{\bol{x}})$
respectively.
\end{enumerate}

We have explained the procedures of the postprocessing so far.
The next thing we have to do is to define the security
of the generated key formally.
By using the convention in Eq.~(\ref{eq:cq-state-mapping})
for the $\{ccq\}$-state $\rho_{\bol{X}\bol{Y}\bol{E}}$
and the mapping that describes the postprocessing,
the generated key pair and Eve's available information
can be described by a $\{cccq\}$-state,
$\rho_{S_A S_B C \bol{E}}^{\bol{z}}$, where classical system
$C$ consists of the random variable
$T$ that describe
 the syndrome transmitted in the IR procedure
and the random variable $F$ that describes the choice  of the 
function in the PA procedure. It should be noted that
the $\{cccq\}$-state $\rho_{S_A S_B C \bol{E}}^{\bol{z}}$
depends on the sample sequence $\bol{z}$ because the
parameters in the postprocessing is determined from it.
To define the security of the distilled key pair
$(S_A, S_B)$, we use the universally composable security 
definition \cite{ben-or:04, renner:05c} (see also \cite{renner:05b}),
which is defined by the trace distance between the
actual key pair and the ideal key pair.
We cannot state security of the QKD protocols 
in the sense that the distilled key
pair $(S_A, S_B)$ is secure 
for a particular sample sequence $\bol{z}$,
because there is a slight possibility that
the channel estimation procedure will underestimate
Eve's information.
\begin{definition}
\label{definition:security-of-key}
The generated key pair is said to be $\varepsilon$-secure (in the sense of the average
over the sample sequence\footnote{If it is obvious from the 
context, we occasionally use terms ``$\varepsilon$-secure'',
``$\varepsilon$-secret'', and ``$\delta$-correct'' for specific realization
$\bol{z}$ instead for average.}) if
\begin{eqnarray}
      \label{eq-varepsilon-secure-average}
\sum_{\bol{z} \in {\cal Q}} P_{\rho}^m(\bol{z}) 
\frac{1}{2} \|
\rho^\bol{z}_{S_A S_B C \bol{E}} - \rho_{S_A S_B}^{\bol{z},\rom{mix}} 
\otimes \rho^{\bol{z}}_{C \bol{E}}
\| \le \varepsilon
\end{eqnarray}
for any (unknown) Choi operator $\rho \in {\cal P}_c$ initially shared by Alice and Bob,
where $\rho_{S_A S_B}^{\bol{z},\rom{mix}} := \sum_{s \in {\cal S}_\bol{z}}
\frac{1}{|{\cal S}_\bol{z}|} \ket{s,s}\bra{s,s}$
is the uniformly distributed key on the key space
${\cal S}_\bol{z} := \{0, 1\}^{\ell(\bol{z})}$.
\end{definition}
\begin{remark}
\cite[Remark 6.1.3]{renner:05b}
The above security definition can be subdivided into
two conditions. If the generated key is
$\varepsilon$-secret, i.e.,
\begin{eqnarray*}
\sum_{\bol{z} \in {\cal Q}} P_{\rho}^m(\bol{z}) 
\frac{1}{2} \|
\rho^\bol{z}_{S_A C \bol{E}} - \rho_{S_A}^{\bol{z},\rom{mix}} 
\otimes \rho^{\bol{z}}_{C \bol{E}}
\| \le \varepsilon
\end{eqnarray*}
and $\delta$-correct, i.e.,
\begin{eqnarray*}
\sum_{\bol{z} \in {\cal Q}} P_{\rho}^m(\bol{z})  P_{S_A S_B}^{\bol{z}}(s_A \neq s_B) \le \delta,
\end{eqnarray*}
then the generated key pair is $(\varepsilon + \delta)$-secure.
\end{remark}

For a given Choi operator 
$\rho \in {\cal P}_c$,
we define the probability distribution 
$P_{XY,\rho} \in {\cal P}(\mathbb{F}_2 \times \mathbb{F}_2)$ as
\begin{eqnarray}
\label{eq:definition-of-pxy}
P_{XY, \rho}(x,y) := \rom{Tr}[ (\ket{x}\bra{x} \otimes
 \ket{y}\bra{y}) \rho].
\end{eqnarray}
Actually, $P_{XY,\rho}$ does not depend on the
parameter $\tau$ in the BB84 protocol. Therefor, we denote
$P_{XY,\rho}$ by $P_{XY,\omega}$ when we 
treat the BB84 protocol.

The following theorem gives a sufficient 
conditions
on $k(\bol{z})$ and $\ell(\bol{z})$ such that
the generated key pair is secure.
\begin{theorem}
\label{theorem:one-way-security}
For each sample sequence $\bol{z} \in {\cal Q}$, 
assume that the 
IR procedure is $\delta$-universally-correct for the
class of distributions
\begin{eqnarray*}
\{ P_{XY, \rho} \mymid \|\hat{\rho}(\bol{z}) - \rho\| \le \alpha \}
\end{eqnarray*}
in the six-state protocol, and for the
class of distributions 
\begin{eqnarray*}
\{ P_{XY,\omega} \mymid \|\hat{\omega}(\bol{z}) - \omega\| \le
 \alpha\}
\end{eqnarray*} 
in the BB84 protocol. 
For each $\bol{z} \in {\cal Q}$, if we set
\begin{eqnarray}
\label{eq:one-way-security-condition}
\frac{\ell(\bol{z})}{n} < 
\hat{H}_{\bol{z}}(X|E) - \eta(\alpha) - \frac{k(\bol{z})}{n} - \nu_n, 
\end{eqnarray}
then the distilled key pair is 
$(\varepsilon + \delta + \mu(\alpha,m))$-secure,
where 
$\nu_n := 5 \sqrt{\frac{\log(3/\varepsilon)}{n}} + \frac{2 \log(3/2\varepsilon)}{n}$.
\end{theorem}  
\begin{proof}
We only prove the statement for the six-state protocol,
because the statement for the BB84 protocol is proved
exactly in the same way by replacing
$\rho \in {\cal P}_c$
with $\omega \in \Omega$ and some other related quantities.
The assertion of the theorem follows from the 
combination of Corollary \ref{corollary-privacy-amplification}, 
Remark \ref{remark:privacy-amplification},
Lemma \ref{lemma:min-entropy-of-product}, and 
Eqs.~(\ref{eq:definition-mu-1}), and
(\ref{eq:consistency-bound-1}).

For any $\rho \in {\cal P}_c$,
Eq.~(\ref{eq:definition-mu-1}) means that 
$\| \hat{\rho}(\bol{z}) - \rho\| \le \alpha$ with probability
$1-\mu(\alpha,m)$. When $\|\hat{\rho}(\bol{z}) - \rho\| > \alpha$,
the distilled key pair trivially satisfies
\begin{eqnarray*}
\frac{1}{2} \| \rho_{S_A S_B C \bol{E}}^\bol{z} - 
  \rho_{S_A S_B}^{\bol{z},\rom{mix}} \otimes \rho_{C\bol{E}}^\bol{z} \| \le 1.
\end{eqnarray*}
On the other hand, when $\| \hat{\rho}(\bol{z}) - \rho\| \le \alpha$,
Eq.~(\ref{eq:one-way-security-condition}) implies
\begin{eqnarray*}
\ell(\bol{z}) < H_{\min}^{2\varepsilon/3}(\rho_{\bol{X}\bol{E}}
 |\bol{E}) - k(\bol{z}) - 2 \log(3/2\varepsilon)
\end{eqnarray*}
by using Lemma \ref{lemma:min-entropy-of-product}.
Thus the distilled key satisfies
\begin{eqnarray*}
\frac{1}{2} \| \rho_{S_A S_B C \bol{E}}^\bol{z} - 
  \rho_{S_A S_B}^{\bol{z},\rom{mix}} \otimes \rho_{C\bol{E}}^\bol{z} \| 
  \le \varepsilon + \delta
\end{eqnarray*}
by Corollary \ref{corollary-privacy-amplification},
Remark \ref{remark:privacy-amplification}, and the assumption that
the IR procedure is $\delta$-universally-correct for the
class of distribution 
$\{ P_{XY, \rho} \mymid  \| \hat{\rho}(\bol{z}) - \rho\| \le \alpha\}$. 
Averaging over the sample sequence $\bol{z} \in {\cal Q}$,
we have the assertion of the theorem.
\end{proof}

From Eq.~(\ref{eq:one-way-security-condition}),
we find that the estimator 
$\hat{H}_{\bol{z}}(X|E)$
of Eve's ambiguity and the syndrome rate $\frac{k(\bol{z})}{n}$
for the IR procedure are the important factors to
decide the key generation rate $\frac{\ell(\bol{z})}{n}$.
In the next section, we investigate the asymptotic 
behavior of the 
key generation rate derived from the right hand side of 
Eq.~(\ref{eq:one-way-security-condition}).

\begin{remark}
\label{remark:abort}
The acceptable region ${\cal Q} \subset {\cal Z}^m$
is defined as follows:
Each $\bol{z} \in {\cal Z}^m$ belongs to
${\cal Q}$ if and only if the right hand side of
Eq.~(\ref{eq:one-way-security-condition}) is positive.
\end{remark}
\begin{remark}
By switching the role of Alice and Bob, we obtain 
a postprocessing with the so-called reverse
reconciliation\footnote{The reverse reconciliation was originally 
proposed by Maurer in the classical key agreement 
context \cite{maurer:93}.}. On the other hand, the original
procedure is usually called the direct reconciliation.

In the reverse reconciliation, 
Bob sends syndrome $M \bol{y}$ to Alice, and Alice recovers
the estimate $\hat{\bol{y}}$ of Bob's sequence.
Then, Alice and Bob's final keys are 
$s_A = f(\hat{\bol{y}})$ and $s_B = f(\bol{y})$
for a randomly chosen function 
$f:\mathbb{F}_2^n \to \{0,1\}^{\ell(\bol{z})}$ from
a universal hash family.

For the postprocessing with the reverse reconciliation, 
we can show almost the same statement as Theorem \ref{theorem:one-way-security}
by replacing $\hat{H}_{\bol{z}}(X|E)$ with $\hat{H}_{\bol{z}}(Y|E)$,
which is defined in a similar manner as
$\hat{H}_{\bol{z}}(X|E)$, and by using
$\delta$-universally-correct for the reverse reconciliation.

In Section \ref{sec:example}, we will show that the asymptotic key 
generation rate of the reverse reconciliation can be higher
than that of the direct reconciliation. Although the fact that the 
asymptotic key generation rate of the direct reconciliation and the
reverse reconciliation are different is already pointed out
for QKD protocols with weak coherent states
\cite{boileau:05, hayashi:07}, it is new for the QKD protocols
with qubit states.
\end{remark}
\begin{remark}
\label{remark:noisy-preprocessing}
Although Alice and Bob
conducted the (direct) IR procedure for the pair of
bit sequence $(\bol{x}, \bol{y})$ in the postprocessing
explained so far, Alice can locally conducts a 
(stochastic) preprocessing for her bit sequence
before conducting the IR procedure. Surprisingly, 
Renner {\em et al}.~\cite{renner:05, renner:05b, kraus:05}
found that Alice should add noise to her bit sequence
in some cases, which is called the noisy preprocessing.
In the postprocessing with the noisy preprocessing,
Alice first flip each bit with
probability $q$ and obtain a bit sequence $\bol{u}$.
Then, Alice and Bob conduct the IR procedure and the 
PA procedure for the pair $(\bol{u}, \bol{y})$.
Renner {\em et al}.~found that, by appropriately choosing the value
$q$, the key generation rate can be improved.
\end{remark}

\subsection{Asymptotic Key Generation Rate of The Six-State Protocol}
\label{subsec:six-state}

In this section, we derive the asymptotic key
generation rate formula for the six-state protocol.
As we have seen
in Section \ref{subsec:key-generation-rate}, the estimator
$\hat{H}_{\bol{z}}(X|E)$ converges to the true value $H_\rho(X|E)$
in probability as $m$ goes to infinity.
On the other hand,
Theorem \ref{theorem:universal-coding} implies that
it is sufficient to set the rate of the syndrome so that
\begin{eqnarray}
\label{eq:asymptotic-syndrome-rate}
\frac{k(\bol{z})}{n}  > 
\min H_\varrho(X|Y)
\end{eqnarray}
for sufficiently large $n$,
where $H_\varrho(X|Y)$ is the conditional entropy\footnote{Equivalently,
we can regard $H_\varrho(X|Y)$ as the quantum conditional 
entropy for the classical density operator $\varrho_{XY}$.} for
the random variables $(X,Y)$ that 
are distributed according to $P_{XY,{\varrho}}$,
and the minimization is taken over the set 
$\{ \varrho \mymid \| \hat{\rho}(\bol{z}) - \varrho \| \le \alpha \}$.
Since the ML estimator $\hat{\rho}(\bol{z})$
is a consistency estimator of $\rho$, 
we can set the sequence of the syndrome rates so that
it converges to
$H_\rho(X|Y)$ in probability as $m,n \to \infty$.
Therefore, we can set the sequence of the key generation rates so that
it converges to the asymptotic key
generation rate formula
\begin{eqnarray}
\label{eq:asymptotic-direct}
H_{\rho}(X|E) - H_\rho(X|Y) 
\end{eqnarray}
in probability as $m, n \to \infty$.

Similarly for the postprocessing with the reverse reconciliation, 
we can set the sequence of the key 
generation rates so that it converges to
the asymptotic key generation rate formula
\begin{eqnarray}
\label{eq:asymptotic-reverse}
H_{\rho}(Y|E) - H_\rho(Y|X).
\end{eqnarray}

\subsection{Asymptotic Key Generation Rate of The BB84 Protocol}
\label{subsec:bb84}

In this section, we derive the asymptotic key
generation rate formula for the BB84 protocol.
As we have seen
in Section \ref{subsec:key-generation-rate}, the estimator
$\hat{H}_{\bol{z}}(X|E)$ converges to the true value
$\min_{\varrho \in {\cal P}_c(\omega)} H_\varrho(X|E)$
in probability as $m$ goes to infinity.
On the other hand, 
Theorem \ref{theorem:universal-coding} implies that
it is sufficient to set the rate of the syndrome so that
\begin{eqnarray}
\label{eq:asymptotic-syndrome-rate2}
\frac{k(\bol{z})}{n}  > 
\min H_\omega(X|Y)
\end{eqnarray}
for sufficiently large $n$,
where $H_\omega(X|Y)$ is the conditional entropy for
the random variables $(X,Y)$ that 
are distributed according to $P_{XY,\omega}$,
and the minimization is taken over the set 
$\{ \omega^\prime \mymid \| \hat{\omega}(\bol{z}) - \omega^\prime \| \le \alpha \}$.
Since the ML estimator $\hat{\omega}(\bol{z})$ is a consistency
estimator of $\omega$, we can set the sequence of the syndrome rates so that 
it converges to $H_\omega(X|Y)$ in probability as $m,n \to \infty$. 
Therefore, we can set the sequence of the key generation rates so that
it converges to the asymptotic key generation rate formula
\begin{eqnarray}
\label{eq:key-rate-one-way-bb84}
\min_{\varrho \in {\cal P}_c(\omega)} H_\varrho(X|E)
 - H_\omega(X|Y).
\end{eqnarray}

Similarly, for the postprocessing with the reverse reconciliation,
we can set the sequence of the key 
generation rates so that it converges to the
asymptotic key generation rate formula
\begin{eqnarray}
\label{eq:eq:key-rate-one-way-bb84-reverse}
\min_{\varrho \in {\cal P}_c(\omega)} H_\varrho(Y|E)
  - H_\omega(Y|X).
\end{eqnarray}


Although the asymptotic key generation rate formulae
for the six-state protocol 
(Eqs.~(\ref{eq:asymptotic-direct}) and (\ref{eq:asymptotic-reverse}))
do not involve the minimization, the asymptotic
key generation rate formulae for the BB84 protocol
(Eqs.~(\ref{eq:key-rate-one-way-bb84}) and (\ref{eq:eq:key-rate-one-way-bb84-reverse}))
involve the minimization, and therefore calculation
of these formula is not straightforward.
The following propositions are very useful for the
calculation of the asymptotic key generation rate
of the BB84 protocol.

\begin{proposition}
\label{proposition:convexity}
For two Choi operators 
$\rho^1, \rho^2 \in {\cal P}_c$
and a probabilistically mixture 
$\rho^\prime := \lambda \rho^1 + (1-\lambda)\rho^2$,
Eve's ambiguity is convex, i.e., we have
\begin{eqnarray*}
H_{\rho^\prime}(X|E) \le 
   \lambda H_{\rho^1}(X|E) + (1-\lambda) H_{\rho^2}(X|E),
\end{eqnarray*}
where $\rho^\prime_{XE}$ is $\{cq\}$-state derived from 
a purification $\psi^\prime_{ABE}$ of $\rho^\prime_{AB}$.
\end{proposition}

\begin{proof}
For $r = 1$ and $2$, let $\psi_{ABE}^r$ be a purification of the 
$\rho_{AB}^r$.
Then the density operator $\rho_{XE}^r$ is derived by
Alice's measurement by $\san{z}$-basis and the partial trace over Bob's system,
 i.e.,
\begin{eqnarray}
\label{eq-cq-state-1}
\rho_{XE}^r = \rom{Tr}_B \left[
\sum_{x} (\ket{x}\bra{x} \otimes I) \psi^r_{ABE} 
   (\ket{x}\bra{x} \otimes I)
\right].
\end{eqnarray}
Let 
\begin{eqnarray*}
\ket{\psi^\prime_{ABER}} := 
   \sqrt{\lambda} \ket{\psi_{ABE}^1} \ket{1} + 
   \sqrt{1-\lambda} \ket{\psi_{ABE}^2} \ket{2} 
\end{eqnarray*}
be a purification of $\rho^\prime_{AB}$, 
where $\mathcal{H}_R$ is the reference system,
and $\{ \ket{1}, \ket{2} \}$ is an orthonormal basis of $\mathcal{H}_R$.
Let 
\begin{eqnarray}
\label{eq-cq-state-2}
\rho^\prime_{XER} :=
    \rom{Tr}_B \left[
\sum_{x} (\ket{x}\bra{x} \otimes I) \psi^\prime_{ABER} 
   (\ket{x}\bra{x} \otimes I)
\right],
\end{eqnarray}
and let 
\begin{eqnarray*}
\rho^{*}_{XER} &:=& \sum_{r \in \{1,2\}} 
   (I \otimes \ket{r}\bra{r}) \rho^\prime_{XER} 
   (I \otimes \ket{r}\bra{r}) \\
&=& \lambda \rho_{XE}^1 \otimes \ket{1}\bra{1}
    + (1 - \lambda) \rho_{XE}^2 \otimes \ket{2}\bra{2}
\end{eqnarray*}
be the density operator such that  the system $\mathcal{H}_R$ is measured
 by $\{ \ket{1}, \ket{2}\}$ basis.
Then we have
\begin{eqnarray*}
\lefteqn{
H_{\rho^\prime}(X|ER) 
} \\
&=& H(X) - I_{\rho^\prime}(X; ER) \\
&\le& H(X) - I_{\rho^*}(X;ER) \\
&=& H_{\rho^*}(X|ER) \\
&=& \lambda H_{\rho^1}(X|E) + (1-\lambda) H_{\rho^2}(X|E),
\end{eqnarray*}
where the inequality follows from the monotonicity
of the quantum mutual information for measurements
(data processing inequality) \cite{hayashi-book:06}.
By renaming the systems $ER$ to $E$, we have the 
assertion of the lemma.
\end{proof}
\begin{remark}
In a similar manner, we can also show the convexity 
\begin{eqnarray*}
H_{\rho^\prime}(Y|E) \le 
   \lambda H_{\rho^1}(Y|E) + (1-\lambda) H_{\rho^2}(Y|E)
\end{eqnarray*}
under the same condition as in Proposition \ref{proposition:convexity}.
\end{remark}

The following proposition reduces the number of free parameters
in the minimization of Eqs.~(\ref{eq:key-rate-one-way-bb84}) 
and (\ref{eq:eq:key-rate-one-way-bb84-reverse}).
\begin{proposition}
\label{proposition:minimization}
For the BB84 protocol,
the minimization in Eqs.~(\ref{eq:key-rate-one-way-bb84}) 
and (\ref{eq:eq:key-rate-one-way-bb84-reverse}) is achieved
by Choi operator $\varrho$ whose components $R_{\san{zy}}$,
$R_{\san{xy}}$, $R_{\san{yz}}$, $R_{\san{yx}}$,
and $t_{\san{y}}$, are all $0$.
\end{proposition}
\begin{proof}
The statement of this proposition easily follows from
Proposition \ref{proposition:convexity}.
We only prove the statement for Eq.~(\ref{eq:key-rate-one-way-bb84})
because the statement for Eq.~(\ref{eq:eq:key-rate-one-way-bb84-reverse})
can be proved exactly in the same manner.

For any $\varrho \in {\cal P}_c(\omega)$,
let $\bar{\varrho}$ be the complex conjugate
of $\varrho$. Note that eigenvalues of density
matrices are unchanged by the complex conjugate, and thus
Eve's ambiguity $H_{\bar{\varrho}}(X|E)$ for
$\bar{\varrho}$ equals to $H_{\varrho}(X|E)$.
By applying Proposition \ref{proposition:convexity}
for $\rho^1 = \varrho$, $\rho^2 = \bar{\varrho}$,
and $\lambda = \frac{1}{2}$, we have
\begin{eqnarray*}
H_{\varrho^\prime}(X|E) \le \frac{1}{2}H_{\varrho}(X|E) + 
  \frac{1}{2} H_{\bar{\varrho}}(X|E),
\end{eqnarray*}
where
$\varrho^\prime = \frac{1}{2} \varrho + \frac{1}{2} \bar{\varrho}$.
Note that the Stokes parameterization of $\bar{\varrho}$
is given by 
\begin{eqnarray*}
\left(
\left[\begin{array}{ccc}
R_{\san{zz}} & R_{\san{zx}} & - R_{\san{zy}} \\
R_{\san{xz}} & R_{\san{xx}} & - R_{\san{xy}} \\
- R_{\san{yz}} & - R_{\san{yx}} & R_{\san{yy}}
\end{array}\right],
\left[\begin{array}{c}
t_{\san{z}} \\ t_{\san{x}} \\ - t_{\san{y}}
\end{array}\right] 
\right) \in {\cal P}_c(\omega).
\end{eqnarray*}
Therefore, the components, $R_{\san{zy}}$,
$R_{\san{xy}}$, $R_{\san{yz}}$, $R_{\san{yx}}$,
and $t_{\san{y}}$, of the Stokes parameterization of
$\varrho^\prime$ are all $0$.
Since ${\cal P}_c(\omega)$ is a convex set, 
$\varrho^\prime \in {\cal P}_c(\omega)$.
Since $\varrho \in {\cal P}_c(\omega)$ was
arbitrary, we have the assertion of the proposition.
\end{proof}

The following proposition can be used to calculate a lower
bound on the asymptotic key generation rate of the BB84
protocol.
\begin{proposition}
\label{proposition:unital-bound}
For the BB84 protocol, we have
\begin{eqnarray}
\lefteqn{ \min_{\varrho \in {\cal P}_c(\omega)} H_{\varrho}(X|E) }
 \nonumber \\
&\ge&  \hspace{-3mm} 1 - h\left(\frac{1 + d_\san{z}}{2}\right)
- h\left(\frac{1+d_\san{x}}{2}\right)
+ h\left( \frac{1+\sqrt{R_{\san{zz}}^2 + R_{\san{xz}}^2}}{2}\right)
\label{eq:unital-bound-direct}
\end{eqnarray}
and 
\begin{eqnarray}
\lefteqn{ \min_{\varrho \in {\cal P}_c(\omega)} H_{\varrho}(Y|E) }
 \nonumber \\
&\ge&  \hspace{-3mm} 1 - h\left(\frac{1 + d_\san{z}}{2}\right)
- h\left(\frac{1+d_\san{x}}{2}\right)
+ h\left( \frac{1+\sqrt{R_{\san{zz}}^2 + R_{\san{zx}}^2}}{2}\right),
\label{eq:unital-bound-reverse}
\end{eqnarray}
where $d_\san{z}$ and $d_\san{x}$ are the singular 
values of the  matrix
\begin{eqnarray}
\label{eq:two-times-two-matrix}
\left[\begin{array}{cc} R_{\san{zz}} & R_{\san{zx}} \\ R_\san{xz} &
  R_\san{xx} \end{array}\right]
\end{eqnarray}
for $\omega := (R_\san{zz},R_\san{zx}, R_\san{xz}, R_\san{xx}, t_\san{z},
t_\san{x})$.
The equalities
in Eqs.~(\ref{eq:unital-bound-direct}) 
and (\ref{eq:unital-bound-reverse}) hold if 
$t_\san{z} = t_\san{x} = 0$.
\end{proposition}
\begin{proof}
We only prove the statement for Eq.~(\ref{eq:unital-bound-direct})
because the statement for Eq.~(\ref{eq:unital-bound-reverse}) is
proved exactly in a similar manner.
By Proposition \ref{proposition:minimization},
it suffice to consider the Choi operator $\varrho$
of the form
\begin{eqnarray*}
\left(
\left[ \begin{array}{ccc}
R_{\mathsf{z}\mathsf{z}} & R_{\mathsf{z}\mathsf{x}} & 0 \\
R_{\mathsf{x}\mathsf{z}} & R_{\mathsf{x}\mathsf{x}} & 0 \\
0 & 0 & R_{\mathsf{y}\mathsf{y}}
\end{array}
\right],
\left[ \begin{array}{c}
t_{\mathsf{z}} \\ t_{\mathsf{x}} \\ 0
\end{array} \right]
\right).
\end{eqnarray*}
Define another Choi operator 
$\varrho^- := (\bar{\sigma}_\san{y} \otimes \sigma_\san{y})\varrho
 (\bar{\sigma}_\san{y} \otimes \sigma_\san{y})$ and the mixed
one $\varrho^\prime := \frac{1}{2}\varrho + \frac{1}{2} \varrho^-$.
Since the Stokes parameterization of $\varrho^-$ is 
\begin{eqnarray*}
\left(
\left[ \begin{array}{ccc}
R_{\mathsf{z}\mathsf{z}} & R_{\mathsf{z}\mathsf{x}} & 0 \\
R_{\mathsf{x}\mathsf{z}} & R_{\mathsf{x}\mathsf{x}} & 0 \\
0 & 0 & R_{\mathsf{y}\mathsf{y}}
\end{array}
\right],
\left[ \begin{array}{c}
- t_{\mathsf{z}} \\ - t_{\mathsf{x}} \\ 0
\end{array} \right]
\right),
\end{eqnarray*}
the vector part (of the Stokes parameterization) of $\varrho^\prime$ is zero vector, 
and the matrix part (of the Stokes parameterization) of $\varrho^\prime$ is the same as
that of $\varrho$.  Furthermore, since $H_\varrho(X|E) = H_{\varrho^-}(X|E)$,
by using Proposition \ref{proposition:convexity}, we have
\begin{eqnarray*}
H_\varrho(X|E) \ge H_{\varrho^\prime}(X|E).
\end{eqnarray*}
The equality holds if $t_\san{z} = t_\san{x} = 0$.

The rest of the proof is to calculate the minimization  of
$H_{\varrho^\prime}(X|E)$ with respect to $R_\san{yy}$.
By the singular value decomposition, we can decompose
the matrix $R^\prime$ corresponding to the Choi operator
 $\varrho^\prime$ as
\begin{eqnarray*}
O_2 \left[ \begin{array}{ccc}
\tilde{d}_\san{z} & 0 & 0 \\
0 & \tilde{d}_\san{x} & 0 \\
0 & 0 & R_\san{yy} 
\end{array} \right] O_1,
\end{eqnarray*}
where $O_1$ and $O_2$ are some rotation matrices within $\san{z}$-$\san{x}$-plane, and 
$|\tilde{d}_\san{z}|$ and $|\tilde{d}_\san{x}|$ are the 
singular value of the matrix in Eq.~(\ref{eq:two-times-two-matrix}).
Then, we have
\begin{eqnarray*}
\lefteqn{ \min_{R_\san{yy}} H_{\varrho^\prime}(X|E) } \\
&=& \min_{R_\san{yy}} \left[
1 - H(\varrho^\prime) + \sum_{x \in \mathbb{F}_2} \frac{1}{2} 
H(\varrho_B^{\prime x}) \right] \\
&=& 1 - \max_{R_\san{yy}} H[q_\san{i}, q_\san{z}, q_\san{x}, q_\san{y}] 
 + h\left(\frac{1 + \sqrt{R_\san{zz}^2 + R_\san{xz}^2}}{2} \right) \\
&=& 1 - h(q_\san{i} + q_\san{z}) - h(q_\san{i} + q_\san{x}) 
 + h\left(\frac{1 + \sqrt{R_\san{zz}^2 + R_\san{xz}^2}}{2} \right),
\end{eqnarray*}
where $(q_\san{i}, q_\san{z}, q_\san{x}, q_\san{y})$ are
the eigenvalues of the Choi operator $\varrho^\prime$,
and $\varrho_B^{\prime x} := 2 \rom{Tr}_A[(\ket{x}\bra{x} \otimes
 I)\varrho^\prime]$.
Note that we used Eq.~(\ref{eq:entropy-of-bloch}) 
to calculate the von Neumann entropy $H(\varrho_B^{\prime x})$.
By noting that $q_\san{i} + q_\san{z} = \frac{1 + \tilde{d}_\san{z}}{2}$
and $q_\san{i} + q_\san{x} = \frac{1 + \tilde{d}_\san{x}}{2}$
(see Eqs.~(\ref{eq-svd}) and (\ref{eq-relation-q-d})), we
have the statement for Eq.~(\ref{eq:unital-bound-direct}).
\end{proof}


The following lemma shows that the function
\begin{eqnarray}
G(\omega) := \min_{\varrho \in {\cal P}_c(\omega)} H_{\varrho}(X|E)
\end{eqnarray}
is a continuous function of $\omega$, which we suspended
in Section \ref{subsec:key-generation-rate}.

\begin{lemma}
\label{lemma:continuity-of-min}
The function $G(\omega)$ is a continuous
function of $\omega$ 
(with respect to the Euclidean distance) for any $\omega \in \Omega$.
\end{lemma}
\begin{proof}
Owing to Proposition \ref{proposition:minimization},
we have
\begin{eqnarray*}
G(\omega) = \min_{R_\san{yy} \in {\cal P}_c^\prime(\omega)}
H_\varrho(X|E),
\end{eqnarray*}
where $\varrho = (\omega,0,0,0,0,R_\san{yy},0)$
and ${\cal P}_c^\prime(\omega)$ is the set of 
all $R_\san{yy}$ such that 
$(\omega,0,0,0,0,R_\san{yy},0) \in {\cal P}_c(\omega)$.

Since the conditional entropy is a continuous function,
the following statement is suffice for proving that
$G(\omega)$ is continuous function at any $\omega_0 \in \Omega$.
For any $\omega \in \Omega$ such that $\| \omega - \omega_0 \| \le
 \varepsilon$,
there exist $\varepsilon^\prime, \varepsilon^{\prime\prime} >0$ such that
\begin{eqnarray}
\label{eq-neighbor-1}
{\cal P}_c^\prime(\omega) &\subset& {\cal B}_{\varepsilon^\prime}({\cal P}^\prime_c(\omega_0)), \\
\label{eq-neighbor-2}
{\cal P}^\prime_c(\omega_0) &\subset& {\cal B}_{\varepsilon^{\prime\prime}}({\cal P}^\prime_c(\omega)),
\end{eqnarray}
and $\varepsilon^\prime$ and $\varepsilon^{\prime\prime}$ converge to
$0$ as $\varepsilon$ goes to $0$,
where ${\cal B}_{\varepsilon^\prime}({\cal P}^\prime_c(\omega_0))$ is the
$\varepsilon^\prime$-neighbor of the set ${\cal P}^\prime_c(\omega_0)$.

Define the set ${\cal P}_c^{\prime\prime} := \{(\omega,R_\san{yy}) \mymid
 \omega \in \Omega, R_\san{yy} \in {\cal P}_c^\prime(\omega) \}$,
which is a closed convex set.
Define functions 
\begin{eqnarray*}
U(\omega) &:=& \max_{R_\san{yy} \in {\cal P}_c^\prime(\omega)} R_\san{yy}, \\
L(\omega) &:=& \min_{R_\san{yy} \in {\cal P}_c^\prime(\omega)} R_\san{yy}
\end{eqnarray*}
as the upper surface and the lower surface of the set 
${\cal P}_c^{\prime\prime}$ respectively.
Then $U(\omega)$ and $L(\omega)$ are concave and convex functions respectively, 
because ${\cal P}_c^{\prime\prime}$ is a convex set.
Thus, $U(\omega)$ and $L(\omega)$ are continuous functions except
the extreme points of $\Omega$.
For any extreme point $\omega^\prime$ of $\Omega$ and for any interior point
 $\omega$ of $\Omega$,
we have $U(\omega) \ge U(\omega^\prime)$ and 
$L(\omega) \le L(\omega^\prime)$, because
${\cal P}_c^{\prime\prime}$ is a convex set.
Since ${\cal P}_c^{\prime\prime}$ is a closed set, we have
$\lim_{\omega \to \omega^\prime} U(\omega) \in {\cal P}_c^\prime(\omega^\prime)$
and $\lim_{\omega \to \omega^\prime} L(\omega) \in {\cal P}_c^\prime(\omega^\prime)$,
which implies that $U(\omega^\prime) = \lim_{\omega \to \omega^\prime}
 U(\omega)$ 
and $L(\omega^\prime) = \lim_{\omega \to \omega^\prime}
 L(\omega)$.
Thus $U(\omega)$ and $L(\omega)$ are also continuous at the extreme
 points.
Since ${\cal P}_c^\prime(\omega)$ is a convex set, the continuity
of $U(\omega)$ and $L(\omega)$ implies that
Eqs.~(\ref{eq-neighbor-1}) and (\ref{eq-neighbor-2})
hold for some $\varepsilon^\prime, \varepsilon^{\prime\prime} >0$,
and $\varepsilon^\prime$ and $\varepsilon^{\prime\prime}$ converge
to $0$ as $\varepsilon$ goes to $0$.
\end{proof}

\section{Comparison to Conventional Estimation}
\label{sec:relation-to-conventional}

In this section, we show the conventional channel 
estimation procedure, and the asymptotic key 
generation rate formulas with the conventional 
channel estimation.
Then, we show that the
asymptotic key generation rates with our
proposed channel estimation are at least as
high as those with the conventional channel estimation
for the six-state protocol (Theorem \ref{theorem:at-least-six-state})
and the BB84 protocol (Theorem \ref{theorem:at-least-bb84})
respectively.

In the conventional channel estimation procedure,
Alice and Bob discard those bits if their bases disagree.
Furthermore, they ignore the difference between
$(x,y) = (0,1)$ and $(x,y) = (1,0)$.
Mathematically, these discarding and ignoring
can be described by 
a function $g:{\cal Z} \to \tilde{\cal Z} := \tilde{\mathbb{F}}_2 \times
{\cal J} \times {\cal J}$ defined by
\begin{eqnarray*}
g(z) = g((x,\san{a},y,\san{b})) := \left\{
\begin{array}{ll}
(x + y,\san{a},\san{b}) & \mbox{if } \san{a} = \san{b} \\
(\Delta, \san{a},\san{b}) & \mbox{else}
\end{array}
\right.,
\end{eqnarray*}
where $\tilde{\mathbb{F}}_2 := \mathbb{F}_2 \cup \{ \Delta\}$ 
and $\Delta$ is a dummy symbol indicating that
Alice and Bob discarded that sample bit.

\subsection{Six-State Protocol}

In the conventional estimation, 
Alice and Bob estimate 
$\rho \in {\cal P}_c$ from
the degraded sample sequence
$g(\bol{z}) := (g(z_1), \ldots,g(z_m))$.
Although the Choi operator $\rho$ is described by $12$ real
parameters (in the Stokes parameterization), from 
Eqs.~(\ref{eq:def-stokes-R}) and (\ref{eq:def-stokes-t}), we find that
the distribution 
\begin{eqnarray*}
\tilde{P}_\rho(\tilde{z}) = P_\rho(\{z \in {\cal Z} \mymid
 g(z) = \tilde{z} \})
\end{eqnarray*}
of the degraded sample symbol $\tilde{z} \in \tilde{{\cal Z}}$
only depends on the parameters $\gamma = (R_\san{zz}, R_\san{xx}, R_\san{yy})$,
and does not depend on
the parameters 
$\kappa = (R_\san{zx},R_\san{zy},R_\san{xz},R_\san{xy}, R_\san{yz}, R_\san{yx}, t_\san{z}, t_\san{x}, t_\san{y})$.
Therefore, we regard the set 
\begin{eqnarray*}
\Gamma := \{ \gamma \in \mathbb{R}^3 \mymid
\exists \kappa \in \mathbb{R}^9 ~ (\gamma, \kappa) \in {\cal P}_c \}
\end{eqnarray*}
as the parameter space, and denote $\tilde{P}_\rho$ by $\tilde{P}_{\gamma}$.
Then, we estimate the parameters $\gamma$ by the ML estimator:
\begin{eqnarray*}
\hat{\gamma}(\tilde{\bol{z}}) := \argmax_{\gamma \in \Gamma}
\tilde{P}_\gamma^m(\tilde{\bol{z}})
\end{eqnarray*}
for $\tilde{\bol{z}} \in \tilde{{\cal Z}}^m$.

Since we cannot estimate the parameters $\kappa$, we have
to consider the worst case, and estimate the quantity
\begin{eqnarray*}
\min_{\varrho \in {\cal P}_c(\gamma)} H_{\varrho}(X|E)
\end{eqnarray*}
for a given $\gamma \in \Gamma$, where the set 
\begin{eqnarray*}
{\cal P}_c(\gamma) := \{ \varrho = (\gamma^\prime, \kappa^\prime) \in
{\cal P}_c \mymid \gamma^\prime = \gamma \}
\end{eqnarray*}
is the candidates of Choi operators for a given $\gamma \in \Gamma$.

By following similar arguments as in 
Sections \ref{subsec:key-generation-rate},
\ref{subsec:sufficient-codition}, and \ref{subsec:six-state},
we can derive the asymptotic key generation rate formula
of the postprocessing with
the direct reconciliation 
\begin{eqnarray}
\label{eq:conventional-one-way-six}
\min_{\varrho \in {\cal P}_c(\gamma)} [
H_\varrho(X|E) - H_{\varrho}(X|Y) ].
\end{eqnarray}
We can also derive the asymptotic key generation
rate formula of the postprocessing with the reverse 
reconciliation 
\begin{eqnarray}
\label{eq:conventional-one-way-six-reverse}
\min_{\varrho \in {\cal P}_c(\gamma)} [
H_\varrho(Y|E) - H_\varrho(Y|X) ].
\end{eqnarray}

Since Eqs.~(\ref{eq:conventional-one-way-six}) 
and (\ref{eq:conventional-one-way-six-reverse}) involves
the minimizations, we have the following 
straight forward but important
theorem.
\begin{theorem}
\label{theorem:at-least-six-state}
The asymptotic key generation rates for the direct and the reverse
reconciliation 
with our proposed channel
estimation procedure (Eqs.~(\ref{eq:asymptotic-direct}) and
(\ref{eq:asymptotic-reverse}))
are at least as high as those with
the conventional channel estimation procedure
(Eqs.~(\ref{eq:conventional-one-way-six}) 
and (\ref{eq:conventional-one-way-six-reverse})) respectively.
\end{theorem} 

The following proposition gives an explicit expression of
Eqs.~(\ref{eq:conventional-one-way-six}) and
(\ref{eq:conventional-one-way-six-reverse}) 
for any Choi operator.
The following proposition also clarifies that the asymptotic key
generation rates of the direct and the reverse reconciliation
coincide for any Choi operator if we use
the conventional channel estimation procedure.
Although the following proposition is implicitly stated in
the literatures \cite{renner:05, renner:05b, kraus:05},
we present it for readers' convenience.
\begin{proposition}
\label{proposition:explicit-formula-convention}
For any $\rho = (\gamma, \tau) \in {\cal P}_c$, we have
\begin{eqnarray}
\label{eq:partial-twirled-rate1}
\lefteqn{
\min_{\varrho \in {\cal P}_c(\gamma)} [
H_\varrho(X|E) - H_{\varrho}(X|Y) ] } \\
\label{eq:partial-twirled-rate2}
 &=& \min_{\varrho \in {\cal P}_c(\gamma)} [
H_\varrho(Y|E) - H_{\varrho}(Y|X) ] \\
\label{eq:partial-twirled-rate3}
 &=& 1 - H[p_\san{i}, p_\san{z}, p_\san{x}, p_\san{y}],
\end{eqnarray}
where the distribution 
$(p_\san{i}, p_\san{z}, p_\san{x}, p_\san{y})$ is given by
\begin{eqnarray*}
p_\san{i} &=& \frac{1 + R_\san{zz} + R_\san{xx} + R_\san{yy}}{4}, \\
p_\san{z} &=& \frac{1 + R_\san{zz} - R_\san{xx} - R_\san{yy}}{4}, \\
p_\san{x} &=& \frac{1 - R_\san{zz} + R_\san{xx} - R_\san{yy}}{4}, \\
p_\san{y} &=& \frac{1 - R_\san{zz} - R_\san{xx} + R_\san{yy}}{4}.
\end{eqnarray*}
\end{proposition}
\begin{proof}
We only prove the equality between Eqs.~(\ref{eq:partial-twirled-rate1}) and
(\ref{eq:partial-twirled-rate3}), because the equality between
Eqs.~(\ref{eq:partial-twirled-rate2}) and
(\ref{eq:partial-twirled-rate3}) can be proved exactly in the same manner.

For any $\varrho \in {\cal P}_c(\gamma)$, let 
$\varrho^\san{z} := (\sigma_\san{z} \otimes \sigma_\san{z}) \varrho
 (\sigma_\san{z} \otimes \sigma_\san{z})$,
$\varrho^\san{x} := (\sigma_\san{x} \otimes \sigma_\san{x}) \varrho
 (\sigma_\san{x} \otimes \sigma_\san{x})$,
and $\varrho^\san{y} := (\sigma_\san{y} \otimes \sigma_\san{y}) \varrho
 (\sigma_\san{y} \otimes \sigma_\san{y})$.
Then, $\varrho^\san{z}$, $\varrho^\san{x}$, and $\varrho^\san{y}$
also belong to the set ${\cal P}_c(\gamma)$.
Define the (partial) twirled\footnote{The (partial) twirling was a
 technique to convert any bipartite density operator into
the  Bell diagonal state (see Section \ref{subsec:pauli-channel} for the
 definition of the Bell diagonal state). The (partial) twirling 
was first proposed by Bennett {\em et al}.~\cite{bennett:96b}.} Choi operator
\begin{eqnarray*}
\varrho^{tw} := \frac{1}{4} \varrho + \frac{1}{4} \varrho^\san{z}
  + \frac{1}{4} \varrho^\san{x} + \frac{1}{4} \varrho^\san{y}. 
\end{eqnarray*}
Then, the convexity of ${\cal P}_c(\gamma)$ implies 
$\varrho^{tw} \in {\cal P}_c(\gamma)$, 
and we can also find that the vector components 
(in the Stokes parameterization)
of $\varrho^{tw}$ is the zero vector and the matrix components
(in the Stokes parameterization)
of $\varrho^{tw}$ is the diagonal matrix with the diagonal
entries $R_\san{zz}$, $R_\san{xx}$, and $R_\san{yy}$.
Furthermore, we find that $\varrho^{tw} = \rho^{tw}$
for any $\varrho \in {\cal P}_c(\gamma)$. 

By using Proposition \ref{proposition:convexity} (twice), 
we have
\begin{eqnarray}
\lefteqn{ \min_{\varrho \in {\cal P}_c(\gamma)} [
H_\varrho(X|E) - H_{\varrho}(X|Y) ] } \nonumber \\
&\ge& H_{\rho^{tw}}(X|E) \nonumber \\
&=& 1 - H(\varrho^{tw}) + \sum_{x \in \mathbb{F}_2} \frac{1}{2} 
H(\varrho_B^{tw x}) \nonumber \\
&=& 1 - H[q_\san{i}, q_\san{z}, q_\san{x}, q_\san{y}] 
 + h\left(\frac{1 + R_\san{zz}}{2} \right),
\label{eq:partial-twirled-rate-proof-1}
\end{eqnarray}
where $\varrho_B^{tw x} := 2 \rom{Tr}_A[(\ket{x}\bra{x} \otimes
 I)\varrho^{tw}]$.

In a similar manner as in Remark \ref{remark:ir-with-code}, we have
\begin{eqnarray}
H_{\varrho}(X|Y) \le H_{\varrho}(W) = H_{\rho^{tw}}(W) 
= h\left(\frac{1 + R_\san{zz}}{2} \right)
\label{eq:partial-twirled-rate-proof-2}
\end{eqnarray}
for any $\varrho \in {\cal P}_c(\gamma)$,
where $H_{\varrho}(W)$ is the entropy of the random variable $W$
whose distribution is 
\begin{eqnarray*}
P_{W,\varrho}(w) := \sum_{y \in \mathbb{F}_2} 
P_{XY, \varrho}(y+w,y).
\end{eqnarray*}

Combining Eqs.~(\ref{eq:partial-twirled-rate-proof-1})
and (\ref{eq:partial-twirled-rate-proof-2}), we have
the equality between Eqs.~(\ref{eq:partial-twirled-rate1}) and
(\ref{eq:partial-twirled-rate3}).
\end{proof}

\begin{remark}
\label{remark:no-difference-sw-er}
As we can find in the proof of 
Proposition \ref{proposition:explicit-formula-convention},
the use of the IR procedure (with the linear Slepian-Wolf coding)
proposed in Section \ref{sec:one-way-IR} and the use of the IR procedure
(with the error correcting code) presented in 
Remark \ref{remark:ir-with-code} make no difference
to the asymptotic key generation rate if we use the conventional
channel estimation procedure.
\end{remark} 
\begin{remark}
It should be noted that Eq.~(\ref{eq:partial-twirled-rate3}) 
is the well known asymptotic key generation rate formula
\cite{lo:01}, which can be derived by using the technique based on the
CSS code (See Section \ref{sec:background} for the CSS code technique). 
\end{remark}

\subsection{BB84 Protocol}

In the conventional estimation,
Alice and Bob estimate 
$\rho \in {\cal P}_c$ from
the degraded sample sequence
$g(\bol{z}) := (g(z_1), \ldots,g(z_m))$.
Although the Choi operator $\rho$ is described by $12$ real
parameters (in the Stokes parameterization), from
Eqs.~(\ref{eq:def-stokes-R}) and (\ref{eq:def-stokes-t}), we find that
the distribution 
\begin{eqnarray*}
\tilde{P}_\omega(\tilde{z}) = P_\omega(\{z \in {\cal Z} \mymid
 g(z) = \tilde{z} \})
\end{eqnarray*}
of the degraded sample symbol $\tilde{z} \in \tilde{{\cal Z}}$
only depends on the parameters $\upsilon = (R_\san{zz}, R_\san{xx})$,
and does not depend on the parameters
$\varsigma = (R_\san{zx},R_\san{zy},R_\san{xz},R_\san{xy}, R_\san{yz}, R_\san{yx}, R_\san{yy}, t_\san{z}, t_\san{x}, t_\san{y})$.
Therefore, we regard the set
\begin{eqnarray*}
\Upsilon := \{ \upsilon \in \mathbb{R}^2 \mymid
\exists \varsigma \in \mathbb{R}^{10}, ~ (\upsilon,\varsigma) \in {\cal P}_c \}
\end{eqnarray*}
as the parameter space, and denote $\tilde{P}_\omega$ by
$\tilde{P}_{\upsilon}$. 
Then, we estimate the parameters $\upsilon$ by the ML estimator:
\begin{eqnarray*}
\hat{\upsilon}(\tilde{\bol{z}}) := \argmax_{\upsilon \in \Upsilon}
\tilde{P}_\upsilon^m(\tilde{\bol{z}})
\end{eqnarray*}
for $\tilde{\bol{z}} \in \tilde{{\cal Z}}^m$.

Since we cannot estimate the parameters $\varsigma$, we have
to consider the worst case, and estimate the quantity
\begin{eqnarray*}
\min_{\varrho \in {\cal P}_c(\upsilon)} H_{\varrho}(X|E)
\end{eqnarray*}
for a given $\upsilon \in \upsilon$, where the set 
\begin{eqnarray*}
{\cal P}_c(\upsilon) := \{ \varrho = (\upsilon^\prime, \varsigma^\prime) \in
{\cal P}_c \mymid \upsilon^\prime = \upsilon \}
\end{eqnarray*}
is the candidates of Choi operators for a given $\upsilon \in \Upsilon$.

By following similar arguments as in 
Sections \ref{subsec:key-generation-rate},
\ref{subsec:sufficient-codition}, and \ref{subsec:bb84},
we can derive the asymptotic key generation rate formula
of the postprocessing with
the direct reconciliation 
\begin{eqnarray}
\label{eq:conventional-one-way-bb84}
\min_{\varrho \in {\cal P}_c(\upsilon)}[
 H_{\varrho}(X|E) - H_{\varrho}(X|Y)].
\end{eqnarray}
We can also derive the asymptotic key generation
rate formula of the postprocessing with the reverse 
reconciliation 
\begin{eqnarray}
\label{eq:conventional-one-way-bb84-reverse}
\min_{\varrho \in {\cal P}_c(\upsilon)}[
 H_{\varrho}(Y|E) - H_{\varrho}(Y|X)]. 
\end{eqnarray}

Since the range ${\cal P}_c(\omega)$ of the minimizations
in Eqs.~(\ref{eq:key-rate-one-way-bb84}) 
and (\ref{eq:eq:key-rate-one-way-bb84-reverse}) is smaller than
the range ${\cal P}_c(\upsilon)$ of the minimizations
in Eqs.~(\ref{eq:conventional-one-way-bb84}) and 
(\ref{eq:conventional-one-way-bb84-reverse}), we have the following obvious but important
theorem.
\begin{theorem}
\label{theorem:at-least-bb84}
The asymptotic key generation rates for the direct and the reverse
reconciliation 
with our proposed channel
estimation procedure (Eqs.~(\ref{eq:key-rate-one-way-bb84}) and
(\ref{eq:eq:key-rate-one-way-bb84-reverse}))
are at least as high as those with
the conventional channel estimation procedure
(Eqs.~(\ref{eq:conventional-one-way-bb84}) 
and (\ref{eq:conventional-one-way-bb84-reverse})) respectively.
\end{theorem}

The following proposition gives an explicit expression of
Eqs.~(\ref{eq:conventional-one-way-bb84}) 
and (\ref{eq:conventional-one-way-bb84-reverse}) 
for any Choi operator.
The following proposition also clarifies that the asymptotic key
generation rates of the direct and the reverse reconciliation
coincide for any Choi operator if we use
the conventional channel estimation procedure.
Although the following proposition is implicitly stated in
the literatures \cite{renner:05, renner:05b, kraus:05},
we present it for readers' convenience.
\begin{proposition}
\label{proposition:explicit-formula-convention-bb84}
For any $\rho = (\upsilon,\varsigma) \in {\cal P}_c$, we have
\begin{eqnarray}
\label{eq:partial-twirled-rate1-bb84}
\lefteqn{
\min_{\varrho \in {\cal P}_c(\upsilon)} [
H_\varrho(X|E) - H_{\varrho}(X|Y) ] } \\
\label{eq:partial-twirled-rate2-bb84}
 &=& \min_{\varrho \in {\cal P}_c(\upsilon)} [
H_\varrho(Y|E) - H_{\varrho}(Y|X) ] \\
\label{eq:partial-twirled-rate3-bb84}
 &=& 1 - h\left(\frac{1 + R_\san{zz}}{2} \right)
 - h\left(\frac{1 + R_\san{xx}}{2} \right).
\end{eqnarray}
\end{proposition}
\begin{proof}
This proposition is proved in a similar manner as 
Proposition \ref{proposition:explicit-formula-convention}.
Therefore, we omit the proof.
\end{proof}

\begin{remark}
It should be noted that the same remark 
as Remark \ref{remark:no-difference-sw-er} also holds
for the BB84 protocol.
\end{remark}
\begin{remark}
It should be noted that Eq.~(\ref{eq:partial-twirled-rate3-bb84}) 
is with the well known asymptotic key generation rate formula
\cite{shor:00}, which can be derived by using the technique based on the
CSS code (See Section \ref{sec:background} for the CSS code technique). 
\end{remark}

\section{Asymptotic Key Generation Rates for Specific Channels}
\label{sec:example}

In this section, we calculate the asymptotic key generation rates
of the BB84 protocol and the six-state protocol
for specific channels, and
clarify the advantage to use 
our proposed channel estimation instead of
the conventional channel estimation.

\subsection{Amplitude Damping Channel}
\label{subsec:amplitude-damping}

When the channel between Alice and Bob is an
amplitude damping channel, the Stokes parameterization
of the corresponding density operator 
$\rho_p \in {\cal P}_c$ is
\begin{eqnarray}
\label{eq:amplitude-damping}
\left(
\left[
\begin{array}{ccc}
1- p & 0 & 0 \\
0 & \sqrt{1-p} & 0 \\
0 & 0 & \sqrt{1-p}
\end{array}
\right],
\left[
\begin{array}{c}
p \\ 0 \\ 0
\end{array}
\right]
\right),
\end{eqnarray}
where $0 \le p \le 1$.

For the six-state protocol, 
since there are no minimization in Eqs.~(\ref{eq:asymptotic-direct})
and (\ref{eq:asymptotic-reverse}), there are no difficulty to calculate
Eqs.~(\ref{eq:asymptotic-direct})
and (\ref{eq:asymptotic-reverse}).

Next, we consider the BB84 protocol.
For $\omega = (1-p,0,0,\sqrt{1-p},p,0)$,
Eqs.~(\ref{eq:key-rate-one-way-bb84})
and (\ref{eq:eq:key-rate-one-way-bb84-reverse})
can be calculated as follows.
By Proposition \ref{proposition:minimization},
it is sufficient to consider $\varrho \in {\cal P}_c(\omega)$
such that $R_\san{zy} = R_\san{xy} = R_\san{yz} = R_\san{yx} = t_\san{y}
= 0$.
Furthermore, by the condition on the TPCP map \cite{fujiwara:99}
\begin{eqnarray*}
(R_{\mathsf{xx}} - R_{\mathsf{yy}})^2 
\le (1 - R_{\mathsf{zz}})^2 - t_{\mathsf{z}}^2,
\end{eqnarray*}
we can decide the remaining parameter as $R_\san{yy} = \sqrt{1-p}$.
Therefore, Eqs.~(\ref{eq:key-rate-one-way-bb84})
and (\ref{eq:eq:key-rate-one-way-bb84-reverse}) coincide with
the true values respectively. Furthermore,
the asymptotic key generation rates for the 
BB84 protocol coincide with those for the
six-state protocol.

The asymptotic key generation rates for the
direct and the reverse reconciliations can
be written as functions of the
parameter $p$:
\begin{eqnarray}
h\left( \frac{1+p}{2}\right) - h\left( \frac{p}{2} \right)
\label{eq-key-rate-forward} 
\end{eqnarray}
and
\begin{eqnarray}
1 - h\left( \frac{p}{2} \right)
\label{eq-key-rate-reverse}
\end{eqnarray}
respectively. They are plotted in Fig.~\ref{fig:amplitude-damping}.

From Fig.~\ref{fig:amplitude-damping}, we find that
the asymptotic key generation rate
with the reverse reconciliation is higher than
that with the forward reconciliation. Actually, 
they are analyzed in detail as follows.
By a straightforward calculation, we have
\begin{eqnarray*}
H_\rho(X|E) &=& 1 + \frac{1}{2}h\left(p\right) -
 h\left(\frac{p}{2}\right)  \\
&=& H_\rho(XY) - h\left(\frac{p}{2}\right)
\end{eqnarray*} 
and
\begin{eqnarray*}
H_\rho(Y|E) &=& h\left(\frac{1 + p}{2}\right) + 
 \frac{1 + p}{2}h\left(\frac{1}{1 + p}\right) 
 - h\left(\frac{p}{2}\right) \\
 &=& H_\rho(XY) - h\left(\frac{p}{2}\right),
\end{eqnarray*}
where $H_\rho(XY)$ is the entropy of the random variables
with distribution $P_{XY,\rho}$.
Therefore, the difference between the asymptotic key
generation rate with the forward and the reverse reconciliations
comes from 
the difference between 
$H_\rho(X|Y)$ and $H_\rho(Y|X)$, which is equal to
the difference between
$H_\rho(Y)$ and $H_\rho(X) = 1$.
Note that $H_\rho(Y)$ goes to $0$ as 
$p \to 1$.

The Bell diagonal entries of the Choi operator
$\rho_p$ are
$\frac{1}{4}(2+2 \sqrt{1-p} - p)$,
$\frac{1}{4}p$, $\frac{1}{4}(2-2 \sqrt{1-p} - p)$,
and $\frac{1}{4}p$.
When Alice and Bob only use the degraded statistic,
i.e., when Alice and Bob use the conventional channel estimation,
the asymptotic key generation rates
of the six-state protocol and the BB84
protocol can be
calculated only from the Bell diagonal entries
(Propositions \ref{proposition:explicit-formula-convention}
and \ref{proposition:explicit-formula-convention-bb84}), and 
are also plotted in Fig.~\ref{fig:amplitude-damping}.

\begin{figure}
\centering
\includegraphics[width=\linewidth]{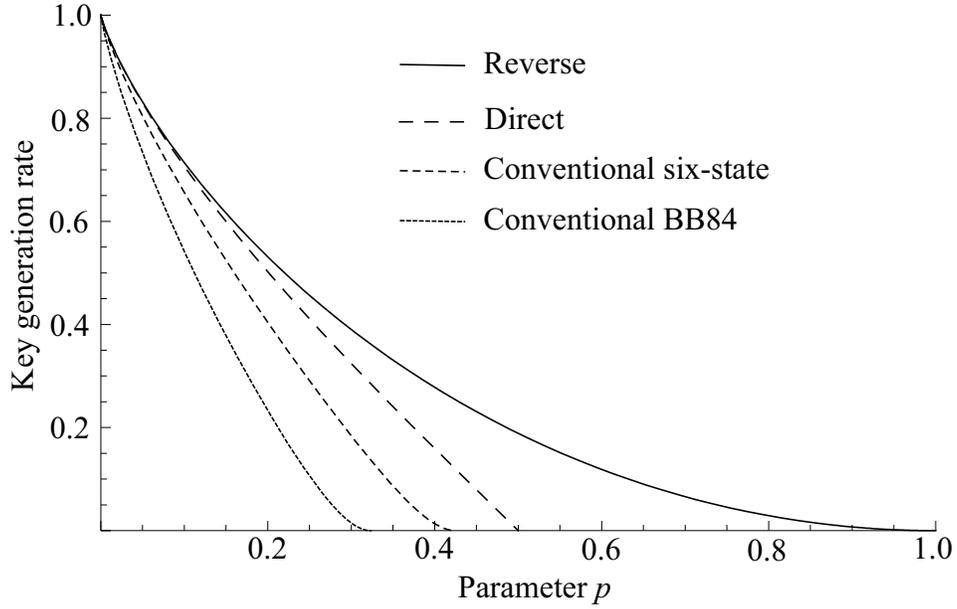}
\caption{ 
Comparison of the asymptotic key generation rates against the parameter
$p$ of the amplitude damping channel (see Eq.~(\ref{eq:amplitude-damping})).
``Reverse'' and ``Direct'' are the asymptotic key generation rates when we use the
reverse reconciliation and the direct reconciliation
with our channel estimation procedure (Eqs.~(\ref{eq-key-rate-reverse}) and
(\ref{eq-key-rate-forward})) respectively.
``Conventional six-state'' and ``Conventional BB84'' are the asymptotic key
generation rates of the six-state protocol and the BB84 protocol
with the conventional channel estimation procedure.
Note that the protocols with
the conventional channel estimation procedure 
involves the noisy preprocessing \cite{renner:05, kraus:05}
in the postprocessing.
}
\label{fig:amplitude-damping}
\end{figure}

\begin{remark}
\label{remark:degradable}
As is mentioned in Remark \ref{remark:noisy-preprocessing},
there is a possibility to improve the asymptotic key 
generation rate in Eq.~(\ref{eq:asymptotic-direct})
by the noisy preprocessing. If a $\{ccq\}$-state
$\rho_{XYE}$ derived from 
a Choi operator $\rho \in {\cal P}_c$ satisfies the condition below, we can show that
the noisy preprocessing does not improve the asymptotic
key generation rate.

We define a $\{ccq\}$-state
\begin{eqnarray*}
\rho_{XYE} = \sum_{x,y \in \mathbb{F}_2} P_{XY}(x,y) \ket{x,y}\bra{x,y}
\otimes \rho_E^{x,y}
\end{eqnarray*}
to be degradable state\footnote{The concept of the degradable
state is an analogy of the degradable channel \cite{devetak:05b}.
For the degradable channel, the quantum wire-tap channel
capacity \cite{devetak:05} is known to be achievable without 
any auxiliary random variable \cite{smith:07,hayashi-book:06}.}
(from Alice to Bob and Eve) 
if there exist states $\{ \hat{\rho}_E^y\}_{y \in \mathbb{F}_2}$
satisfying 
\begin{eqnarray*}
\sum_{y \in \mathbb{F}_2} P_{Y|X}(y|x) \hat{\rho}_E^y = \rho_E^x := 
  \sum_{y \in \mathbb{F}_2} P_{Y|X}(y|x) \rho_E^{x,y}
\end{eqnarray*}
for any $x \in \mathbb{F}_2$. 
If a $\{ccq\}$-state
$\rho_{XYE}$ derived from a Choi operator $\rho$ is degradable,
then the asymptotic key generation rate in Eq.~(\ref{eq:asymptotic-direct})
is optimal, that is,  it cannot be improved by the noisy preprocessing.

The above statement is proved as follows. 
Even if we know the Choi operator $\rho$ in advance,
the asymptotic key generation rate
of any postprocessing is upper bounded by
the quantum intrinsic information\footnote{It is the quantum 
analogy of the intrinsic information proposed by
Maurer and Wolf \cite{maurer:99}.}
\begin{eqnarray*}
I_\rho(X;Y \downarrow E) := \inf I_\rho(X;Y | E^\prime),
\end{eqnarray*} 
where 
\begin{eqnarray*}
I_\rho(X;Y | E^\prime) := H_\rho(XE) + H_\rho(YE) - H_\rho(XYE) - H_\rho(E)
\end{eqnarray*}
is the quantum conditional mutual information,
and the infimum is taken over all $\{ccq\}$-states
$\rho_{XYE^\prime} = (\rom{id} \otimes \mathcal{N}_{E \to
 E^\prime})(\rho_{XYE})$ for CPTP maps $\mathcal{N}_{E \to E^\prime}$
from system $E$ to $E^\prime$ \cite{christandl:06}.
Taking the identity map $\rom{id}_E$, the quantum conditional mutual
information $I_\rho(X;Y|E)$ itself is an upper bound on
the asymptotic key generation rate for any postprocessing.

Since we are now considering the postprocessing in which 
only Alice sends the public message, the maximum of
the asymptotic key generation rate only depends on
the distribution $P_{XY}$ and $\{cq\}$-state $\rho_{XE}$.
Thus the maximum of the asymptotic key generation rate for
$\rho_{XYE}$ is equals to  that for degraded version of it,
\begin{eqnarray*}
\hat{\rho}_{XYE} := \sum_{x,y} P_{XY}(x,y) \ket{x}\bra{x} \otimes
 \ket{y}\bra{y} \otimes \hat{\rho}_E^{y}.
\end{eqnarray*}
Applying the above upper bound $I_\rho(X;Y|E)$ for
the degraded $\{ccq\}$-state $\hat{\rho}_{XYE}$,
the maximum of the asymptotic key generation rate is upper bounded
by 
\begin{eqnarray*}
\lefteqn{
I_{\hat{\rho}}(X;Y | E) 
} \\
&=& I_{\hat{\rho}}(X ; YE) - I_{\hat{\rho}}(X;E) \\
&=& H_{\hat{\rho}}(X|E) - H(X|Y) + I_{\hat{\rho}}(X;E | Y) \\
&=& H_{\rho}(X|E) - H(X|Y),
\end{eqnarray*}
which is the desired upper bound, and equals to
Eq.~(\ref{eq:asymptotic-direct}).

For the amplitude damping channel, we can show that the $\{ccq\}$-state
$\rho_{XYE}$ is degradable by a straightforward 
calculation. Therefore, the asymptotic key generation rate
in Eq.~(\ref{eq:asymptotic-direct}) is optimal for the
amplitude damping channel. 
%
\end{remark}

Although we exclusively considered a key 
generated from the bit sequences transmitted and received
by the $\san{z}$-basis, 
we can also obtain a key from the 
bit sequences transmitted and received by 
the $\san{x}$-basis (or the $\san{y}$-basis for the six-state protocol). 
In this case, the asymptotic key generation
rates are also given by Eqs.~(\ref{eq:asymptotic-direct}),
(\ref{eq:asymptotic-reverse}), (\ref{eq:key-rate-one-way-bb84}),
and (\ref{eq:eq:key-rate-one-way-bb84-reverse}), where
the definition of the $\{cq\}$-state $\rho_{XE}$ and
the distribution $P_{XY}$ must be replaced
appropriately.

For the amplitude damping channel\footnote{By the symmetry of
the amplitude damping channel for the $\san{x}$-basis and 
the $\san{y}$-basis, the asymptotic key generation rates 
for the $\san{y}$-basis are the same as those for
the $\san{x}$-basis}, 
the asymptotic key generation rates for the forward
and the reverse reconciliations can be written as functions
of the parameter $p$:
\begin{eqnarray}
1 + h\left(\frac{1 + \sqrt{1 - p + p^2}}{2}\right) - 
 h\left(\frac{p}{2}\right)
 - h\left(\frac{1 + \sqrt{1 - p}}{2}\right),
\label{eq-key-rate-forward-x-basis}
\end{eqnarray}
and
\begin{eqnarray}
1 - h\left(\frac{p}{2}\right)
\label{eq-key-rate-reverse-x-basis}
\end{eqnarray}
respectively. They are plotted in Fig.~\ref{fig:amplitude-damping-x-basis},
and compared to the asymptotic key generation rates 
with the conventional channel estimation.

From Fig.~\ref{fig:amplitude-damping-x-basis}, we find
that the asymptotic key generation rate with
the reverse reconciliation is higher than
that with the forward reconciliation.
Although the difference between the asymptotic key
generation rate with the forward and the reverse reconciliations
comes from the difference between 
$H_\rho(X|Y)$ and $H_\rho(Y|X)$ in the case of
the $\san{z}$-basis, the difference between the asymptotic key
generation rate with the forward and the reverse reconciliations
comes from the difference between $H_\rho(X|E)$ and
$H_\rho(Y|E)$, because 
$H_\rho(X|Y) = H_\rho(Y|X)$ in the case of
the $\san{x}$-basis.

\begin{figure}
\centering
\includegraphics[width=\linewidth]{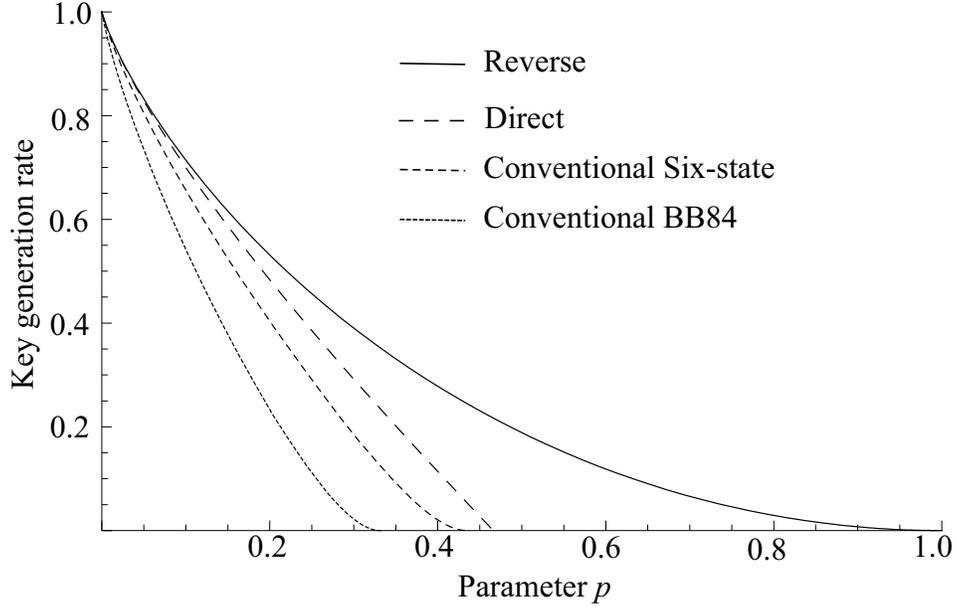}
\caption{ 
Comparison of the asymptotic key generation rates against the parameter
$p$ of the amplitude damping channel (see
 Eq.~(\ref{eq:amplitude-damping}))
for a key generated from the bit sequences transmitted and
received by the $\san{x}$-basis.
``Reverse'' and ``Direct'' are the asymptotic key generation rates when we use the
reverse reconciliation and the direct reconciliation
with our channel estimation procedure
 (Eqs.~(\ref{eq-key-rate-reverse-x-basis}) and
(\ref{eq-key-rate-forward-x-basis})) respectively.
``Conventional six-state'' and ``Conventional BB84'' are the asymptotic key
generation rates of the six-state protocol and the BB84 protocol
with the conventional channel estimation procedure.
Note that the protocols with
the conventional channel estimation procedure 
involves the noisy preprocessing \cite{renner:05, kraus:05}
in the postprocessing.
}
\label{fig:amplitude-damping-x-basis}
\end{figure}

\subsection{Unital Channel and Rotation Channel}
\label{subsec:unital-channel}

A channel is called a unital channel if 
${\cal E}_B$ maps the completely mixed state
$I/2$ to itself, or equivalently
the corresponding Choi operator 
$\rho \in {\cal P}_c$
satisfies $\rom{Tr}_A[\rho] = I/2$.
In the Stokes parameterization, a unital channel
$(R,t)$ satisfies that $t$ is the zero vector.
The unital channel has the following physical meaning
in QKD protocols. When Eve conducts the Pauli cloning \cite{cerf:00}
with respect to an orthonormal basis that is a rotated
version of $\{\ket{0_\san{z}}$, $\ket{1_\san{z}}\}$, the quantum channel
from Alice to Bob is not a Pauli channel but a unital
channel. It is natural to assume that Eve cannot determine
the direction of the basis $\{\ket{0_\san{z}}$, $\ket{1_\san{z}}\}$
accurately, and the unital channel deserve consideration
in the QKD research as well as the Pauli channel. 

By the singular value decomposition, we can 
decompose the matrix $R$ of the Stokes parameterization as
\begin{eqnarray}
\label{eq-svd}
O_2 \left[\begin{array}{ccc}
e_\san{z} & 0 & 0 \\
0 & e_\san{x} & 0 \\
0 & 0 & e_\san{y} 
\end{array} \right]
O_1,
\end{eqnarray}
where $O_1$ and $O_2$ are some rotation 
matrices\footnote{The rotation matrix is the real orthogonal
matrix with determinant $1$.}, and
$|e_\san{z}|$, $|e_\san{x}|$, and $|e_\san{y}|$
are the singular value of the matrix 
$R$\footnote{The decomposition is not unique 
because we can change the order of 
$(e_\san{z}, e_\san{x}, e_\san{y})$ or
the sign of them by adjusting the rotation matrices $O_1$ and $O_2$.
However, the result in this paper does not depends on a choice of the decomposition.}.  
Thus, we can consider the unital channel as a composition of a
unitary channel, a Pauli channel 
\begin{eqnarray*}
\rho \mapsto q_\san{i} \rho + q_\san{z} \sigma_\san{z} \rho \sigma_\san{z}
+ q_\san{x} \sigma_\san{x} \rho \sigma_\san{x} + q_\san{y} \sigma_\san{y} \rho \sigma_\san{y},
\end{eqnarray*}
and a unitary channel \cite{bourdon:04}, where 
\begin{eqnarray}
\begin{array}{rcl}
q_\san{i} &=& \frac{1 + e_\san{z} + e_\san{x} + e_\san{y}}{4}, \\
q_\san{z} &=& \frac{1 + e_\san{z} - e_\san{x} - e_\san{y}}{4}, \\
q_\san{x} &=& \frac{1 - e_\san{z} + e_\san{x} - e_\san{y}}{4}, \\
q_\san{y} &=& \frac{1 - e_\san{z} - e_\san{x} + e_\san{y}}{4}.
\end{array}
\label{eq-relation-q-d}
\end{eqnarray}

For the six-state protocol, 
we can derive simple forms of 
$H_\rho(X|E)$ and $H_\rho(Y|E)$ as follows.
\begin{lemma}
\label{lemma:unital-six-state} 
For the unital channel, we have
\begin{eqnarray}
H_\rho(X|E) = 1 - H[q_\san{i}, q_\san{z}, q_\san{x}, q_\san{y}] + 
h\left( \frac{1 + \sqrt{R_\san{zz}^2 + R_\san{xz}^2 + R_\san{yz}^2}}{2} \right)
\label{eq-unital-six-state}
\end{eqnarray}
and 
\begin{eqnarray}
H_\rho(Y|E) = 1 - H[q_\san{i}, q_\san{z}, q_\san{x}, q_\san{y}] + 
h\left( \frac{1 + \sqrt{R_\san{zz}^2 + R_\san{zx}^2 + R_\san{zy}^2}}{2} \right).
\label{eq-unital-six-state-reverse}
\end{eqnarray}
\end{lemma}
\begin{proof}
We omit the proof because it can be proved in a
similar manner as the latter half of the proof
of Proposition \ref{proposition:unital-bound}.
\end{proof}
From this lemma, we can find that
$R_\san{xz}^2 + R_\san{yz}^2 = R_\san{zx}^2 + R_\san{zy}^2$
is the necessary and sufficient condition for 
$H_\rho(X|E) = H_\rho(Y|E)$.
Furthermore, we can show
$H_\rho(X|Y) = H_\rho(Y|X) = h((1+ R_\san{zz})/2)$
by a straightforward calculation.

For the BB84 protocol, ${\cal P}_c(\omega)$ consists of
infinitely many elements in general.
By using Proposition \ref{proposition:unital-bound},
we can calculate Eve's worst case ambiguity as
\begin{eqnarray}
\lefteqn{ \min_{\varrho \in {\cal P}_c(\omega)}
H_{\varrho}(X|E) } \nonumber \\
&=& \hspace{-3mm}
1 - h\left(\frac{1 + d_\san{z}}{2}\right)
- h\left(\frac{1+d_\san{x}}{2}\right)
+ h\left( \frac{1+\sqrt{R_{\san{zz}}^2 + R_{\san{xz}}^2}}{2}\right)
\label{eq:asymptotic-key-rate-unital-bb84}
\end{eqnarray}
and 
\begin{eqnarray}
\lefteqn{ \min_{\varrho \in {\cal P}_c(\omega)}
H_{\varrho}(Y|E) } \nonumber \\  
&=& \hspace{-3mm}
1 - h\left(\frac{1 + d_\san{z}}{2}\right)
- h\left(\frac{1+d_\san{x}}{2}\right)
+ h\left( \frac{1+\sqrt{R_{\san{zz}}^2 + R_{\san{zx}}^2}}{2}\right),
\label{eq:asymptotic-key-rate-unital-bb84-reverse}
\end{eqnarray}
where $d_\san{z}$ and $d_\san{x}$ are the singular values of
the matrix $\left[\begin{array}{cc} R_{\san{zz}} & R_{\san{zx}} \\ R_\san{xz} &
  R_\san{xx} \end{array}\right]$.
From these formulae, we find that
$R_\san{xz} = R_\san{zx}$ is the necessary and  sufficient
condition for 
$\min_{\varrho \in {\cal P}_c(\omega)} H_{\varrho}(X|E)$
coincides with
$\min_{\varrho \in {\cal P}_c(\omega)} H_{\varrho}(Y|E)$.
It should be noted that the singular values $(d_\san{z}, d_\san{x})$
are different from the singular values
$(|e_\san{z}|, |e_\san{x}|)$ in general because there exist
off-diagonal elements $(R_\san{zy}, R_\san{xy}, R_\san{yz},
R_\san{yx})$. By a straightforward calculation, we can show that
$H_\omega(X|Y) = H_\omega(Y|X) = h((1+
R_\san{zz})/2)$.

In the rest of this section, we analyze a special class
of the unital channel, the rotation channel,
for the BB84 protocol.
The rotation channel is a channel whose Stokes parameterization
is given by
\begin{eqnarray*}
\left(
\left[ \begin{array}{ccc}
\cos \vartheta & - \sin \vartheta & 0 \\
\sin \vartheta & \cos \vartheta & 0 \\
0 & 0 & 1
\end{array}
\right],
\left[
\begin{array}{c}
0 \\ 0 \\ 0
\end{array}
\right]
\right).
\end{eqnarray*}
The rotation channels occur, for example, when
the directions of 
the transmitter and the receiver
are not properly aligned.

For the rotation channel, 
Eq.~(\ref{eq:asymptotic-key-rate-unital-bb84})
gives 
$\min_{\varrho \in {\cal P}_c(\omega)} H_{\varrho}(X|E) = 1$,
which implies that Eve gained no information.
Thus, Eve's worst case ambiguity,
$\min_{\varrho \in {\cal P}_c(\omega)} H_{\varrho}(X|E)$
coincide with the true value $H_\rho(X|E)$, and the 
BB84 protocol can achieve the same asymptotic key
generation rate as the six-state protocol.

The reason why we show this example is that
Alice and Bob can share a secret key with a positive 
asymptotic key generation rate 
even though the so-called error rate is higher
than the $25$\% limit \cite{gottesman:03} in the BB84 protocol.
The Bell diagonal entries of
the Choi operator $\rho$ that corresponds to the rotation 
channel are $\cos^2 (\vartheta/2)$, $0$, $0$,
and $\sin^2(\vartheta/2)$. Thus the error rate is $\sin^2(\vartheta/2)$.
For $\pi /3 \le \vartheta \le 5 \pi/3$,
the error rate is higher than $25$\%, but
we can obtain the positive key rate, $1 - h(\sin^2 (\vartheta/2))$
except $\vartheta = \pi/2, 3 \pi /2$.
Note that the asymptotic key generation rate
in Eq.~(\ref{eq:conventional-one-way-bb84}) is given by
$1 - 2 h(\sin^2 (\vartheta/2))$.
This fact verifies Curty et al's suggestion \cite{curty:04} that
key agreement might be possible even for the error rates
higher than $25$\% limits.

\section{Condition for Strict Improvement}
\label{sec:condition-strict-improvement}

So far, we have seen that 
the asymptotic key generation rates with our
proposed channel estimation is at least as
high as those with the conventional channel estimation
(Section \ref{sec:relation-to-conventional}), and that
the former is strictly higher than the latter for 
some specific channels (Section \ref{sec:example}).
For the BB84 protocol,
the following theorems show the 
necessary and sufficient condition such that
the former is strictly higher than the latter
is that the channel is a Pauli channel.
\begin{theorem}
\label{theorem:strict-improve-bb84}
Suppose that $R_\san{zz} \neq 0$ and $R_\san{xx} \neq 0$.
In the BB84 protocol,
for the bit sequences transmitted and received by either
$\san{z}$-basis or the $\san{x}$-basis, 
the asymptotic key generation rates 
with our proposed channel estimation
are strictly higher than those with the conventional channel
estimation if and only if 
$(t_\san{z}, t_\san{x}) \neq (0,0)$ or $(R_\san{zx},R_\san{xz}) \neq (0,0)$.
\end{theorem}
\begin{proof}
We only prove the statement for the direct reconciliation,
because the statement for the reverse reconciliation
can be proved in a similar manner.
\paragraph{``only if" part}
Suppose that $(t_\san{z}, t_\san{x}) = (0,0)$
and $(R_\san{zx},R_\san{xz}) = (0,0)$.
Then, Propositions \ref{proposition:unital-bound} and
\ref{proposition:explicit-formula-convention-bb84} implies that
Eq.~(\ref{eq:key-rate-one-way-bb84}) is 
equal to Eq.~(\ref{eq:conventional-one-way-bb84}).
Similarly, the asymptotic key generation rate
for the $\san{x}$-basis with our proposed channel estimation
is equal to that with the conventional channel estimation.
\paragraph{``if" part}
Suppose that $t_\san{z} \neq 0$. 
Let $\varrho^*$ be the Choi operator satisfying
\begin{eqnarray*}
H_{\varrho^*}(X|E) - H_{\varrho^*}(X|Y) = \min_{\varrho \in {\cal P}_c(\upsilon)}[
H_\varrho(X|E) - H_\varrho(X|Y)].
\end{eqnarray*}
Then, we have
\begin{eqnarray*}
H_{\varrho^*}(X|Y)
 = h\left(\frac{1 + R_\san{zz}}{2} \right) = H_\omega(W),
\end{eqnarray*}
where $H_\omega(W)$ is the entropy of the distribution
defined by
\begin{eqnarray*}
P_{W,\omega}(w) := \sum_{y \in \mathbb{F}_2} P_{XY,\omega}(y+w,y).
\end{eqnarray*}
Then, $t_\san{z} \neq 0$ and the arguments at the end of
Remark \ref{remark:ir-with-code} imply
\begin{eqnarray*}
H_{\omega}(X|Y) < H_\omega(W).
\end{eqnarray*}
Since 
\begin{eqnarray*}
\min_{\varrho \in {\cal P}_c(\omega)} H_\varrho(X|E) \ge 
  \min_{\varrho \in {\cal P}_c(\upsilon)} H_\varrho(X|E)
\ge H_{\varrho^*}(X|E),
\end{eqnarray*}
Eq.~(\ref{eq:key-rate-one-way-bb84}) is strictly higher
than Eq.~(\ref{eq:conventional-one-way-bb84}).
In a similar manner, we can show that
the asymptotic key generation rate
for the $\san{x}$-basis with our proposed channel estimation
is strictly higher that that with the conventional channel estimation
if $t_\san{x} \neq 0$.

Suppose that $(t_\san{z},t_\san{x}) = (0,0)$ and
$R_\san{zx} \neq 0$.
By using Proposition \ref{proposition:unital-bound}, we have
\begin{eqnarray}
\lefteqn{
\min_{\varrho \in {\cal P}_c(\omega)} H_\varrho(X|E) 
 - H_\omega(X|Y) } \nonumber \\
&=& 1 - h\left(\frac{1 + d_\san{z}}{2} \right) 
  - h\left(\frac{1 + d_\san{x}}{2} \right) \nonumber \\ 
&& ~~~~+  h\left( \frac{1 + \sqrt{R_\san{zz}^2 + R_\san{xz}^2}}{2} \right) 
 - h\left(\frac{1 + R_\san{zz}}{2} \right) .
\label{eq:proof-strict-condition1}
\end{eqnarray}
By the singular value decomposition, we have
\begin{eqnarray*}
\left[ \begin{array}{cc}
R_\san{zz} & R_\san{zx} \\
R_\san{xz} & R_\san{xx}
\end{array} \right]
&=& B~\mbox{diag}[d_\san{z}, d_\san{x}]~A \\
 &=& \left[ \begin{array}{c} 
  \bra{B_\san{z}} \\ \bra{B_\san{x}} \end{array} \right]
 \left[ \begin{array}{cc} d_\san{z} & 0 \\ 0 & d_\san{x} \end{array} \right]
  \left[ \begin{array}{cc} \ket{A_\san{z}} & \ket{A_\san{x}} \end{array} 
   \right] \\
 &=& \left[ \begin{array}{cc}
  \braket{B_\san{z}}{\tilde{A}_\san{z}} & \braket{B_\san{z}}{\tilde{A}_\san{x}} \\
  \braket{B_\san{x}}{\tilde{A}_\san{z}} & \braket{B_\san{x}}{\tilde{A}_\san{x}}
 \end{array} \right],
\end{eqnarray*}
where $A$ and $B$ are the rotation matrices, and we set
$\bra{\tilde{A}_\san{z}} = (d_\san{z} A_\san{zz}, d_\san{x} A_\san{zx})$
and
$\bra{\tilde{A}_\san{x}} = (d_\san{z} A_\san{xz}, d_\san{x} A_\san{xx})$.
From Proposition \ref{proposition:explicit-formula-convention-bb84},
we have
\begin{eqnarray}
\lefteqn{
\min_{\varrho \in {\cal P}_c(\upsilon)} [
H_\varrho(X|E) - H_{\varrho}(X|Y) ] } \nonumber \\
&=& 1 - h\left( \frac{1 + \braket{B_\san{z}}{\tilde{A}_\san{z}}}{2} \right)
 - h\left( \frac{1 + \braket{B_\san{x}}{\tilde{A}_\san{x}}}{2}
    \right). 
\label{eq:proof-strict-condition2}
\end{eqnarray}
Subtracting Eq.~(\ref{eq:proof-strict-condition2}) from 
Eq.~(\ref{eq:proof-strict-condition1}), we have
\begin{eqnarray}
\lefteqn{ h\left(\frac{1 +
 \braket{B_\san{x}}{\tilde{A}_\san{x}}}{2}\right)
 + h\left(\frac{1 + \sqrt{R_\san{zz}^2 + R_\san{xz}^2}}{2}\right) }
 \nonumber \\
&&~~- h\left(\frac{1 + d_\san{z}}{2}\right) 
  - h\left(\frac{1 + d_\san{x}}{2}\right) \nonumber \\
&>& h\left( \frac{1 + \| \ket{\tilde{A}_\san{x}} \|}{2} \right)
  + h\left( \frac{1 + \| \ket{\tilde{A}_\san{z}} \|}{2} \right)
  \nonumber \\
&&~~- h\left(\frac{1 + d_\san{z}}{2}\right) 
  - h\left(\frac{1 + d_\san{x}}{2}\right) \nonumber \\
&=& h\left( \frac{1 + \sqrt{d_\san{z}^2 A_\san{xz}^2 + d_\san{x} A_\san{xx}^2}}{2} \right) 
 + h\left(\frac{1 + \sqrt{d_\san{z}^2 A_\san{zz}^2 + d_\san{x}^2 A_\san{zx}^2}}{2} \right) \nonumber \\
&&~~- h\left(\frac{1 + d_\san{z}}{2}\right) 
  - h\left(\frac{1 + d_\san{x}}{2}\right) \nonumber \\
&\ge& A_\san{xz}^2 h\left(\frac{1 + d_\san{z}}{2}\right) + 
 A_\san{xx}^2 h\left(\frac{1 + d_\san{x}}{2}\right) \nonumber \\
&& ~~+ A_\san{zz}^2 h\left(\frac{1 + d_\san{z}}{2}\right) 
  + A_\san{zx}^2 h\left(\frac{1 + d_\san{x}}{2}\right) \nonumber \\
&&~~- h\left(\frac{1 + d_\san{z}}{2}\right) 
  - h\left(\frac{1 + d_\san{x}}{2}\right) \nonumber \\
&=& 0,
\end{eqnarray}
where the second inequality follows from 
the concavity of the function
\begin{eqnarray*}
h\left( \frac{1 + \sqrt{x}}{2} \right),
\end{eqnarray*}
which can be shown by a straight forward calculation.
Thus,
we have shown that
Eq.~(\ref{eq:key-rate-one-way-bb84}) is strictly higher
than Eq.~(\ref{eq:conventional-one-way-bb84}).
In a similar manner, we can show that
the asymptotic key generation rate
for the $\san{x}$-basis with our proposed channel estimation
is strictly higher that with the conventional channel estimation
if $R_\san{xz} \neq 0$.
\end{proof}

\section{Summary}
\label{sec:summary-of-chapter3}

The results in this chapter is summarized as 
follows:
In Section \ref{sec:bb84-and-six-state-protocol},
we formally described the problem setting
of the QKD protocols.

In Section \ref{sec:one-way-IR},
we showed the most basic IR procedure with
one-way public communication.
We introduced the condition such the IR procedure
is universally correct (Definition \ref{def:one-way-correct}). 
This condition was required because the IR procedure have to
be robust against the fluctuation of the estimated probability
of Alice and Bob's bit sequences.
We also explained the conventionally used IR procedure with
the error correcting code, and we clarified that the length
of the syndrome that must be transmitted in the conventional
IR procedure is larger than that in our IR procedure
(Remark \ref{remark:ir-with-code}).
We showed how to apply the LDPC code with the
sum product algorithm in our IR procedure
(Remark \ref{remark:ldpc}). 

In Section \ref{subsec:key-generation-rate},
we showed our proposed channel estimation procedure.
We clarified a sufficient condition on the key generation rate such 
that Alice and Bob can share a secure key (Theorem \ref{theorem:one-way-security}),
and we derived the asymptotic key generation rate formulae.
We developed some techniques to calculate the
asymptotic key generation rates 
(Propositions \ref{proposition:minimization} and
\ref{proposition:unital-bound}) for the BB84 protocol.

In Section \ref{sec:relation-to-conventional},
we explained the conventional estimation
procedure. Then, we derived the asymptotic key generation
rate formulae with the conventional channel estimation.

In Section \ref{sec:example},
we investigated 
the asymptotic key generation rates for some examples
of channels.
We also introduced the concept of the degradable state,
and we clarified that the asymptotic key generation rate
in Eq.~(\ref{eq:asymptotic-direct}) is optimal if the state shared
by Alice, Bob, and Eve is degradable
(Remark \ref{remark:degradable}).
For the rotation channel, we clarified that the
asymptotic key generation rate can be positive even
if the error rate is higher than the 25\% limit
(Section \ref{subsec:unital-channel}).

Finally in Section \ref{sec:condition-strict-improvement},
for the BB84 protocol we clarified the
necessary and sufficient condition such that
the asymptotic key generation rates
with our proposed channel estimation is strictly higher than 
those with the conventional channel estimation
is that the channel is a Pauli channel.

\chapter{Postprocessing}
\label{ch:postprocessing}

\section{Background}

The postprocessing shown in Chapter \ref{ch:channel-estimation}
consists of the IR procedure and the PA procedure.
Roughly speaking, Alice and Bob can share a secret key
with the key generation rate
\begin{eqnarray}
\label{eq:one-way-rough-key-rate}
H_{\rho}(X|E) - H_\rho(X|Y)
\end{eqnarray}
in that postprocessing. An interpretation
of Eq.~(\ref{eq:one-way-rough-key-rate}) is that the
key generation rate is given by the difference between
Eve's ambiguity about Alice's bit sequence 
subtracted by Bob's ambiguity about Alice's 
bit sequence. Therefore, when Eve's ambiguity
about Alice's bit sequence is smaller than
Bob's ambiguity about Alice's bit sequence,
the key generation rate of that postprocessing 
is $0$. 

In \cite{maurer:93}, Maurer
proposed a procedure, the so-called advantage distillation.
The advantage distillation is conducted before
the IR procedure, and the resulting postprocessing
can have positive key generation rate even though
Eq.~(\ref{eq:one-way-rough-key-rate}) is negative.
Gottesman and Lo applied the advantage distillation
to the QKD protocols \cite{gottesman:03}.
In the QKD protocols,
the postprocessing with the advantage distillation was
extensively studied by Bae and Ac\'in \cite{bae:07}.

In this chapter, we propose a new kind of postprocessing,
which can be regarded as a generalization of the postprocessing
that consists of the advantage distillation,
the IR procedure, and the PA procedure.
In our proposed postprocessing, the advantage distillation
and the IR procedure are combined into
one procedure, the two-way IR procedure.
After the two-way IR procedure, we conduct the 
standard PA procedure.

The rest of this chapter is organized as follows:
In Section \ref{sec:advantage-distillation},
we review the advantage distillation.
Then in Section \ref{sec:two-way-ir}, we propose
the two-way information reconciliation procedure.
In Section \ref{sec:key-rate-two-way},
we show a sufficient condition of the key generation rate such
that Alice and Bob can share a secure key by
our proposed postprocessing. 
In Section \ref{sec:example-two-way},
we clarify that the key generation rate of our
proposed postprocessing is higher than the other
postprocessing by showing examples.
Finally, we mention the relation between
our proposed postprocessing and the entanglement
distillation protocols
in Section \ref{sec:relation-edp}.

\section{Advantage Distillation}
\label{sec:advantage-distillation}

In order to clarify the relation between the two-way
IR procedure and the advantage distillation proposed by
Maurer \cite{maurer:93},
we review the postprocessing with the advantage distillation in this section.
For convenience, the notations are adapted to this thesis.
We assume that Alice and Bob have 
correlated binary sequences $\bol{x}, \bol{y} \in \mathbb{F}_2^{2n}$
of even length.
The pair of sequences $(\bol{x},\bol{y})$ is independently identically
distributed (i.i.d.) according to a joint probability distribution
$P_{XY} \in {\cal P}(\mathbb{F}_2 \times \mathbb{F}_2)$.

First, we need to define some auxiliary random variables
to describe the postprocessing with the advantage distillation
procedure.
Let $\xi:\mathbb{F}_2^2 \to \mathbb{F}_2$ be a function defined as
$\xi(a_1,a_2) := a_1 + a_2$ for $a_1,a_2 \in \mathbb{F}_2$, and
let $\zeta:\mathbb{F}_2^2 \to \mathbb{F}_2$ be a function defined as
$\zeta(a,0) := a$ and $\zeta(a,1):=0$ for $a \in \mathbb{F}_2$.
For a pair of joint random variables $((X_1, Y_1)$, $(X_2, Y_2))$ with a
distribution, $P_{XY}^2$, we define random variables
$U_1 := \xi(X_1, X_2)$, $V_1:= \xi(Y_1, Y_2)$ and
$W_1:= U_1 + V_1$.
Furthermore, define random variables  $U_2 := \zeta(X_2,W_1)$,
$V_2:= \zeta(Y_2,W_1)$ and $W_2 := U_2 + V_2$.
For the pair of sequences, 
$\bol{x} = (x_{11},x_{12},\ldots,x_{n1},x_{n2})$ and
$\bol{y} = (y_{11},y_{12}, \ldots, y_{n1}, y_{n2})$,
which is distributed according to $P_{XY}^{2n}$,
let $\bol{u}$, $\bol{v}$ and $\bol{w}$ be $2n$-bit sequences such that
\begin{eqnarray*}
u_{i1} := \xi(x_{i1}, x_{i2}),~~
v_{i1} := \xi(y_{i1}, y_{i2}),~~
w_{i1} := u_{i1} + v_{i1}
\end{eqnarray*}
and
\begin{eqnarray*}
u_{i2} := \zeta(x_{i2}, w_{i1}),~~
v_{i2} := \zeta(y_{i2}, w_{i1}),~~
w_{i2} := u_{i2} + v_{i2}
\end{eqnarray*}
for $1 \le i \le n$.
Then, the pair $(\bol{u},\bol{v})$ and the discrepancy, $\bol{w}$
between $\bol{u}$ and $\bol{v}$
are distributed
according to the distribution $P_{U_1 U_2 V_1 V_2 W_1 W_2}^n$.

The purpose of the advantage distillation 
is to classify blocks of length $2$ according to the
parity $w_{i1}$ of the discrepancies in each block.
When $P_{XY}$ is a distribution such that $P_X$ is
the uniform distribution and $P_{Y|X}$ is a binary
symmetric channel (BSC), 
the validity of this classification
can be understood because we have
\begin{eqnarray*}
H(X_{i2}| Y_{i1} Y_{i2}, W_i = 1) = 1.
\end{eqnarray*}
This formula means that Alice have to send $X_{i2}$ itself
if she want to tell Bob $X_{i2}$. Therefore, they cannot obtain
any secret key from $X_{i2}$, and they should discard
$X_{i2}$ if $W_i = 1$. For general $P_{XY}$, the validity of
above mentioned classification is unclear.
For this reason, we employ a function which is more general than
$\zeta$ in the next section.

By using above preparations, we can describe the postprocessing
with the advantage distillation as follows.
First, Alice sends the parity sequence
$\bol{u}_1 := (u_{11},\ldots,u_{n1})$ to 
Bob so that he can identify the parity sequence 
$\bol{w}_1 := (w_{11},\ldots,w_{n1})$ of the
discrepancies. Bob sends $\bol{w}_1$ back to Alice.
Then, they discard $\bol{u}_1$ and 
$\bol{v}_1 := (v_{11},\ldots,v_{n1})$ respectively,
because $\bol{u}_1$ is revealed to Eve.
As the final step of the advantage distillation, 
Alice calculate\footnote{Conventionally, Alice discard those blocks
if $w_{i1} = 1$. In our procedure, Alice convert the second
bit of those blocks into the constant $u_{i2} = 0$, which is
mathematically equivalent to discarding those blocks.} the sequence 
$\bol{u}_2 := (u_{12},\ldots,u_{n2})$ by using
$\bol{x}$ and $\bol{w}_1$.

At the end of the advantage distillation,
Alice has $\bol{u}_2$ and Bob has $\bol{y}$ and $\bol{w}_1$
as a seed for the key agreement.
By conducting the (one-way) IR procedure and 
the PA procedure for $(\bol{u}_2,(\bol{y},\bol{w}_1))$,
Alice and Bob share a secret key. 

\section{Two-Way Information Reconciliation}
\label{sec:two-way-ir}

In this section, we show the two-way IR procedure.
The essential difference between the two-way IR procedure
and the advantage distillation  is that
Alice does not send the sequence $\bol{u}_1$ itself.
As is usual in information theory, if we allow  negligible error probability,
Alice does not need  to send the parity sequence, $\bol{u}_1$, to Bob
to identify parity sequence $\bol{u}_1$.
More precisely, Bob can decode $\bol{u}_1$ 
with negligible decoding error probability if Alice sends
a syndrome with a sufficient length.
Since Eve's available information from the syndrome is
much smaller than that from sequence $\bol{u}_1$ itself,
Alice and Bob can use $\bol{u}_1$ as a seed for the key agreement.

First, we need to define some auxiliary random variables.
As we have mentioned in the previous section, we use
a function which is more general than $\zeta$.
Let $\chi_A, \chi_B$ be arbitrary functions from
$\mathbb{F}_2^2$ to $\mathbb{F}_2$.
Then, let $\zeta_A:\mathbb{F}_2^3 \to \mathbb{F}_2$
be a function defined as
$\zeta_A(a_1,a_2,a_3) := a_1$ if 
$\chi_A(a_2,a_3) = 0$, and 
$\zeta_A(a_1,a_2,a_3) := 0$ else.
Let $\zeta_B:\mathbb{F}_2^3 \to \mathbb{F}_2$
be a function defined as
$\zeta_B(b_1,b_2,b_3) := b_1$ if 
$\chi_B(b_2,b_3) = 0$, and 
$\zeta_B(b_1,b_2,b_3) := 0$ else.
By using these functions and 
the function $\xi$ defined in the previous section, 
we define the auxiliary random variables:
$U_1 := \xi(X_1, X_2)$, $V_1 := \xi(Y_1, Y_2)$,
$W_1 := U_1 + V_1$, 
$U_2 := \zeta_A(X_2, U_1,V_1)$, and
$V_2 := \zeta_B(Y_2, U_1,V_1)$.
These auxiliary random variables mean that
either Alice or Bob's second bits are kept or 
discarded depending on the values of
$\chi_A(U_1,V_1)$ and $\chi_B(U_1,V_1)$.
The specific form of $\chi_A$ and $\chi_B$
will be given in Section \ref{sec:example-two-way}
so that the asymptotic key generation rates are
maximized.

Our proposed two-way IR procedure is conducted as follows:
\begin{enumerate}
\renewcommand{\theenumi}{\roman{enumi}}
\renewcommand{\labelenumi}{(\theenumi)}
\item \label{step1-two-way-IR}
Alice calculate $\bol{u}_1$ and Bob does the same for
      $\bol{v}_1$.

\item \label{step2-two-way-IR}
Alice calculates syndrome $t_1 = t_1(\bol{u}_1) := M_1 \mathbf{u}_1$,
and sends it to Bob over the public channel.

\item \label{step3-two-way-IR}
Bob decodes $(\mathbf{y}, t_1)$ into estimate of $\mathbf{u}_1$ by 
a decoder $\hat{\bol{u}}_1: (\mathbb{F}_2^2)^n \times \mathbb{F}_2^{k_1} \to \mathbb{F}_2^n$.
Then, he calculates $\hat{\bol{w}}_1 = \hat{\bol{u}}_1 + \bol{v}_1$, and
      sends it to Alice over the public channel.

\item \label{step4-two-way-IR}
Alice calculates $\tilde{\bol{u}}_2$ by using $\bol{x}$,
      $\hat{\bol{w}}_1$, and the function $\zeta_A$.
Bob also calculates $\tilde{\bol{v}}_2$ by using
$\bol{y}$, $\hat{\bol{w}}$, and the function $\zeta_B$.

\item \label{step5-two-way-IR}
Alice calculates syndrome $\tilde{t}_{A,2} := M_{A,2}
      \tilde{\bol{u}}_2$, and sends it to Bob over the public channel.
Bob also calculate syndrome $\tilde{t}_{B,2} := M_{B,2}
      \tilde{\bol{v}}_2$, and sends it to Alice over the public channel.

\item \label{step6-two-way-IR}
Bob decodes $(\bol{y}, \hat{\bol{w}}_1, \tilde{t}_{A,2})$ into estimate
of $\bol{u}_2$ by using a decoder 
$\hat{\bol{u}}_2: (\mathbb{F}_2^2)^n \times \mathbb{F}_2^n \times
      \mathbb{F}_2^{k_{A,2}} \to \mathbb{F}_2^n$.
Alice also decodes $(\bol{x}, \hat{\bol{w}}_1, \tilde{t}_{B,2})$ by
      using a decoder
$\hat{\bol{v}}_2: (\mathbb{F}_2^2)^n \times \mathbb{F}_2^n \times
      \mathbb{F}_2^{k_{B,2}} \to \mathbb{F}_2^n$.
\end{enumerate}

As we mentioned in Section \ref{sec:one-way-IR},
the decoding error probability of the two-way
IR procedure have to be universally small for
any distribution in the candidate $\{ P_{XY, \theta} \mymid \theta \in
\Theta\}$ that are estimated by Alice and Bob.
For this reason, we introduce the concept that
a two-way IR procedure is $\delta$-universally-correct
in a similar manner as in 
Definition \ref{def:one-way-correct}:
\begin{definition}
\label{def:two-way-correct}
We define a two-way IR procedure to be $\delta$-universally-correct
for the class $\{ P_{XY, \theta} \mymid \theta \in \Theta\}$
of probability distribution if
\begin{eqnarray*}
\lefteqn{ P_{XY, \theta}^{2n}(\{ (\bol{x}, \bol{y}) \mymid
  (\bol{u}_1,\tilde{\bol{u}}_2, \hat{\bol{v}}_2) \neq
  (\bol{u}_1,\bol{u}_2, \bol{v}_2)
 \mbox{ or } } \\ 
&& ~~~~~~~~~~~~~~~(\hat{\bol{u}}_1,\hat{\bol{u}}_2,\tilde{\bol{v}}_2) 
 \neq (\bol{u}_1,\bol{u}_2,\bol{v}_2)
  \}) \le \delta
\end{eqnarray*}
for any $\theta \in \Theta$.
\end{definition}

An example of a decoder that fulfils the universality is the
minimum entropy decoder. For Step (\ref{step3-two-way-IR}), the minimum entropy
decoder is defined by
\begin{eqnarray*}
\hat{\bol{u}}_1(\bol{y}, t_1) :=
  \argmin_{\bol{u}_1 \in \mathbb{F}_2^n: M_1 \bol{u}_1 = t_1}
  H(P_{\bol{u}_1 \bol{y}}),
\end{eqnarray*}
where $P_{\bol{u}_1 \bol{y}} \in {\cal P}_n(\mathbb{F}_2^3)$ is the joint type of the sequence
\begin{eqnarray*}
(\bol{u}_1, \bol{y}) = ((u_{11},y_{11},y_{12}), \ldots,(u_{n1},y_{n1},y_{n2}))
\end{eqnarray*}
of length $n$. For Step (\ref{step6-two-way-IR}),
the minimum entropy decoder is defined by 
\begin{eqnarray*}
\hat{\bol{u}}_2(\bol{y},\bol{w}_1,t_2) :=
 \argmin_{\bol{u}_2 \in \mathbb{F}_2^n: M_{A,2} \bol{u}_2 = t_{A,2}}
   H(P_{\bol{u}_2 \bol{w}_1 \bol{y}}),
\end{eqnarray*}
where $P_{\bol{u}_2 \bol{w}_1 \bol{y}} \in {\cal P}_n(\mathbb{F}_2^4)$ 
is the joint type of the sequence
\begin{eqnarray*}
(\bol{u}_2,\bol{w}_1,\bol{y}) = ((u_{12}, w_{11}, y_{11},y_{12}),
 \ldots, (u_{n2}, w_{n1}, y_{n1}, y_{n2}))
\end{eqnarray*}
of length $n$. The minimum entropy decoder for $\hat{\bol{v}}_2$
is defined in a similar manner.
\begin{theorem}
\label{theorem:universal-coding-two-way}
\cite[Theorem 1]{csiszar:82}
Let $r_1$, $r_{A,1}$, and  $r_{A,2}$ be real numbers that satisfy
\begin{eqnarray*}
r_1 > \min_{\theta \in \Theta} H(U_{1,\theta}|Y_{1,\theta} Y_{2,\theta}),
\end{eqnarray*}
\begin{eqnarray*}
r_{A,2} > \min_{\theta \in \Theta} H(U_{2,\theta}| W_{1,\theta} Y_{1,\theta}
 Y_{2,\theta}),
\end{eqnarray*}
and 
\begin{eqnarray*}
r_{B,2} > \min_{\theta \in \Theta} H(V_{2,\theta}| W_{1,\theta} X_{1,\theta}
 X_{2,\theta}),
\end{eqnarray*}
respectively, where $U_{1,\theta} = \xi(X_{1,\theta},X_{2,\theta})$,
$W_{1,\theta} = U_{1,\theta} + \xi(Y_{1,\theta},Y_{2,\theta})$,
and $U_{2,\theta} = \zeta(X_{2,\theta},W_{1,\theta})$ for
the random variables $(X_{1,\theta},X_{2,\theta},
 Y_{1,\theta},Y_{2,\theta})$
that are distributed according to $P_{XY,\theta}^2$.
Then, for every sufficiently large $n$, there exist a $k_1 \times n$
parity check matrix $M_1$, a $k_{A,2} \times n$ parity check matrix
$M_{A,2}$, and a $k_{B,2} \times n$ parity check matrix $M_{B,2}$ 
 such that $\frac{k_1}{n} \le r_1$, $\frac{k_{A,2}}{n} \le r_{A,2}$,
and $\frac{k_{B,2}}{n} \le r_{B,2}$,
and the decoding error probability by the minimum entropy
decoding satisfies
\begin{eqnarray*}
\lefteqn{ P_{XY, \theta}^{2n}(\{ (\bol{x}, \bol{y}) \mymid
  (\bol{u}_1,\tilde{\bol{u}}_2, \hat{\bol{v}}_2) \neq
  (\bol{u}_1,\bol{u}_2, \bol{v}_2)   
 \mbox{ or } } \\
 && ~~~~~~~~~~~~(\hat{\bol{u}}_1,\hat{\bol{u}}_2,\tilde{\bol{v}}_2) 
 \neq (\bol{u}_1,\bol{u}_2,\bol{v}_2)
  \})  \\
&\le& e^{-n E_1} + e^{-n E_{A,2}} + e^{-n E_{B,2}}
\end{eqnarray*}
for any $\theta \in \Theta$,
where $E_1, E_{A,2},E_{B,2} > 0$ are constants that do not
depends on $n$.
\end{theorem}

\section{Security and Asymptotic Key Generation Rate}
\label{sec:key-rate-two-way}

\subsection{Sufficient Condition on Key Generation Rate for Secure Key Agreement}

In this section, we show how Alice and Bob decide
the parameters of the postprocessing and share
a secret key. Then, we show a sufficient condition
on the parameters such that Alice and Bob can share
a secure key.
We employ almost the same notations 
as in Section \ref{subsec:key-generation-rate}.

Let us start with the six-state protocol.
Instead of the conditional von Neumann entropy
$H_\rho(X|E)$, the quantities
\begin{eqnarray}
\label{eq:definition-two-way-conditional-entropy-1}
H_{\rho}(U_1 U_2 V_2|W_1 E_1 E_2) = H(\rho_{U_1 U_2 V_2 W_1 E_1 E_2})
   - H(\rho_{W_1 E_1 E_2})
\end{eqnarray}
and 
\begin{eqnarray}
\label{eq:definition-two-way-conditional-entropy-2}
H_{\rho}(U_2 V_2|U_1 W_1 E_1 E_2) = H(\rho_{U_1 U_2 V_2 W_1 E_1 E_2})
   - H(\rho_{U_1 W_1 E_1 E_2})
\end{eqnarray}
play important roles in our postprocessing,
where the von Neumann entropies are calculated 
with respect to the 
operator $\rho_{U_1 U_2 V_2 W_1 E_1 E_2}$ derived from
$\rho_{AB}^{\otimes 2}$ via the measurement and
the functions $\xi, \zeta_A, \zeta_B$.
For the ML estimator $\hat{\rho}(\bol{z})$
of $\rho \in {\cal P}_c$,
we set
\begin{eqnarray*}
\hat{H}_\bol{z}(U_1 U_2 V_2|W_1 E_1 E_2) 
 := H_{\hat{\rho}(\bol{z})}(U_1 U_2 V_2|W_1 E_1 E_2)
\end{eqnarray*}
and
\begin{eqnarray*}
\hat{H}_\bol{z}(U_2 V_2|U_1 W_1 E_1 E_2) 
 := H_{\hat{\rho}(\bol{z})}(U_2 V_2|U_1 W_1 E_1 E_2),
\end{eqnarray*}
which are the ML estimators of the quantities 
in Eqs.~(\ref{eq:definition-two-way-conditional-entropy-1})
and (\ref{eq:definition-two-way-conditional-entropy-2}) 
respectively.

For the BB84 protocol, we similarly set
\begin{eqnarray*}
\hat{H}_\bol{z}(U_1 U_2 V_2|W_1 E_1 E_2) 
 := \min_{\varrho \in {\cal P}_c(\hat{\omega}(\bol{z}))}
  H_{\varrho}(U_1 U_2 V_2|W_1 E_1 E_2)
\end{eqnarray*}
and
\begin{eqnarray*}
\hat{H}_\bol{z}(U_2 V_2|U_1 W_1 E_1 E_2)  
:= \min_{\varrho \in {\cal P}_c(\hat{\omega}(\bol{z}))}
H_{\varrho}(U_2 V_2|U_1 W_1 E_1 E_2)
\end{eqnarray*}
respectively.

According to the sample bit sequence $\bol{z}$,
Alice and Bob decide the rate $\frac{k_1(\bol{z})}{n}$,
$\frac{k_{A,2}(\bol{z})}{n}$, and $\frac{k_{B,2}(\bol{z})}{n}$ 
of the parity check matrices used in the 
two-way IR procedure. Furthermore, they also
decide the length $\ell(\bol{z})$ of the finally distilled
key according to the sample bit sequence $\bol{z}$.
Then, they conduct the postprocessing as follows.
\begin{enumerate}
\renewcommand{\theenumi}{\roman{enumi}}
\renewcommand{\labelenumi}{(\theenumi)}
\item \label{step1-two-way-postprocessing}
Alice and Bob undertake the two-way IR procedure
of Section \ref{sec:two-way-ir}, and they obtain
$(\bol{u}_1, \tilde{\bol{u}}_2, \hat{\bol{v}}_2)$
and $(\hat{\bol{u}}_1,\hat{\bol{u}}_2, \tilde{\bol{v}}_2)$
respectively.

\item \label{step2-two-way-postprocessing}
Alice and Bob carry out the PA procedure to distill
a key pair $(s_A, s_B)$. First, Alice randomly
chooses a hash function,
 $f:\mathbb{F}_2^{3n} \to \{0,1\}^{\ell(\bol{z})}$,
from a family of two-universal hash functions, and sends the
choice of $f$ to Bob over the public channel. Then, Alice's distilled
key is $s_A = f(\bol{u}_1,\tilde{\bol{u}}_2, \hat{\bol{v}}_2)$ and Bob's distilled key
is $s_B = f(\hat{\bol{u}}_1,\hat{\bol{u}}_2, \tilde{\bol{v}}_2)$ respectively.
\end{enumerate}

The distilled key pair and Eve's available information can be
described by a $\{cccq\}$-state,
$\rho_{S_A S_B C \bol{E}}^{\bol{z}}$, where classical system
$C$ consists of random variables $T_1$, $\tilde{T}_{A,2}$, and $\tilde{T}_{B,2}$ that
describe the syndromes transmitted in Steps (\ref{step2-two-way-IR})
and (\ref{step5-two-way-IR}) of the two-way IR procedure
and random variable $F$ that describe the choice of the
function in the PA procedure. Then, the security
of the distilled key pair is defined in the same way as
in Section \ref{subsec:key-generation-rate}, i.e.,
the key pair is said to be $\varepsilon$-secure if
Eq.~(\ref{eq-varepsilon-secure-average}) is satisfied.


The following theorem gives a sufficient condition
on $k_1(\bol{z})$, $k_{A,2}(\bol{z})$, 
$k_{B,2}(\bol{z})$, and $\ell(\bol{z})$ such
that the distilled key is secure.
\begin{theorem}
\label{theorem:security-two-way}
For each sample sequence $\bol{z} \in {\cal Q}$, 
assume that the 
IR procedure is $\delta$-universally-correct for the
class of distributions
\begin{eqnarray*}
\{ P_{XY, \rho} \mymid \|\hat{\rho}(\bol{z}) - \rho\| \le \alpha \}
\end{eqnarray*}
in the six-state protocol, and for the
class of distributions 
\begin{eqnarray*}
\{ P_{XY,\omega} \mymid \|\hat{\omega}(\bol{z}) - \omega\| \le
 \alpha\}
\end{eqnarray*} 
in the BB84 protocol. 
For each $\bol{z} \in {\cal Q}$, if we set
\begin{eqnarray}
\lefteqn{ \frac{\ell(\bol{z})}{2n} } \nonumber \\
&<& 
 \frac{1}{2}  \max \left[
\hat{H}_{\bol{z}}(U_1 U_2 V_2| W_1 E_1 E_2) - \eta(\alpha) -
\frac{k_1(\bol{z})}{n} - \frac{k_{A,2}(\bol{z})}{n} - \frac{k_{B,2}(\bol{z})}{n},
\right. \nonumber \\
&& ~~~~\left. \hat{H}_{\bol{z}}(U_2 V_2|U_1 W_1 E_1 E_2) -
                       \eta(\alpha) - \frac{k_{A,2}(\bol{z})}{n} - \frac{k_{B,2}(\bol{z})}{n}
\right] - \nu_n, 
\label{eq:two-way-key-generation-rate}
\end{eqnarray}
then the distilled key is $(\varepsilon + 3 \delta + \mu(\alpha,m))$-secure,
where $\nu_n := 5 \sqrt{\frac{\log(36/\varepsilon^2)}{n}} + \frac{2
 \log(3/\varepsilon)}{n}$.
\end{theorem}
\begin{proof}
We only prove the statement for the six-state protocol,
because the statement for the BB84 protocol is proved
exactly in the same way by replacing
$\rho \in {\cal P}_c$
with $\omega \in \Omega$ and some other related quantities.
The assertion of the theorem is proved by using
Corollary \ref{corollary-privacy-amplification},
Lemma \ref{lemma:min-entropy-of-product},
Lemma \ref{lemma-error-prob-continuity},
and Eq.~(\ref{eq:definition-mu-1}).

For any $\rho \in {\cal P}_c$,
Eq.~(\ref{eq:definition-mu-1}) means that 
$\| \hat{\rho}(\bol{z}) - \rho\| \le \alpha$ with probability
$1-\mu(\alpha,m)$.
When $\|\hat{\rho}(\bol{z}) - \rho\| > \alpha$,
the distilled key pair is $1$-secure.
For $\| \hat{\rho}(\bol{z}) - \rho\| \le \alpha$,
we
first assume (proved later) that the dummy key 
$S := f(\bol{U}_1,\bol{U}_2,\bol{V}_2)$ is
$\varepsilon$-secret under the condition that
Eve can access $(\bol{W}_1,T_1, T_{A,2},T_{B,2}, F, \bol{E})$, i.e.,
\begin{eqnarray}
\label{eq:two-way-proof-1}
\frac{1}{2} \| \rho^\bol{z}_{S \bol{W}_1 T_1 T_{A,2} T_{B,2} F \bol{E}} - 
  \rho_S^{\bol{z},\mix} \otimes \rho^\bol{z}_{\bol{W}_1 T_1 T_{A,2} T_{B,2} F \bol{E}} \|
\le \varepsilon.
\end{eqnarray}
The assumption that the two-way IR procedure is 
$\delta$-universally-correct implies that
$\hat{\bol{w}}_1 = \bol{w}_1$,  
$\tilde{t}_{A,2} = t_{A,2} := M_{A,2} \bol{u}_2$,
and $\tilde{t}_{B,2} = t_{B,2} := M_{B,2} \bol{v}_2$ with probability
at least $1- \delta$.
Since $(\bol{u}_2,\tilde{\bol{u}}_2)$, $(\bol{v}_2, \hat{\bol{v}}_2)$, 
$(\bol{w}_1,\hat{\bol{w}}_1)$,
$(t_{A,2},\tilde{t}_{A,2})$, and $(t_{B,2}, \tilde{t}_{B,2})$
 can be computed from $(\bol{x},\bol{y})$,
by using Lemma \ref{lemma-error-prob-continuity}, we have
\begin{eqnarray*}
\| \rho^\bol{z}_{\bol{X}\bol{Y} \bol{U}_1 \tilde{\bol{U}}_2 \hat{\bol{V}}_2
 \hat{\bol{W}}_1 T_1 \tilde{T}_{A,2} \tilde{T}_{B,2} F
 \bol{E}}
 - \rho^\bol{z}_{\bol{X} \bol{Y} \bol{U}_1 \bol{U}_2 \bol{V}_2 \bol{W}_1 T_1
 T_{A,2} T_{B,2} F \bol{E}} \| \le 2 \delta.
\end{eqnarray*}
Since the trace distance does not increase by CP maps, we have
\begin{eqnarray*}
\| \rho^\bol{z}_{S_A \bol{W}_1 T_1 \tilde{T}_{A,2} \tilde{T}_{B,2} F \bol{E}} -
  \rho^\bol{z}_{S \bol{W}_1 T_1 T_{A,2} T_{B,2} F \bol{E}} \| \le 2 \delta.
\end{eqnarray*}
Therefore, the statement that the dummy key $S$ is 
$\varepsilon$-secret implies that the actual key
$S_A$ is $(\varepsilon + 2 \delta)$-secret
as follows:
\begin{eqnarray*}
\lefteqn{
\| 
\rho^\bol{z}_{S_A \hat{\bol{W}}_1 T_1 \tilde{T}_{A,2} \tilde{T}_{B,2} F \bol{E}} -
\rho_{S_A}^{\bol{z},\rom{mix}} \otimes 
\rho^\bol{z}_{ \hat{\bol{W}}_1 T_1 \tilde{T}_{2,A} \tilde{T}_{2,B} F \bol{E}}
\|
} \\
&\le& \|
\rho^\bol{z}_{S_A \hat{\bol{W}}_1 T_1 \tilde{T}_{A,2} \tilde{T}_{B,2} F \bol{E}} -
\rho^\bol{z}_{S \bol{W}_1 T_1 T_{A,2} T_{B,2} F \bol{E}} 
\| \\
&& +  \|
\rho^\bol{z}_{S \bol{W}_1 T_1 T_{A,2} T_{B,2} F \bol{E}} -
\rho_S^{\bol{z},\rom{mix}} \otimes 
\rho^\bol{z}_{\bol{W}_1 T_1 T_{A,2} T_{B,2} F \bol{E}}
\| \\
&& + \|
\rho_S^{\bol{z}, \rom{mix}} \otimes \rho^\bol{z}_{\bol{W}_1 T_1 T_{A,2} T_{B,2} F \bol{E}} -
\rho_{S_A}^{\bol{z},\rom{mix}} \otimes 
\rho^\bol{z}_{\hat{\bol{W}}_1 T_1 \tilde{T}_{A,2} \tilde{T}_{B,2} F \bol{E}}
\|,
\end{eqnarray*}
where the first term is upper bounded by $2 \delta$,
the second term is upper bounded by $2 \varepsilon$,
and the third term is also upper bounded by $2 \delta$
because $\rho_S^{\bol{z},\mix} = \rho_{S_A}^{\bol{z},\mix}$. The assumption 
that the two-way IR procedure is $\delta$-universally-correct
also implies that the distilled key pair $(S_A,S_B)$ is
$\delta$-universally-correct. Thus, the key pair is 
$(\varepsilon + 3 \delta)$-secure if 
$\| \hat{\rho}(\bol{z}) - \rho\| \le \alpha$.
Averaging over the sample sequence $\bol{z} \in {\cal Q}$,
the distilled key pair is 
$(\varepsilon + 3 \delta +  \mu(\alpha,m))$-secure.

One thing we have left is to prove Eq.~(\ref{eq:two-way-proof-1}).
According to  Lemma \ref{lemma:min-entropy-of-product}, the inequality
\begin{eqnarray*}
\lefteqn{ \frac{\ell(\bol{z})}{2n} < } \\
&&
   \frac{1}{2} \left[ \hat{H}_{\bol{z}}(U_1 U_2 V_2| W_1 E_1 E_2)
                - \eta(\alpha) - \frac{k_1(\bol{z})}{n} -
   \frac{k_{A,2}(\bol{z})}{n} - \frac{k_{B,2}(\bol{z})}{n} \right] - \nu_n
\end{eqnarray*}
implies the inequality
\begin{eqnarray*}
\lefteqn{ \ell(\bol{z}) < } \\
&&
 H_{\min}^{\bar{\varepsilon}}(\rho_{\bol{U}_1 \bol{U}_2 \bol{V}_2 \bol{W}_1
 \bol{E}}|\bol{W}_1 \bol{E}) -  k_1(\bol{z}) - k_{A,2}(\bol{z}) - k_{B,2}(\bol{z})
 - 2 \log(3/2 \varepsilon).
\end{eqnarray*}
Thus, Corollary \ref{corollary-privacy-amplification}
implies that the dummy key $S$ is $\varepsilon$-secret.

Since the syndrome $T_1$ is computed from the sequence
$\bol{U}_1$, if the dummy key $S$ is $\varepsilon$-secret
in the case that Eve can access the sequence $\bol{U}_1$,
then the dummy key $S$ must be $\varepsilon$-secret in the
case that Eve can only access the syndrome $T_1$ instead
of $\bol{U}_1$. According to 
Lemma \ref{lemma:min-entropy-of-product}, the inequality
\begin{eqnarray*}
\frac{\ell(\bol{z})}{2n} <
   \frac{1}{2} \left[ \hat{H}_{\bol{z}}(U_2 V_2| U_1 W_1 E_1 E_2)
                - \eta(\alpha) -
   \frac{k_{A,2}(\bol{z})}{n} - \frac{k_{B,2}(\bol{z})}{n} \right] - \nu_n
\end{eqnarray*}
implies the inequality
\begin{eqnarray*}
\ell(\bol{z}) <
 H_{\min}^{\bar{\varepsilon}}(\rho_{\bol{U}_1 \bol{U}_2 \bol{V}_2 \bol{W}_1
 \bol{E}}|\bol{U}_1 \bol{W}_1 \bol{E}) - k_{A,2}(\bol{z}) - k_{B,2}(\bol{z})
 - 2 \log(3/2 \varepsilon).
\end{eqnarray*}
Thus, Corollary \ref{corollary-privacy-amplification}
implies that the dummy key $S$ is $\varepsilon$-secret.

Combining above two arguments, we have the assertion of
the theorem.
\end{proof}

\begin{remark}
The maximization in Eq.~(\ref{eq:two-way-key-generation-rate})
is very important. If either of them is omitted,
the key generation rate of the postprocessing
can be underestimated, as will be discussed in 
Section \ref{sec:example-two-way}.
\end{remark}
\begin{remark}
By switching the role of Alice and Bob, we obtain 
a postprocessing with the reverse two-way 
IR procedure. For the postprocessing with
the reverse two-way IR procedure,
we can show almost the same statement as 
Theorem \ref{theorem:security-two-way}
by replacing $U_1$ with $V_1$, and by using the 
$\delta$-universally-correct for the reverse
two-way IR procedure.
\end{remark}

\subsection{Asymptotic Key Generation Rates}

In this section, we derive the asymptotic key generation
rate formula for the postprocessing with the
two-way IR procedure.
First, we consider the six-state protocol.
Since the ML estimator is a consistent
estimator, in a similar arguments
as in Sections \ref{subsec:key-generation-rate} 
and \ref{subsec:six-state}, we can set the sequence of the key generation rates so that
it converges to the asymptotic key
generation rate formula
\begin{eqnarray}
\lefteqn{ \frac{1}{2} \max
  \left[
   H_{\rho}(U_1 U_2 V_2|W_1 E_1 E_2) -
   H_\rho(U_1|Y_1 Y_2) 
   \right. } \nonumber \\ 
   && ~~~~~~~~~~~~~~~ - H_\rho(U_2|W_1 Y_1
   Y_2) - H_\rho(V_2|W_1 X_1 X_2), \nonumber \\
&& \hspace{-5mm} \left. H_{\rho}(U_2 V_2|U_1 W_1 E_1 E_2) -
    H_\rho(U_2|W_1 Y_1
   Y_2) - H_\rho(V_2|W_1 X_1 X_2) \right]
\label{eq:asymptotic-key-rate-two-way}
\end{eqnarray}
in probability as $m,n \to \infty$.
We can also derive the asymptotic key generation formula
for the postprocessing with the reverse two-way IR procedure
as
\begin{eqnarray}
\lefteqn{ \frac{1}{2} \max
  \left[
   H_{\rho}(V_1 U_2 V_2|W_1 E_1 E_2) -
   H_\rho(V_1|X_1 X_2)
   \right. } \nonumber \\ 
   && ~~~~~~~~~~~~~~~ - H_\rho(U_2|W_1 Y_1 Y_2) - H_\rho(V_2|W_1 X_1
   X_2) ,  \nonumber \\
&& \hspace{-5mm} \left. H_{\rho}(U_2 V_2|U_1 W_1 E_1 E_2) -
    H_\rho(U_2|W_1 Y_1 Y_2) - H_\rho(V_2|W_1 X_1
   X_2) \right].
\label{eq:asymptotic-key-rate-two-way-reverse}
\end{eqnarray}

Next, we consider the BB84 protocol.
Since the ML estimator is a consistent
estimator, in a similar arguments
as in Sections \ref{subsec:key-generation-rate} 
and \ref{subsec:bb84}, we can set the sequence of the key generation rates so that
it converges to the asymptotic key generation rate formula
\begin{eqnarray}
\lefteqn{ \frac{1}{2} \min_{\varrho \in {\cal P}_c(\omega)} \max
  \left[
   H_{\varrho}(U_1 U_2 V_2|W_1 E_1 E_2) -
   H_\omega(U_1|Y_1 Y_2) 
   \right. } \nonumber \\ 
   && ~~~~~~~~~~~~~~~ - H_\omega(U_2|W_1 Y_1
   Y_2) - H_\omega(V_2|W_1 X_1 X_2) ,  \nonumber \\
&& \hspace{-5mm} \left. H_{\varrho}(U_2 V_2|U_1 W_1 E_1 E_2) -
    H_\omega(U_2|W_1 Y_1
   Y_2) - H_\omega(V_2|W_1 X_1 X_2) \right],
\label{eq:asymptotic-key-rate-two-way-bb84}
\end{eqnarray}
in probability as $m,n \to \infty$.

We can also derive the asymptotic key generation rate formula
for the postprocessing with the reverse two-way IR procedure
as
\begin{eqnarray}
\lefteqn{ \frac{1}{2} \min_{\varrho \in {\cal P}_c(\omega)} \max
  \left[
   H_{\varrho}(V_1 U_2 V_2|W_1 E_1 E_2) -
   H_\omega(U_1|X_1 X_2)
   \right. } \nonumber \\ 
   && ~~~~~~~~~~~~~~~ - H_\omega(U_2|W_1 Y_1
   Y_2) - H_\omega(V_2|W_1 X_1
   X_2),  \nonumber \\
&& \hspace{-5mm} \left. H_{\varrho}(V_2|U_1 W_1 E_1 E_2) -
    H_\omega(U_2|W_1 Y_1
   Y_2) - H_\omega(V_2|W_1 X_1
   X_2) \right].
\label{eq:asymptotic-key-rate-two-way-reverse-bb84}
\end{eqnarray}

The following propositions are useful to calculate
the minimizations in 
Eqs.~(\ref{eq:asymptotic-key-rate-two-way-bb84}) 
and (\ref{eq:asymptotic-key-rate-two-way-reverse-bb84}).
\begin{proposition}
\label{proposition:convexity-2}
For two density operator 
$\rho^1, \rho^2 \in {\cal P}_c$
and a probabilistically mixture 
$\rho^\prime := \lambda \rho^1 + (1-\lambda) \rho^2$,
Eve's ambiguities are convex, i.e., we have
\begin{eqnarray*}
\lefteqn{ H_{\rho^\prime}(U_1 U_2 V_2|W_1 E_1 E_2) } \\
&\le&
\lambda H_{\rho^1}(U_1 U_2 V_2|W_1 E_1 E_2) + 
 (1 - \lambda) H_{\rho^2}(U_1 U_2 V_2 |W_1 E_1 E_2)
\end{eqnarray*}
and 
\begin{eqnarray*}
\lefteqn{ H_{\rho^\prime}(U_2 V_2|U_1 W_1 E_1 E_2) } \\
&\le&
\lambda H_{\rho^1}(U_2 V_2|U_1 W_1 E_1 E_2) + 
 (1 - \lambda) H_{\rho^2}(U_2 V_2|U_1 W_1 E_1 E_2),
\end{eqnarray*}
where $\rho^\prime_{U_1 U_2 V_2 W_1 E_1 E_2}$ is 
the density operator derived from a purification 
$(\psi^{\prime}_{ABE})^{\otimes 2}$ of 
$(\rho^\prime_{AB})^{\otimes 2}$.
\end{proposition}
\begin{proof}
The statement of this proposition is shown exactly
in the same way as Proposition \ref{proposition:convexity}.
\end{proof}

\begin{proposition}
\label{proposition:minimization-two-way}
For the BB84 protocol,
the minimization in Eqs.~(\ref{eq:asymptotic-key-rate-two-way-bb84}) 
and (\ref{eq:asymptotic-key-rate-two-way-reverse-bb84}) is achieved
by Choi operator $\varrho$ whose components $R_{\san{zy}}$,
$R_{\san{xy}}$, $R_{\san{yz}}$, $R_{\san{yx}}$,
and $t_{\san{y}}$, are all $0$.
\end{proposition}
\begin{proof}
The statement of this proposition is shown exactly
in the same way as Proposition \ref{proposition:minimization}
by using Proposition \ref{proposition:convexity-2}.
\end{proof}

\begin{remark}
By using the chain rule of von Neumann entropy, we can rewrite
 Eq.~(\ref{eq:asymptotic-key-rate-two-way}) as
\begin{eqnarray}
\lefteqn{ \frac{1}{2} \{ \max[ H_\rho(U_1|W_1 E_1 E_2) -  H(U_{1,\rho}|Y_{1,\rho}
 Y_{2,\rho}), 0] } \nonumber \\
&& \hspace{-13mm} + H_{\rho}(U_2 V_2|U_1 W_1 E_1 E_2)
  - H_\rho(U_2|W_1 Y_1 Y_2 ) - H_\rho(V_2|W_1 X_1 X_2 ) \}.
\label{eq:explanation-singular-point}
\end{eqnarray}
We can interpret this formula as follows. 
If Bob's ambiguity $H_\rho(U_1|Y_1 Y_2)$
about bit $U_1$ is smaller than Eve's ambiguity
$H_\rho(U_1|W_1 E_1 E_2)$ about $U_1$, then
Eve cannot decode sequence 
$\bol{U}_1$ \cite{slepian:73, devetak:03},
and there exists some remaining ambiguity
about bit $U_1$ for Eve. We can thus distill
some secure key from bit $U_1$.
On the other hand, if Bob's 
ambiguity $H_\rho(U_1|Y_1 Y_2)$
about bit $U_1$, i.e., the amount of transmitted 
syndrome per bit, is larger than Eve's ambiguity
$H_\rho(U_1|W_1 E_1 E_2)$ about $U_1$, then
Eve might be able to decode sequence $\bol{U}_1$ from 
her side information and the transmitted syndrome
\cite{slepian:73, devetak:03}.
Thus, there exists the possibility that Eve can completely
know bit $U_1$, and we can distill no 
secure key from bit $U_1$, because we have to
consider the worst case in a cryptography scenario.
Consequently, sending the compressed version (syndrome)
of sequence $\bol{U}_1$ instead of $\bol{U}_1$ itself
is not always effective, and the slope of the key rate curves
change when Eve becomes able to decode
$\bol{U}_1$ (see Figs.~\ref{fig:hash-six-depolarizing},
\ref{fig:hash-bb84-depolarizing},
\ref{fig:quarter-bb84},
\ref{fig:quarter-six-state},
\ref{fig:two-way-amplitude}).

A similar argument also holds for the BB84 protocol.
\end{remark}

\begin{remark}
If we take the functions $\chi_A$ and $\chi_B$
as 
\begin{eqnarray}
\chi_A(a_1,a_2) := \left\{ \begin{array}{ll}
0 & \mbox{if } a_1 = a_2 \\
1 & \mbox{else}
\end{array} \right.
\label{eq:chi-A}
\end{eqnarray}
and
\begin{eqnarray}
\label{eq:chi-B} 
\chi_B(b_1,b_2) = 1. 
\end{eqnarray}
Then, the postprocessing proposed in
this thesis reduces to the postprocessing proposed in
\cite{watanabe:07}.
\end{remark}

\begin{remark}
The asymptotic key generation rate (for the six-state protocol) of the postprocessing
with the advantage distillation is given by
\begin{eqnarray}
\label{eq:key-rate-of-adv}
\frac{1}{2}[ H_\rho(U_2|U_1 W_1 E_1 E_2) - H_\rho(U_2|W_1 Y_1 Y_2)],
\end{eqnarray}
where the auxiliary random variables $U_1, U_2, W_1$ are defined as in
Section \ref{sec:advantage-distillation}, or they are defined by
using the functions $\chi_A, \chi_B$ given in
Eqs.~(\ref{eq:chi-A}) and (\ref{eq:chi-B}).
From Eqs.~(\ref{eq:asymptotic-key-rate-two-way})
and (\ref{eq:key-rate-of-adv}), we can find that
the asymptotic key generation rate of the proposed postprocessing
is at least as high as that of the postprocessing with 
the advantage distillation if we employ appropriate functions
$\chi_A, \chi_B$.

A similar argument also holds for the BB84 protocol.
\end{remark}
\begin{remark}
In \cite{gohari:08}, Gohari and Anantharam 
proposed\footnote{It should be noted that they consider the classical 
key agreement problem instead of the postprocessing of the QKD protocol.
However, as we mentioned in Chapter \ref{ch:introduction},
they are essentially the same.} a two-way postprocessing which is
similar to our proposed two-way postprocessing.
They derived the asymptotic key generation rate formula of their proposed
postprocessing. Although their postprocessing seems to be
a generalization of our proposed postprocessing, 
the asymptotic key generation 
rate (Eq.~(\ref{eq:asymptotic-key-rate-two-way})) of our proposed
postprocessing cannot be derived by their
asymptotic key generation rate formula.
By modifying their formula for the QKD protocol, we can only derive
the asymptotic key
generation rate
\begin{eqnarray}
\lefteqn{ 
\frac{1}{2} [ H_\rho(U_1|E_1 E_2) - H_\rho(U_1|Y_1 Y_2)
} \nonumber \\
&&+ H_\rho(W_1|U_1 E_1 E_2) - H_\rho(W_1| U_1 X_1 X_2) \nonumber \\
&&+ H_\rho(U_2|U_1 W_1 E_1 E_2) - H_\rho(U_2|U_1 W_1 Y_1 Y_2) \nonumber \\
&& + H_\rho(V_2|U_1 W_1 U_2 E_1 E_2) - H_\rho(V_2|U_1 W_1 U_2 X_1 X_2) ].
\label{eq:rate-derived-by-isit-paper}
\end{eqnarray}
For a Pauli channel, since $W_1$ is independent from 
$(X_1, X_2)$ and $H_\rho(W_1| E_1 E_2) = 0$, 
Eq.~(\ref{eq:rate-derived-by-isit-paper}) is
strictly smaller than Eq.~(\ref{eq:asymptotic-key-rate-two-way}). 

The underestimation of the asymptotic key generation rate comes from 
the following reason.
In Gohari and  Anantharam's postprocessing, a syndrome
of $\bol{w}_1$ is transmitted over the public channel,
and the length of the syndrome is roughly
$H_\rho(W_1|U_1X_1  X_2)$. When the syndrome is transmitted
over the public channel, Eve cannot obtain more
information than $\bol{w}_1$ itself.
The lack of this observation results into 
Eq.~(\ref{eq:rate-derived-by-isit-paper}).
\end{remark}

\section{Comparison of Asymptotic Key Generation Rates for Specific Channels}
\label{sec:example-two-way}

In this section, we compare the asymptotic key generation rates
of the proposed postprocessing, the postprocessing with the
advantage distillation, the one-way postprocessing
for representative specific channels.

\subsection{Pauli Channel}
\label{subsec:pauli-channel}

When the channel between Alice and Bob is a 
Pauli channel, the Stokes parameterization of the corresponding
density operator $\rho \in {\cal P}_c$
is
\begin{eqnarray}
\label{eq:pauli-channel}
\left(
\left[
\begin{array}{ccc}
e_\san{z} & 0 & 0 \\
0 & e_\san{x} & 0 \\
0 & 0 & e_\san{y}
\end{array}
\right],
\left[
\begin{array}{c}
0 \\ 0 \\ 0
\end{array}
\right]
\right),
\end{eqnarray}
for $-1 \le e_\san{z},e_\san{x},e_\san{y} \le 1$.
The Choi operator of the Pauli channel is a Bell diagonal state:
\begin{eqnarray}
\label{eq:bell-diagonal}
\rho = \sum_{\san{k},\san{l} \in \mathbb{F}_2}
  P_{\san{KL}}(\san{k},\san{l}) \ket{\psi(\san{k},\san{l})}\bra{\psi(\san{k},\san{l})},
\end{eqnarray}
where $P_{\san{KL}}$ is a distribution on $\mathbb{F}_2 \times
\mathbb{F}_2$ defined by
\begin{eqnarray}
\begin{array}{rcl}
P_{\san{KL}}(0,0) &=& \frac{1 + e_\san{z} + e_\san{x} + e_\san{y}}{4}, \\
P_{\san{KL}}(0,1) &=& \frac{1 + e_\san{z} - e_\san{x} - e_\san{y}}{4}, \\
P_{\san{KL}}(1,0) &=& \frac{1 - e_\san{z} + e_\san{x} - e_\san{y}}{4}, \\
P_{\san{KL}}(1,1) &=& \frac{1 - e_\san{z} - e_\san{x} + e_\san{y}}{4},
\end{array}
\label{eq:relation-P-e}
\end{eqnarray}
and 
\begin{eqnarray*}
\ket{\psi(0,0)} &:=& \frac{\ket{00} + \ket{11}}{\sqrt{2}}, \\
\ket{\psi(1,0)} &:=& \frac{\ket{01} + \ket{10}}{\sqrt{2}}, \\
\ket{\psi(0,1)} &:=& \frac{\ket{00} - \ket{11}}{\sqrt{2}}, \\
\ket{\psi(1,1)} &:=& \frac{\ket{01} - \ket{10}}{\sqrt{2}}.
\end{eqnarray*}
We occasionally abbreviate $P_{\san{KL}}(\san{k},\san{l})$ as
$p_{\san{kl}}$.
Note that the Pauli channel is a special class of the unital channel
discussed in Section \ref{subsec:unital-channel}.

The following lemma simplify the calculation of 
Eq.~(\ref{eq:asymptotic-key-rate-two-way-bb84})
for a Pauli channel.
\begin{lemma}
\label{lemma:calc-bell-diagonal}
For a Bell diagonal Choi operator $\rho$, the minimizations
in  Eqs.~(\ref{eq:asymptotic-key-rate-two-way-bb84}) 
(\ref{eq:asymptotic-key-rate-two-way-reverse-bb84}) are achieved by
a Bell diagonal operator $\varrho \in {\cal P}_c(\omega)$.
\end{lemma}
\begin{proof}
This lemma is a straightforward corollary of 
Proposition \ref{proposition:convexity-2}.
\end{proof}

\begin{lemma}
\label{lemma-optimal-function}
For Bell diagonal state $\rho$, the asymptotic key 
generation rate is maximized when we employ the
functions $\chi_A, \chi_B$ given by
Eqs.~(\ref{eq:chi-A}) and (\ref{eq:chi-B}).
\end{lemma}
\begin{proof}
Since $H_\rho(X_2|W_1 = 1, Y_1 Y_2) = 1$
and $H_\rho(X_2|W_1 = 1, E_1 E_2) \le 1$,
$X_2$ should be discarded if $W_1 = 1$.
Similarly, $Y_2$ should be discarded if $W_1 = 0$.
Since the Bell diagonal Choi operator is symmetric with respect
to Alice and Bob's subsystem, we have
\begin{eqnarray*}
H_\rho(X_2|W_1 = 0, U_1 E_1 E_2) = H_\rho(Y_2|W_1 = 0, U_1 E_1 E_2),
\end{eqnarray*} 
and
\begin{eqnarray*}
H_\rho(X_2|W_1 = 0, Y_1 Y_2) = H_\rho(Y_2| W_1 = 0, X_1 X_2).
\end{eqnarray*}
Furthermore, we have
\begin{eqnarray}
\label{eq:not-both}
H_\rho(Y_2|W_1 = 0, U_1 X_2 E_1 E_2) \le
H_\rho(Y_2|W_1 = 0, X_1 X_2).
\end{eqnarray}
Therefore,
the functions given by Eqs.~(\ref{eq:chi-A}) and (\ref{eq:chi-B})
are optimal. Note that Eq.~(\ref{eq:not-both}) means
that we should not keep $Y_2$ if we keep $X_2$.
\end{proof}

By Lemmas \ref{lemma:calc-bell-diagonal} and \ref{lemma-optimal-function}, 
it suffice to consider the functions
given by Eqs.~(\ref{eq:chi-A}) and (\ref{eq:chi-B}) if the channel is
a Pauli channel. Therefore, we employ the functions 
given by Eqs.~(\ref{eq:chi-A}) and (\ref{eq:chi-B})
throughout this subsection. Furthermore, we can find that
the asymptotic key generation rates for the direct and the reverse 
IR procedure coincide, because 
$H_\rho(U_1|W_1 E_1 E_2) = H_\rho(V_1|W_1 E_1 E_2)$
and
$H_\rho(U_1|Y_1 Y_2) = H_\rho(V_1| X_1 X_2)$.
Therefore, we only consider the asymptotic key 
generation rate for the direct IR procedure throughout
this subsection.

\begin{theorem}
\label{theorem:calc-bell-diagonal}
For a Bell diagonal state $\rho$, 
we have
\begin{eqnarray}
\lefteqn{ \frac{1}{2}  \max
  \left[
   H_{\rho}(U_1 U_2|W_1 E_1 E_2) -
   H_\rho(U_1|Y_1 Y_2)
   \right. } \nonumber \\ 
   && ~~~~~~~~~~~~~~~ - H_\rho(U_2|W_1 Y_1
   Y_2),  \nonumber \\
&& ~~~~~~~~~~\left. H_{\rho}(U_2|U_1 W_1 E_1 E_2) -
    H_\rho(U_2|W_1 Y_1
   Y_2) \right], \nonumber \\
&=& \max[
1 - H(P_{\san{KL}}) \nonumber \\
&&  + \frac{P_{\bar{\san{K}}}(1)}{2}
h\left( \frac{p_{00} p_{10} + p_{01} p_{11}}{
(p_{00} + p_{01})(p_{10} + p_{11})}\right), \nonumber \\
&& \hspace{10mm}
\frac{P_{\bar{\san{K}}}(0)}{2} (
1 - H(P_{\san{KL}}^\prime) ) ] ,
\label{eq:two-way-asymptotic-key-pauli}
\end{eqnarray}
where 
\begin{eqnarray*}
P_{\bar{\san{K}}}(0) &:=& (p_{00} + p_{01})^2
+ (p_{10} + p_{11})^2, \\
P_{\bar{\san{K}}}(1) &:=& 2 (p_{00} + p_{01})
(p_{10} + p_{11}),
\end{eqnarray*}
and
\begin{eqnarray*}
P_{\san{KL}}^\prime(0,0) &:=&
\frac{p_{00}^2 + p_{01}^2}{(p_{00} + p_{01})^2 + (p_{10}+p_{11})^2}, \\
P_{\san{KL}}^\prime(1,0) &:=& 
\frac{2 p_{00} p_{01}}{(p_{00} + p_{01})^2 + (p_{10}+p_{11})^2}, \\
P_{\san{KL}}^\prime(0,1) &:=&
\frac{p_{10}^2 + p_{11}^2}{(p_{00} + p_{01})^2 + (p_{10}+p_{11})^2}, \\
P_{\san{KL}}^\prime(1,1) &:=&
\frac{2 p_{10}p_{11} }{(p_{00} + p_{01})^2 + (p_{10}+p_{11})^2}.
\end{eqnarray*}
\end{theorem}
The theorem is proved by a straightforward calculation,
and the proof is presented at the end of this section.

Combining Lemma \ref{lemma:calc-bell-diagonal},
Theorem \ref{theorem:calc-bell-diagonal}, 
and Eq~(\ref{eq:relation-P-e}), it is straightforward to
calculate the asymptotic key generation rate
for a Pauli channel.
As a special case of the Pauli channel,
we consider the depolarizing channel.
The depolarizing channel is parameterized by one
real parameter $e \in [0,1/2]$, and the Bell diagonal
entries of the Choi operator are given by
$p_{00} = 1 - 3 e/2$,
$p_{10} = p_{01} = p_{11} = e/2$.
For the six-state protocol,
it is straightforward to calculate the 
asymptotic key generation rate, which is plotted 
in Fig.~\ref{fig:hash-six-depolarizing}.
According to Lemma \ref{lemma:calc-bell-diagonal},
it is sufficient to take the minimization over the 
subset ${\cal P}_{c,\rom{Bell}}(\omega) \subset {\cal P}_c(\omega)$
that consists of all Bell diagonal operators in ${\cal P}_c(\omega)$.
For the depolarizing channel, the set ${\cal P}_{c,\rom{Bell}}(\omega)$
consists of Bell diagonal state 
$\varrho = \sum_{\san{k},\san{l} \in \mathbb{F}_2}
p_{\san{kl}}^\prime
\ket{\psi(\san{k},\san{l})}\bra{\psi(\san{k},\san{l})}$
satisfying $p_{00}^\prime = 1 - e + \kappa$,
$p_{10}^\prime = p_{11}^\prime = e/2 - \kappa$,
and $p_{11}^\prime = \kappa$ for $\kappa \in [0,e/2]$.
We can calculate the asymptotic key generation rate
by taking the minimum with respect to
the one free parameter $\kappa \in [0,e/2]$,
which is plotted in Fig.~\ref{fig:hash-bb84-depolarizing}.

It should be noted that the asymptotic key generation 
rate of the standard one-way postprocessing \cite{shor:00, lo:01}
is $1 - H(P_{\san{KL}})$ for the six-state protocol
and $\min_{\kappa}[1 - H(P_{\san{KL}})]$ for the
BB84 protocol. 
Therefore, Eq.~(\ref{eq:two-way-asymptotic-key-pauli})
analytically clarifies that the asymptotic key generation rate
of our postprocessing is at least as high as that
of the standard postprocessing.

\begin{figure}
\centering
\includegraphics[width=\linewidth]{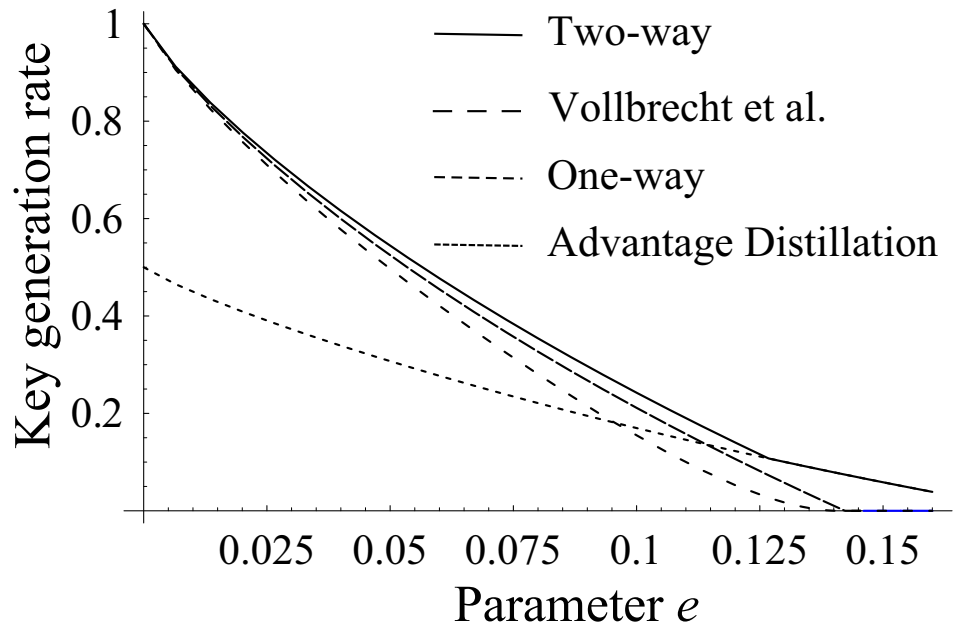}
\caption{Comparison of the asymptotic key generation rates of the  
six-state protocols. 
``Two-way'' is
the asymptotic key generation rate of
the proposed postprocessing. 
``Vollbrecht et al.'' is the asymptotic key generation rate of the two-way postprocessing
of \cite{ma:06, watanabe:06}.
``Advantage Distillation'' is the asymptotic key generation rate 
of the postprocessing with the advantage distillation
 \cite{gottesman:03}.
``One-way'' is the asymptotic key generation rate of
the one-way postprocessing \cite{renner:05}.
It should be noted that the asymptotic key generation rates of the six-state protocols
with the advantage distillation in
\cite{renner:05b, gottesman:03, chau:02, bae:07} are slightly higher
than that of the proposed protocol for much higher error rate.
}
\label{fig:hash-six-depolarizing}
\end{figure}
\begin{figure}
\centering
\includegraphics[width=\linewidth]{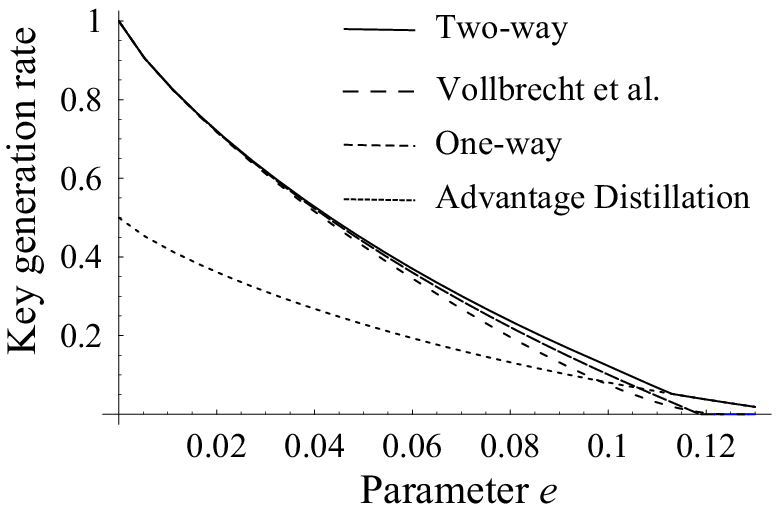}
\caption{ Comparison of the asymptotic key generation rates of the  
BB84 protocols. 
``Two-way'' is
the asymptotic key generation rate of 
the proposed postprocessing. 
``Vollbrecht et al.'' is the asymptotic key generation rate of the two-way 
postprocessing of \cite{ma:06, watanabe:06}.
``Advantage Distillation'' is the asymptotic key generation rate 
of the postprocessing with the advantage distillation
 \cite{gottesman:03}.
``One-way'' is the asymptotic key generation rate of
the one-way postprocessing \cite{renner:05}.
}
\label{fig:hash-bb84-depolarizing}
\end{figure}

\noindent{\em Proof of Theorem \ref{theorem:calc-bell-diagonal}}

Let 
\begin{eqnarray*}
\ket{\psi_{ABE}} &:=& \sum_{\san{k}, \san{l} \in \mathbb{F}_2 }
\sqrt{P_{\san{KL}}(\san{k},\san{l})} \ket{\psi(\san{k},\san{l})} 
\ket{\san{k}, \san{l}} \\
&=& \sum_{x, \san{k} \in \mathbb{F}_2}
\sqrt{P_{\san{K}}(\san{k})} \ket{x, x + \san{k}} \ket{\phi(x,\san{k})}
\end{eqnarray*}
be a purification of $\rho = \sum_{\san{k},\san{l} \in
\mathbb{F}_2} \ket{\psi(\san{k},\san{l})} \bra{\psi(\san{k},\san{l})}$,
where we set
\begin{eqnarray*}
\ket{\phi(x, \san{k})} := \frac{1}{\sqrt{P_{\san{K}}(\san{k})}}
\sum_{\san{l} \in \mathbb{F}_2} (-1)^{x \san{l}} 
\sqrt{P_{\san{KL}}(\san{k}, \san{l})} \ket{\san{k}, \san{l}},
\end{eqnarray*}
and where $P_{\san{K}}(\san{k}) = \sum_{\san{l} \in \mathbb{F}_2}
P_{\san{KL}}(\san{k},\san{l})$
is a marginal distribution.
Then, let
\begin{eqnarray*}
\lefteqn{ \rho_{X_1 X_2 Y_1 Y_2 E_1 E_2} } \\
&=& \sum_{\vec{x}, \vec{\san{k}} \in \mathbb{F}_2^2} \frac{1}{4}
P_{\san{K}}^2(\vec{\san{k}}) \ket{\vec{x}, \vec{x} + \vec{\san{k}}}
\bra{\vec{x}, \vec{x} + \vec{\san{k}}} \otimes \rho_{E_1
E_2}^{\vec{x}, \vec{\san{k}}},
\end{eqnarray*}
where 
\begin{eqnarray*}
\rho_{E_1 E_2}^{\vec{x}, \vec{\san{k}}} :=
\ket{\phi(x_1, \san{k}_1)} \bra{\phi(x_1, \san{k}_1)}
\otimes \ket{\phi(x_2, \san{k}_2)} \bra{\phi(x_2, \san{k}_2)}
\end{eqnarray*}
for $\vec{x} = (x_1,x_2)$ and $\vec{\san{k}} = (\san{k}_1,\san{k}_2)$.

Note that $H(U_1|Y_1 Y_2) = H(W_1)$ for
the Pauli channel.
Let $W_2$ be a random variable defined by
$W_2 := \xi_2(W_1,Y_2) + U_2$. Then, for the 
Pauli channel, we have
$H(U_2|W_1 Y_1 Y_2) = P_{W_1}(0) H(P_{W_2|W_1=0})$.

Noting that 
\begin{eqnarray*}
P_{X_1 X_2 Y_1 Y_2}(\vec{x}, \vec{x} + \vec{\san{k}}) =
\frac{1}{4} P_{\san{K}}^2(\vec{\san{k}}),
\end{eqnarray*}
we have
\begin{eqnarray*}
P_{U_1}(u_1) &=& \frac{1}{2} \\
P_{W_1}(w_1) &=& \sum_{\vec{\san{k}} \in \mathbb{F}_2^2
 \atop \san{k}_1 + \san{k}_2 = w_1 } 
P_{\san{K}}^2(\vec{\san{k}}) \\
P_{U_2|W_1 = 0}(u_2) &=& \frac{1}{2} \\
P_{U_2|W_1 = 1}(u_2) &=& 1 \\
P_{W_2|W_1=0}(w_2) &=& \frac{P_{\san{K}}^2(w_2,w_2)}{P_{W_1}(w_1)} \\
P_{W_2|W_1=1}(0) &=& 1.
\end{eqnarray*}
Using these formulas, we can write
\begin{eqnarray*}
&& \rho_{U_1 U_2 W_1 E_1 E_2} =
\sum_{\vec{u} \in \mathbb{F}_2^2} \sum_{w_1 \in \mathbb{F}_2}
P_{U_1}(u_1) P_{W_1}(w_1) \\
&& ~~~~~~~~~P_{U_2|W_1 = w_1}(u_2)  
 \ket{\vec{u}, w_1}\bra{\vec{u}, w_1} \otimes
\bar{\rho}_{E_1 E_2}^{\vec{u}, w_1} 
\end{eqnarray*}
for $\vec{u} = (u_1,u_2)$, where 
\begin{eqnarray*}
\bar{\rho}_{E_1 E_2}^{\vec{u},w_1} :=
\sum_{w_2 \in \mathbb{F}_2} P_{W_2|W_1=0}(w_2)
\rho_{E_1 E_2}^{\vec{u} G, (w_1,w_2)G}
\end{eqnarray*}
for $w_1 = 0$ and a matrix $G = \left(\begin{array}{cc} 1 & 1 \\ 1 & 0
				      \end{array} \right)$,
and 
\begin{eqnarray*}
\bar{\rho}_{E_1 E_2}^{\vec{u}, w_1} :=
\sum_{a,b \in \mathbb{F}_2} \frac{1}{4} \rho_{E_1 E_2}^{(u_1,a)G, (w_1,b)G}
\end{eqnarray*}
for $w_1 = 1$.

Since supports of rank $1$ matrices $\{ \rho_{E_1 E_2}^{\vec{x},
\vec{\san{k}}} \}_{\vec{\san{k}} \in \mathbb{F}_2^2}$ are
orthogonal to each other, 
$\rho_{E_1 E_2}^{\vec{u}, w_1}$ for $w_1 = 0$ is already
eigen value decomposed.
Applying Lemma \ref{lemma-for-key-rate}
for $\san{J} = \{ 00, 10\}$ and 
$C = C^{\bot} = \{ 00, 11 \}$, we can 
eigen value decompose $\rho_{E_1 E_2}^{\vec{u}, w_1}$
for $w_1 = 1$ as 
\begin{eqnarray*}
\rho_{E_1 E_2}^{\vec{u}, w_1} =
\sum_{b \in \mathbb{F}_2} \frac{1}{2}
\sum_{\vec{\san{j}} \in \san{J} } 
P_{\san{J}|\vec{\san{K}}=\vec{\san{k}}}(\vec{\san{j}})
\ket{\vartheta((u_1,0), \san{k}, \vec{\san{j}})}
\bra{\vartheta((u_1,0), \san{k}, \vec{\san{j}})},
\end{eqnarray*}
where we follow the notations in 
Lemma \ref{lemma-for-key-rate}
for $m=2$.

Thus, we have
\begin{eqnarray}
\lefteqn{ H(\rho_{U_1 U_2 W_1 E_1 E_2}) } \nonumber \\
&=& H(P_{U_1}) + H(P_{W_1}) + \sum_{w_1 \in \mathbb{F}_2}
P_{W_1}(w_1) \{ H(P_{U_2|W_1 = w_1}) \nonumber \\
&& +
\sum_{\vec{u} \in \mathbb{F}_2^2} P_{U_1}(u_1) P_{U_2|W_1 = w_1}(u_2)
H(\rho_{E_1 E_2}^{\vec{u}, w_1}) \} \nonumber \\
&=& 1 + H(P_{\bar{\san{K}}}) + P_{\bar{\san{K}}}(0)\{
1 + H(P_{\vec{\san{K}}|\bar{\san{K}}=0})\} \nonumber \\ 
&&~~ +
P_{\bar{\san{K}}}(1) H(P_{\vec{\san{K}}\san{J}| \bar{\san{K}}=1}). 
\label{entropy-uwe}
\end{eqnarray}

Taking the partial trace of $\rho_{U_1 U_2 W_1 E_1 E_2}$
over systems $U_1, U_2$, we have
\begin{eqnarray*}
\rho_{W_1 E_1 E_2} &=& \sum_{w_1 \in \mathbb{F}_2}
P_{W_1}(w_1) \ket{w_1}\bra{w_1}  \\
&& \otimes
\left(
\sum_{\vec{u} \in \mathbb{F}_2^2} P_{U_1}P_{U_2|W_1 = w_1}(u_2)
\bar{\rho}_{E_1 E_2}^{\vec{u},w_1}
\right).
\end{eqnarray*}
Thus, we have
\begin{eqnarray}
H(\rho_{W_1 E_1 E_2}) &=&
H(P_{W_1}) + \sum_{w_1 \in \mathbb{F}_2} P_{W_1}(w_1) \nonumber \\
&& H \left(
\sum_{\vec{u} \in \mathbb{F}_2^2} P_{U_1}P_{U_2|W_1 = w_1}(u_2)
\bar{\rho}_{E_1 E_2}^{\vec{u},w_1}
\right) \nonumber \\
&=& 
H(P_{\bar{\san{K}}}) + 
\sum_{\bar{\san{k}} \in \mathbb{F}_2}
P_{\bar{\san{K}}}(0) H(P_{\vec{\san{K}}\vec{\san{L}}| \bar{\san{K}}= \bar{\san{k}}}). 
\label{entropy-we}
\end{eqnarray}

Combining Eqs.~(\ref{entropy-uwe}) and (\ref{entropy-we}), we have
\begin{eqnarray*}
\lefteqn{ H_{\rho}(U_1 U_2 | W_1 E_1 E_2) - H(U_1|Y_1 Y_2) 
- H(U_2|U_1 W_1 Y_1 Y_2)} \\
&=& H_{\rho}(U_1 U_2 | W_1 E_1 E_2) - H(P_{W_1}) -
P_{W_1}(0) H(P_{W_2|W_1 = 0})  \\
&=&
2 - H(P_{\vec{\san{K}}\vec{\san{L}}}) +
P_{\bar{\san{K}}}(1) \{
H(P_{\vec{\san{K}}\san{J}| \bar{\san{K}}=1}) - 1
\} \\
&=& 
2 - 2 H(P_{\san{K}\san{L}}) +
P_{\bar{\san{K}}}(1) 
h\left( \frac{p_{00} p_{10} + p_{01} p_{11}}{
(p_{00} + p_{01})(p_{10} + p_{11})}\right).
\end{eqnarray*}

On the other hand, by taking partial trace of
$\rho_{U_1 U_2 W_1 E_1 E_2}$ over the system
$U_1$, we have
\begin{eqnarray*}
\rho_{U_1 W_1 E_1 E_2} &=&
\sum_{u_1, w_1 \in \mathbb{F}_2} \frac{1}{2} P_{W_1}(w_1) 
\ket{u_1,w_1}\bra{u_1,w_1} \\
&& \otimes
\left(
\sum_{u_2 \in \mathbb{F}_2} P_{U_2 | W_1 = w_1}(u_2)
\rho_{E_1 E_2}^{(u_1,u_2), w_1}
\right) .
\end{eqnarray*}
Thus, we have
\begin{eqnarray}
H(\rho_{U_1 W_1 E_1 E_2}) &=&
1 + H(P_{W_1}) + \sum_{u_1, w_1 \in \mathbb{F}_2}
\frac{1}{2} P_{W_1}(w_1) \nonumber \\
&& ~~
H\left(
\sum_{u_2 \in \mathbb{F}_2} P_{U_2 | W_1 = w_1}(u_2)
\rho_{E_1 E_2}^{(u_1,u_2), w_1}
\right) \nonumber \\
&=& 1 + H(P_{\bar{\san{K}}}) + 
\sum_{\bar{\san{k}} \in \mathbb{F}_2} P_{\bar{\san{K}}}(\bar{\san{k}})
H(P_{\vec{\san{K}} \san{J} | \bar{\san{K}}=1}). \nonumber \\
\label{entropy-uwe2}
\end{eqnarray}

Combining Eqs.~(\ref{entropy-uwe}) and (\ref{entropy-uwe2}), we have
\begin{eqnarray*}
\lefteqn{H_{\rho}(U_2|W_1 U_1 E_1 E_2) - H(U_2|W_1 U_1 E_1 E_2) } \\
&=&  H_{\rho}(U_2|W_1 U_1 E_1 E_2) - P_{W_1}(0) 
H(P_{W_2|W_1=0})  \\
&=& P_{\bar{\san{K}}}(0) (
1- H(P_{\san{K}\san{L}}^\prime)
). 
\end{eqnarray*}

\hspace*{\fill}\qed \\
  \vspace{2ex}

\begin{lemma}
\label{lemma-for-key-rate}
Let $C$ be a linear subspace of $\mathbb{F}_2^m$.
Let 
\begin{eqnarray*}
\ket{\varphi^m(\vec{x},\vec{\san{k}})} :=
\frac{1}{\sqrt{P_{\san{K}}^m(\vec{\san{k}})}}
\sum_{\vec{\san{l}} \in \mathbb{F}_2^m}
(-1)^{\vec{x} \cdot \vec{\san{l}}}
\sqrt{P_{\san{KL}}^m(\vec{\san{k}},\vec{\san{l}})}
 \ket{\vec{\san{k}},\vec{\san{l}}}, 
\end{eqnarray*}
and 
$\rho_{E^m}^{\vec{x},\vec{\san{k}}} := \ket{\varphi^m(\vec{x},\vec{\san{k}})} 
\bra{\varphi^m(\vec{x},\vec{\san{k}})}$.
Let $\san{J}$ be a set of coset representatives of the cosets $\mathbb{F}_2^m/C$,
and
\begin{eqnarray*}
P_{\san{J}|\san{K}^m = \vec{\san{k}}}(\vec{\san{j}}) :=
\frac{ \sum_{\vec{\san{c}} \in C^\bot} P_{\san{KL}}^m(\vec{\san{k}}, \vec{\san{j}}+\vec{\san{c}})}
{P_{\san{K}}^m(\vec{\san{k}})}
\end{eqnarray*}
be  conditional probability distributions on $\san{J}$.
Then, for any $\vec{a} \in \mathbb{F}_2^m$, we have
\begin{eqnarray}
\label{eq-mixture-of-eve-state}
\sum_{\vec{x} \in C} \frac{1}{|C|} \rho_{E^m}^{\vec{x} + \vec{a}, \vec{\san{k}}} =
\sum_{\vec{\san{j}} \in \san{J}} P_{\san{J}|\san{K}^m=\vec{\san{k}}}(\vec{\san{j}})
\ket{\vartheta(\vec{a},\vec{\san{k}}, \vec{\san{j}})}
\bra{\vartheta(\vec{a},\vec{\san{k}},\vec{\san{j}})},
\end{eqnarray}
where 
\begin{eqnarray*}
\ket{\vartheta(\vec{a},\vec{\san{k}},\vec{\san{j}})} &:=& 
\frac{1}{\sqrt{\sum_{\vec{\san{e}} \in C^\bot} 
P_{\san{KL}}^m(\vec{\san{k}}, \vec{\san{j}}+\vec{\san{e}})} } \\
&& \sum_{\vec{\san{c}} \in C^\bot} (-1)^{\vec{a} \cdot \vec{\san{c}}} 
 \sqrt{P_{\san{KL}}^m(\vec{\san{k}}, \vec{\san{j}}+\vec{\san{c}})}
\ket{\vec{\san{k}}, \vec{\san{j}} + \vec{\san{c}} }.
\end{eqnarray*}
\end{lemma}
\begin{remark}
\label{remark-of-lemma-for-key-rate}
If $\vec{\san{j}} \neq \vec{\san{i}}$, obviously we have 
$\langle \vartheta(\vec{a},\vec{\san{k}},\vec{\san{j}}) |
\vartheta(\vec{a},\vec{\san{k}},\vec{\san{i}}) \rangle = 0$.
Thus, the right hand side of Eq.~(\ref{eq-mixture-of-eve-state})
is an eigen value decomposition. Moreover, if $\vec{a} + \vec{b} \in C$,
then we have $\ket{\vartheta(\vec{a},\vec{\san{k}},\vec{\san{j}})} =
\ket{\vartheta(\vec{b},\vec{\san{k}},\vec{\san{j}})}$.
\end{remark}
\begin{proof}
For any $\vec{x} \in C$ and $\vec{a} \in \mathbb{F}_2^m$, we can rewrite
\begin{eqnarray*}
\ket{\varphi(\vec{x}+\vec{a},\vec{\san{k}})} &=&
\frac{1}{\sqrt{P_{\san{K}}^m(\vec{\san{k}})}} \sum_{\vec{\san{j}} \in \san{J}}
\sum_{\vec{\san{c}} \in C^\bot} 
(-1)^{(\vec{x}+\vec{a})\cdot(\vec{\san{j}}+\vec{\san{c}})} \\
&& ~~~\sqrt{ P_{\san{KL}}^m(\vec{\san{k}}, \vec{\san{j}} + \vec{\san{c}} )}
\ket{\vec{\san{k}}, \vec{\san{j}} + \vec{\san{c}} } \\
&=& \sum_{\vec{\san{j}} \in \san{J}}
(-1)^{(\vec{x}+\vec{a})\cdot \vec{\san{j}} }
\sqrt{ P_{\san{J}|\san{K}^m = \vec{\san{k}} }(\vec{\san{j}}) }
\ket{\vartheta(\vec{a}, \vec{\san{k}}, \vec{\san{j}} ) }.
\end{eqnarray*}
Then, we have
\begin{eqnarray*}
\lefteqn{ 
\sum_{\vec{x} \in C} \frac{1}{|C|} \rho_{E^m}^{\vec{x}+\vec{a}, \vec{\san{k}} } } \\
&=& \sum_{\vec{x} \in C} \frac{1}{|C|} 
\sum_{\vec{\san{i}}, \vec{\san{j}} \in \san{J} } 
(-1)^{(\vec{x}+\vec{a})\cdot(\vec{\san{i}}+\vec{\san{j}}) }
\sqrt{ P_{\san{J}|\san{K}^m = \vec{\san{k}} }(\vec{\san{i}})
P_{\san{J}|\san{K}^m = \vec{\san{k}} }(\vec{\san{j}}) } \\
&& ~~~\ket{\vartheta(\vec{a}, \vec{\san{k}}, \vec{\san{i}} ) }
\bra{\vartheta(\vec{a}, \vec{\san{k}}, \vec{\san{j}} ) } \\
&=& \sum_{\vec{\san{i}}, \vec{\san{j}} \in \san{J} }
(-1)^{\vec{a} \cdot(\vec{\san{i}}+\vec{\san{j}}) }
\sum_{\vec{x} \in C} \frac{1}{|C|}
(-1)^{\vec{x} \cdot(\vec{\san{i}}+\vec{\san{j}}) }
\sqrt{ P_{\san{J}|\san{K}^m = \vec{\san{k}} }(\vec{\san{i}})
P_{\san{J}|\san{K}^m = \vec{\san{k}} }(\vec{\san{j}}) } \\
&& ~~~\ket{\vartheta(\vec{a}, \vec{\san{k}}, \vec{\san{i}} ) }
\bra{\vartheta(\vec{a}, \vec{\san{k}}, \vec{\san{j}} ) } \\
&=& \sum_{\vec{\san{j}} \in \san{J} }
P_{\san{J}|\san{K}^m = \vec{\san{k}} }(\vec{\san{j}}) 
\ket{\vartheta(\vec{a}, \vec{\san{k}}, \vec{\san{j}} )}
\bra{\vartheta(\vec{a}, \vec{\san{k}}, \vec{\san{j}} )},
\end{eqnarray*}
where $\cdot$ is the standard inner product on the vector space
 $\mathbb{F}_2^m$, and 
we used the following equality,
\begin{eqnarray*}
\sum_{\vec{x} \in C} (-1)^{\vec{x} \cdot (\vec{\san{i}} + \vec{\san{j}} ) }
= 0
\end{eqnarray*}
for $\vec{\san{i}} \neq \vec{\san{j}}$.
\end{proof}

\subsection{Unital Channel}
\label{subsec:two-way-unital-channel}

In this section, we calculate the asymptotic 
key generation rates 
for the Unital channel. Although we succeeded to
show a closed formula of the asymptotic key
generation rate for the Pauli channel,
which is a special class of the unital channel, in 
Section \ref{subsec:pauli-channel}, we do not
know any closed formula of the asymptotic key
generation rate for the unital channel in general.

For the six-state protocol, it is straightforward
to numerically calculate the asymptotic key generation rate. 
For the BB84 protocol, 
owing to Proposition \ref{proposition:minimization-two-way},
the asymptotic key generation rate can be calculated by
taking the minimization over one free parameter
$R_{\san{yy}}$. 

As an example of non Pauli but unital channel, we 
 numerically calculated asymptotic
key generation rates for the depalarizing channel whose
axis is rotated by $\pi/4$, i.e., the channel whose 
Stokes parameterization is given by
\begin{eqnarray}
\left(
\left[
\begin{array}{ccc}
\cos (\pi/4) & -\sin (\pi/4) & 0 \\
\sin(\pi/4) & \cos(\pi/4) & 0 \\
0 & 0 & 1
\end{array}
\right]
\left[
\begin{array}{ccc}
1-2e & 0 & 0 \\
0 & 1-2e & 0 \\
0 & 0 & 1-2e
\end{array}
\right],
\left[
\begin{array}{c}
0 \\ 0 \\ 0
\end{array}
\right]
\right).
\label{eq:quarter-rotated-depolarizing}
\end{eqnarray}

For this channel, 
since the Choi operator is symmetric with respect to
Alice and Bob's subsystem,
we can also show that the asymptotic key 
generation rate is maximized when we employ the
functions $\chi_A, \chi_B$ given by
Eqs.~(\ref{eq:chi-A}) and (\ref{eq:chi-B})
in a similar manner as Lemma \ref{lemma-optimal-function}.
Therefore, we employ the functions 
given by Eqs.~(\ref{eq:chi-A}) and (\ref{eq:chi-B})
throughout this subsection.
Furthermore, we can find that
the asymptotic key generation rates for the direct and the reverse 
IR procedure coincide, because 
$H_\rho(U_1|W_1 E_1 E_2) = H_\rho(V_1|W_1 E_1 E_2)$
and
$H_\rho(U_1|Y_1 Y_2) = H_\rho(V_1| X_1 X_2)$.
Therefore, we only consider the asymptotic key 
generation rate for the direct IR procedure throughout
this subsection.

For the BB84 protocol
and the six-state protocol, the asymptotic key 
generation rate of the postprocessing with the two-way
IR procedure and that of the postprocessing with
the one-way IR procedure are compared in  
Fig.~\ref{fig:quarter-bb84}
and Fig.~\ref{fig:quarter-six-state} respectively.
We find that the asymptotic key generation rates of the 
postprocessing with our proposed two-way IR procedure
is higher than those of the one-way postprocessing,
which suggest that our proposed IR procedure is effective not
only for the Pauli channel, but also for non-Pauli channels.
It should be noted that the asymptotic key generation
rates of the postprocessing with the direct one-way IR procedure
and the reverse one-way IR procedure coincide for this 
example.

\begin{figure}
\centering
\includegraphics[width=\linewidth]{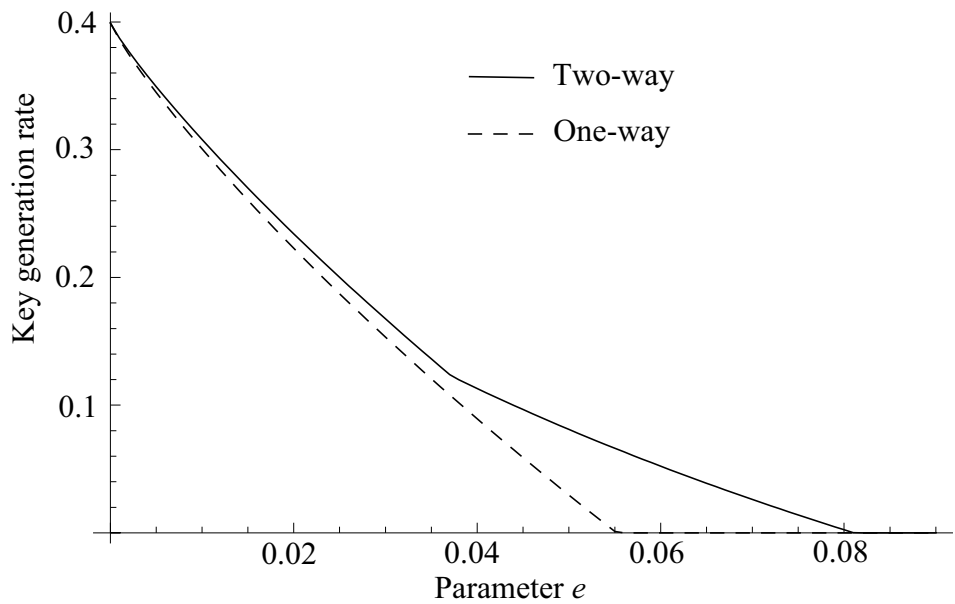}
\caption{ 
Comparison of the asymptotic key generation rates of the  
BB84 protocol. 
``Two-way'' is
the asymptotic key generation rate of the postprocessing with
two-way IR procedure (Eq.~(\ref{eq:asymptotic-key-rate-two-way-bb84})). 
``One-way'' is the asymptotic key generation rate of
the postprocessing with one-way IR procedure
(Eq.~(\ref{eq:key-rate-one-way-bb84})).
}
\label{fig:quarter-bb84}
\end{figure}
\begin{figure}
\centering
\includegraphics[width=\linewidth]{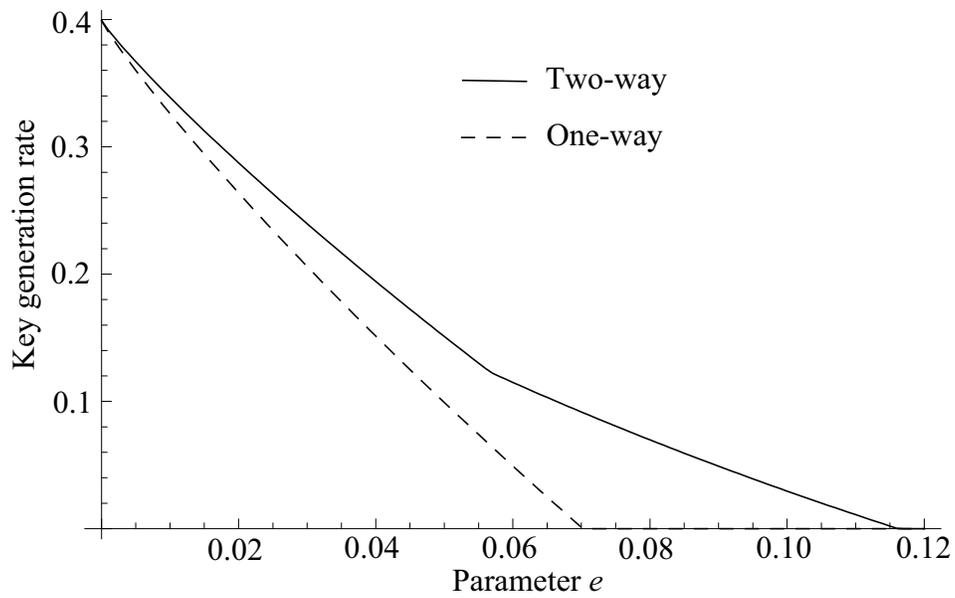}
\caption{ 
Comparison of the asymptotic key generation rates of the  
six-state protocol. 
``Two-way'' is
the asymptotic key generation rate of the postprocessing with
two-way IR procedure (Eq.~(\ref{eq:asymptotic-key-rate-two-way})). 
``One-way'' is the asymptotic key generation rate of
the postprocessing with one-way IR procedure
(Eq.~(\ref{eq:asymptotic-direct})).
}
\label{fig:quarter-six-state}
\end{figure}

\subsection{Amplitude Damping Channel}
\label{subsec:two-way-amplitude}

In this section, we calculate the asymptotic
key generation rates (for the direct two-way IR procedure
and the reverse two-way IR procedure)
for the amplitude damping channel.
Although we succeeded to derive
a closed formulae of the asymptotic key generation
rates of the one-way postprocessing in 
Section \ref{subsec:amplitude-damping}, we do not know
any closed formula of the asymptotic key generation rates
of the postprocessing with the two-way IR procedure
for the amplitude damping channel. Furthermore, it is not
clear whether the asymptotic key generation rate is maximized when we 
employ the functions 
given by Eqs.~(\ref{eq:chi-A}) and (\ref{eq:chi-B}).
Therefore, we (numerically) optimize the choice of the functions
$\chi_A, \chi_B$ so that the asymptotic key generation rate
is maximized. 

Since the set ${\cal P}_c(\omega)$ consists of
only $\rho$ itself for both the BB84 protocol
(refer Section
\ref{subsec:amplitude-damping}), 
we can easily conduct 
the numerical calculation of the asymptotic key generation
rates for the six-state protocol and the BB84 protocol. 
The asymptotic key generation rates of the postprocessing
with the direct two-way IR procedure, the reverse two-way IR
procedure, the direct one-way IR procedure, and the reverse
one-way IR procedure are compared in Fig.~\ref{fig:two-way-amplitude}.
It should be noted that the asymptotic key generation rates for
the BB84 protocol and the six-state protocol coincide in
this example. 
We numerically found that the functions given by
$\chi_A(a_1,a_2) := 1$ and
\begin{eqnarray*}
\chi_B(a_1,a_2) = \left\{\begin{array}{ll}
0 & \mbox{if } a_1 = a_2 \\
1 & \mbox{else}
\end{array}
\right.
\end{eqnarray*}
maximizes the asymptotic key generation rates for 
both the direct two-way IR procedure and the reverse
IR procedure.

\begin{figure}
\centering
\includegraphics[width=\linewidth]{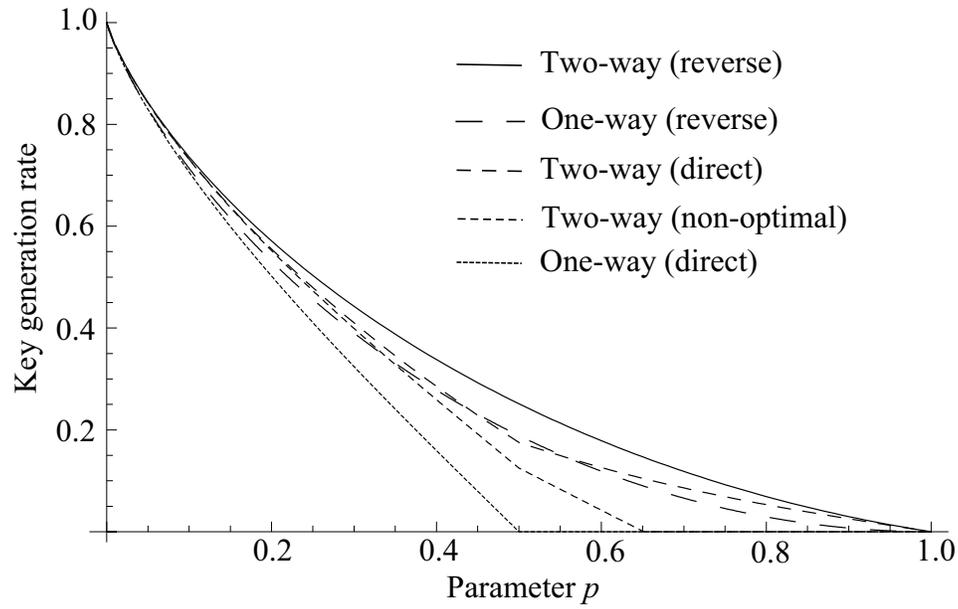}
\caption{ 
Comparison of the asymptotic key generation rates. 
``Two-way (reverse)'' is
the asymptotic key generation rate of the postprocessing with
reverse two-way IR procedure 
(Eq.~(\ref{eq:asymptotic-key-rate-two-way-reverse})).
``One-way (reverse)'' is the asymptotic key generation rate of
the postprocessing with reverse one-way IR procedure
(Eq.~(\ref{eq:asymptotic-reverse})).
 ``Two-way (direct)'' is
the asymptotic key generation rate of the postprocessing with
direct two-way IR procedure 
(Eq.~(\ref{eq:asymptotic-key-rate-two-way})).
``Two-way (non-optimal)'' is the asymptotic key generation rate
of the postprocessing with direct two-way IR procedure when we employ
the functions $\chi_A, \chi_B$ given by 
Eqs.~(\ref{eq:chi-A}) and (\ref{eq:chi-B}).
``One-way (direct)'' is the asymptotic key generation rate of
the postprocessing with one-way IR procedure
(Eq.~(\ref{eq:asymptotic-direct})).
}
\label{fig:two-way-amplitude}
\end{figure}


\section{Relation to Entanglement Distillation Protocol}
\label{sec:relation-edp}

As is mentioned in Chapter \ref{ch:introduction},
the security of the QKD protocols have been studied by using 
the quantum error correcting code and the entanglement
distillation protocol (EDP) since Shor and Preskill found
the relation between them \cite{shor:00}. 
The crucial point in Shor and Preskill's proof is to find an EDP
that corresponds to a postprocessing of the QKD protocols.
Indeed, the security of 
the QKD protocols with the two-way classical communication
\cite{gottesman:03} was proved by finding the corresponding 
EDPs. 

We will explain the EDP proposed by Vollbrecht and Vestraete
\cite{vollbrecht:05} in this section.
Then, we present the postprocesing\footnote{The postprocessing presented
in this section is a modified version of the postprocessing
presented in \cite{ma:06, watanabe:06} so that it fit into 
the notations in this thesis.} of the QKD
protocols that corresponds to Vollbrecht and Vestraete's EDP.
Furthermore, we compare the posptocessing
(corresponding to Vollbrecht and Vestraete's EDP) and the postprocessing
shown in Section \ref{sec:key-rate-two-way}, and clarify the relation
between them,
where we employ the functions given by
Eqs.~(\ref{eq:chi-A}) and (\ref{eq:chi-B}).
 The comparison result suggests\footnote{Renner {\em et
al}.~suggested  that there exist no EDP which corresponds to
the noisy preprocessing (see Remark \ref{remark:noisy-preprocessing}) 
proposed by themselves.} 
that there exists no
EDP that corresponds to the postprocessing shown in 
Section \ref{sec:key-rate-two-way}.

Suppose that Alice and Bob share $2n$ pairs bipartite
qubits systems, and the state of each bipartite system is
a Bell diagonal state\footnote{There is an entanglement distillation
protocol that works for  bipartite states that are not 
necessarily Bell diagonal states \cite{devetak:04}.
However, we only consider EDPs for the Bell diagonal states.}
\begin{eqnarray}
\rho = \sum_{\san{k},\san{l} \in \mathbb{F}_2}
  P_{\san{KL}}(\san{k},\san{l}) \ket{\psi(\san{k},\san{l})}\bra{\psi(\san{k},\san{l})}.
\end{eqnarray} 
The EDP is a protocol to distill the mixed entangled state 
$\rho^{\otimes 2n}$ into the maximally entangled state 
$\ket{\psi}^{\otimes \ell}$ by using the local operation
and the classical communication \cite{bennett:96b}.

Vollbrecht and Vestraete proposed the following 
EDP \cite{vollbrecht:05}, where it is slightly modified 
(essentially the same) from the original version because we want
to clarify the relation among this EDP, the corresponding postprocessing, and 
the postprocessing shown in Section \ref{sec:key-rate-two-way}.
\begin{enumerate}
\renewcommand{\theenumi}{\roman{enumi}}
\renewcommand{\labelenumi}{(\theenumi)}
\item Alice and Bob divide $2n$ pairs of the bipartite 
systems into $n$ blocks of length $2$, and locally carry
out the controlled-NOT (CNOT) operation on each 
block, where the $2i$th pair is the source and the 
$(2i-1)$th pair is the target. 

\item Then, Alice and Bob undertake
the breeding protocol \cite{bennett:96} to guess
bit-flip errors in the $(2i-1)$th pair for all $i$.
The guessed bit-flip errors can be described by a
sequence $\hat{\bol{w}}_1$ (Note that two-way 
classical communication is used in this step).

\item According to $\hat{\bol{w}}_1$, Alice and Bob classify indices
of blocks into two sets $\san{T}_0 := \{i \mymid \hat{w}_i = 0\}$
and $\san{T}_1 := \{ i \mymid \hat{w}_i = 1\}$.

\item For a collection of $2i$th pairs such that $i \in \san{T}_0$,
Alice and Bob conduct the breeding protocol to correct bit-flip
errors.

\item For a collection of $2i$th pairs such that $i \in \san{T}_1$,
Alice and Bob perform measurements in the $\san{z}$-basis,
and obtain measurement results $\bol{x}_{2,\san{T}_1}$
and $\bol{y}_{2,\san{T}_1}$ respectively.

\item Alice sends $\bol{x}_{2,\san{T}_1}$ to Bob.

\item Alice and Bob correct the phase errors for the 
remaining pairs by using information $\san{T}_0$,
$\san{T}_1$, and the bit-flip error 
$\bol{x}_{2,\san{T}_1} + \bol{y}_{2, \san{T}_1}$.
\end{enumerate}

The yield of this EDP is given by
\begin{eqnarray}
1 - H(P_{\san{KL}}) 
+ \frac{P_{\bar{\san{K}}}(1)}{4} \left\{
 h\left( 
\frac{p_{01}}{p_{00} + p_{01} }
\right) 
+ h\left(
\frac{p_{11}}{p_{10} + p_{11}}
\right) \right\}.
      \label{eq-rate-of-volbrecht}
\end{eqnarray}
We can find by the concavity of the binary entropy function
that the first argument in the maximum of the r.h.s.
of Eq.~(\ref{eq:two-way-asymptotic-key-pauli}) is larger than
the value in Eq.~(\ref{eq-rate-of-volbrecht}).

If we convert this EDP into a postprocessing of the QKD protocols, 
the difference between that postprocessing 
and ours is as follows.
In the postprocessing converted from 
the EDP \cite{vollbrecht:05}, after Step (\ref{step4-two-way-IR}), 
Alice reveals the sequence, $\bol{x}_{2, \hat{\san{T}}_1}$,
which consists of the second bit, $x_{i2}$, of the $i$th
block such that the parity of discrepancies $\hat{w}_{i1}$
is $1$. 
However, Alice discards $\bol{x}_{2, \hat{\san{T}}_1}$
in the proposed IR protocol of Section
 \ref{sec:two-way-ir}.
Since sequence $\bol{x}_{2,\hat{\san{T}}_1}$ has 
some correlation to sequence $\bol{u}_1$ from
the view point of Eve, Alice should not reveal 
$\bol{x}_{2, \hat{\san{T}}_1}$ 
to achieve a higher key generation rate.

In the EDP context, on the other hand, since the bit flip error,
$\bol{x}_{2, \hat{\san{T}}_1} + \bol{y}_{2, \hat{\san{T}}_1}$, has some
 correlation to the phase flip errors in the $(2i-1)$-th pair with 
$i \in \hat{\san{T}}_1$, Alice should send the measurement results,
$\bol{x}_{2, \hat{\san{T}}_1}$, to Bob. If Alice discards measurement
results $\bol{x}_{2, \hat{\san{T}}_1}$ without telling Bob what
the result is, then the yield of the resulting EDP is worse than
Eq.~(\ref{eq-rate-of-volbrecht}). 
Consequently, there seems to be no correspondence between the EDP
and our proposed classical processing.


\section{Summary}
\label{sec:summary-of-chapter4}

The results in this chapter is summarized as 
follows: In Section \ref{sec:advantage-distillation},
we reviewed the advantage distillation.
In Section \ref{sec:two-way-ir}, we proposed the
two-way IR procedure.
In Section \ref{sec:key-rate-two-way},
we derived 
a sufficient condition on the key generation
rate such that a secure key agreement
is possible with our proposed postprocessing
(Theorem \ref{theorem:security-two-way}).
We also derived the asymptotick key generation
rate formulae.

In Section \ref{sec:example-two-way},
we investigated the asymptotic key generation rate
of our proposed postprocessing. 
Especially in Section \ref{subsec:pauli-channel}, we derived 
a closed form of the asymptotic key generation rate
for the Pauli channel (Theorem \ref{theorem:calc-bell-diagonal}), 
which clarifies that the asymptotic key generation rate of our proposed
postprocessing is at least as high as the asymptotic key 
generation rate of the standard postprocessing.
We also numerically clarified that the asymptotic
key generation rate of our proposed postprocessing
is higher than the asymptotic key generation rate
of any other postprocessing for 
the Pauli channel (Section \ref{subsec:pauli-channel}), 
the unital channel (Section \ref{subsec:two-way-unital-channel}), and
the amplitude damping channel (Section \ref{subsec:two-way-amplitude})
respectively.

Finally in Section \ref{sec:relation-edp},
we clarified the  relation between our proposed
postprocessing and the EDP proposed by
Vollbrecht and Vestraete \cite{vollbrecht:05}.

\chapter{Conclusion}
\label{ch:conclusions}

In this thesis, we investigated the channel estimation phase
and the postprocessing phase of the QKD protocols.
The contribution of this thesis is summarized as follows.

For the channel estimation phase, we proposed a new
channel estimation procedure in which
we use the mismatched measurement outcomes in addition
to the samples from the matched measurement outcomes.
We clarified that the key generation rate decided according
to our proposed channel estimation procedure is at least as
high as the key generation rate decided according to
the conventional channel estimation procedure.
We also clarified that the former is strictly higher than
the latter for the amplitude damping channel and
the unital channel.

For the postprocessing phase, we proposed a new
kind of postprocessing procedure with two-way public communication.
For the Pauli channel, we clarified that
the key generation rate of the QKD protocols with our proposed
postprocessing is higher than the key generation
rate of the QKD protocols with the standard one-way postprocessing.
For the Pauli channel,
the amplitude damping channel, and the
unital channel, we numerically clarified that
 the QKD protocols with our proposed
postprocessing is higher than the key generation
rate of the QKD protocols with any other postprocessing.

There are some problems that should be investigated
in a future. 
\begin{itemize}
\item To show the necessary and sufficient condition
on the channel for that the (asymptotic) key generation rate decided according
our proposed channel estimation procedure is
strictly higher than that decided 
according to the conventional channel estimation procedure
for the six-state protocol.

\item To analytically show that the (asymptotic) 
key generation rate of our proposed two-way postprocessing is at
least as high as that of the standard one-way postprocessing,
or to find a counter example.
\end{itemize}

\appendix

\chapter{Notations}


\section*{Notations first appeared in Chapter \ref{ch:preliminaries}}

\begin{center}
\begin{longtable}{p{2.5cm}|| p{8.5cm}}
${\cal P}({\cal X})$ 
& the set of all probability distributions on the set
${\cal X}$ \\
$P_X, P_{XY}$ & probability distributions \\
$P_{\bol{x}}$ & the type of the sequence $\bol{x}$ \\
${\cal P}({\cal H})$ & the set of all density operators
on the quantum system ${\cal H}$ \\
${\cal P}^\prime({\cal H})$ & the set of all non-negative 
operators on ${\cal H}$ \\
$\rho, \rho_{AB}$ & density operators \\
$\| \cdot \|$ & the trace distance (variational distance) \\
$F(\cdot, \cdot)$ & the fidelity \\
$H(X)$ & the entropy of the random variable $X$ \\
$H(P_X)$ & the entropy of the random variable with the
 distribution $P_X$ \\
$h(\cdot)$ & the binary entropy function \\
$H(X|Y)$ & the (Shannon) conditional entropy of $X$ given $Y$ \\
$I(X;Y)$ & the mutual information between $X$ and $Y$ \\
$H(\rho)$ & the von Neumann entropy of the system whose state is $\rho$ \\
$H_\rho(A|B)$ & the conditional von Neumann entropy of the system $A$
conditioned by the system $B$ \\
$I_\rho(A; B)$ & the quantum mutual information between the systems
$A$ and $B$ \\
$\sigma_\san{x}, \sigma_\san{y}, \sigma_\san{z}$ & the Pauli operators
 \\
$\ket{\psi}$ & the maximally entangled state defined in 
Eq.~(\ref{eq:definition-of-maximally-entangled}) \\
${\cal P}_c$ & the set of all Choi
 operators \\
$(R, t)$ & the Stokes parameterization of the channel \\
$H_{\min}(\rho_{AB}|\sigma_B)$ & the min-entropy of $\rho_{AB}$ relative
 to $\sigma_B$ \\
$H_{\max}(\rho_{AB}|\sigma_B)$ & the max-entropy of $\rho_{AB}$ relative
 to $\sigma_B$ \\
$H_{\min}^\varepsilon(\rho_{AB}|B)$ & the $\varepsilon$-smooth min-entropy of
$\rho_{AB}$ given the system $B$ \\
$H_{\max}^\varepsilon(\rho_{AB}|B)$ & the $\varepsilon$-smooth max-entropy of
$\rho_{AB}$ given the system $B$  \\
${\cal B}^{\varepsilon}(\rho)$ & the set of all operators 
$\bar{\rho} \in {\cal P}^\prime({\cal H})$ such that 
$\| \bar{\rho} - \rho\| \le \rom{Tr}[\rho] \varepsilon$ \\
$d(\rho_{AB}|B)$ & the distance from the uniform 
(see Definition \ref{definition-distance-from-uniform})
\end{longtable}
\end{center}

\section*{Notations first appeared in Chapter
 \ref{ch:channel-estimation}}

\begin{center}
\begin{longtable}{p{2.5cm}|| p{8.5cm}}
$\ket{0_\san{a}}, \ket{1_\san{a}}$ & the eigenstates of the Pauli
 operator $\sigma_\san{a}$ \\
$\rho_{\bol{X}\bol{Y}\bol{E}}$ & the $\{ccq\}$-state
describing Alice and Bob's bit sequences $(\bol{X},\bol{Y})$
and the state in Eve's system \\
$M$ & the parity check matrix \\
$t$ & the syndrome \\
$P_{XY}$ & the probability distribution of Alice and Bob's bits \\
$P_W$ & the probability distribution of the discrepancy between Alice
 and Bos's bits \\
$\omega$ & the components
 $(R_\san{zz},R_\san{zx},R_\san{xz},R_\san{xx},t_\san{z},t_\san{x})$
of the Stokes parameterization \\
$\tau$ & the components 
 $(R_\san{zy},R_\san{xy},R_\san{yz},R_\san{yx},R_\san{yy},t_\san{y})$
of the Stokes parameterization \\
$\Omega$ & the range of $\omega$ \\
${\cal P}_c(\omega)$ & the set of all Choi operator for a fixed $\omega$ \\
$\gamma$ & the components
 $(R_\san{zz},R_\san{xx}, R_\san{yy})$ of the Stokes parameterization \\
$\kappa$ & the components
 $(R_\san{zx}$,$R_\san{zy}$,$R_\san{xz}$,$R_\san{xy}$,$R_\san{yz}$,$R_\san{yx}$,$t_\san{z}$,$t_\san{x}$,$t_\san{y})$ of the Stokes parameterization \\
$\Gamma$ & the range of $\gamma$ \\
${\cal P}_c(\gamma)$ & the set of all Choi operator for a fixed $\gamma$ \\
$\upsilon$ & the components
 $(R_\san{zz}, R_\san{xx})$ of the Stokes parameterization \\
$\varsigma$ & the components
 $(R_\san{zx}$,$R_\san{zy}$,$R_\san{xz}$,$R_\san{xy}$,$R_\san{yz}$,$R_\san{yx}$,$R_\san{yy}$,$t_\san{z}$,$t_\san{x}$,$t_\san{y})$ of the Stokes
 parameterization \\
$\Upsilon$ & the range of $\upsilon$ \\
${\cal P}_c(\upsilon)$ & the set of all Choi operators for a fixed
 $\upsilon$
\end{longtable}
\end{center}

\section*{Notations first appeared in Chapter
 \ref{ch:postprocessing}}

\begin{center}
\begin{longtable}{p{2.5cm}|| p{8.5cm}}
$\xi$ & the function $\xi:\mathbb{F}_2^2 \to \mathbb{F}_2$ such that
$\xi(a_1,a_2) = a_1 + a_2$ \\
$\zeta$ & the function $\zeta:\mathbb{F}_2^2 \to \mathbb{F}_2$
such that $\zeta(a,0) = a$ and $\zeta(a,1) = 0$ \\
$\chi_A, \chi_B$ & arbitrary functions from $\mathbb{F}_2^2$ to
 $\mathbb{F}_2$ \\
$\zeta_A$ & the function $\mathbb{F}_2^3 \to \mathbb{F}_2$ such that
$\zeta_A(a_1,a_2,a_3) = a_1$ for $\chi_A(a_2,a_3) = 0$ and 
$\zeta_A(a_1,a_2,a_3) = 0$ for else \\
$\zeta_B$ & the function $\mathbb{F}_2^3 \to \mathbb{F}_2$ such that
$\zeta_B(a_1,a_2,a_3) = a_1$ for $\chi_B(a_2,a_3) = 0$ and 
$\zeta_B(a_1,a_2,a_3) = 0$ for else \\
$U_1$ & the random variable defined as $U_1 = \xi(X_1,X_2)$ \\
$V_1$ & the random variable defined as $V_1 = \xi(Y_1, Y_2)$ \\
$W_1$ & the random variable defined as $W_1 = U_1 + V_1$ \\
$U_2$ & the random variable defined as $U_2 = \zeta(X_2,W_1)$ or
the random variable defined as $U_2 = \zeta_A(X_2,U_1,V_1)$ \\
$V_2$ & the random variable defined as $V_2 = \zeta(Y_2,W_1$ or
the random variable defined as $V_2 = \zeta_B(X_2,U_1,V_1)$ \\
$\ket{\psi(\san{k},\san{l})}$ & Bell states \\
$P_{\san{KL}}$ & the distribution such that the Bell diagonal components
of a Bell diagonal state 
\end{longtable}
\end{center}

\chapter{Publications Related to This Thesis}

\section*{Articles in Journals}

\begin{itemize}
\item S.~Watanabe, R.~Matsumoto, T.~Uyematsu, and Y.~Kawano, ''Key rate
      of quantum key distribution with hashed two-way classical
      communication,'' {\em Phys.~Rev.~A}, vol.~76,
      no.~3,pp.~032312-1--7, Sep. 2007.
\item S.~Watanabe, R.~Matsumoto, and T.~Uyematsu, ''Tomography increases
      key rate of quantum-key-distribution protocols,'' {\em
      Phys.~Rev.~A}, vol.~78, no.~4, pp.~042316-1--11, Oct. 2008. 
\end{itemize}

\section*{Peer-Reviewed Articles in International Conferences}

\begin{itemize}
\item S.~Watanabe, R.~Matsumoto, and T.~Uyematsu, ''Security of quantum
      key distribution protocol with two-way classical communication
      assisted by one-time pad encryption,'' in Proc. Asian Conference
      on Qauntum Information Science 2006, Beijing, China, September 2006.
\item S.~Watanabe, R.~Matsumoto, T.~Uyematsu, and Y.~Kawano, ''Key rate
      of quantum key distribution with hashed two-way classical
      communication,'' in {\em Proc.~2007 IEEE
      Int.~Symp.~Inform.~Theory}, Nice, France, June, 2007.
\end{itemize}

\section*{Non-Reviewd Articles in Conferences}

\begin{itemize}
\item S.~Watanabe, R.~Matsumoto, T.~Uyematsu, and Y.~Kawano, ''Key rate
      of quantum key distribution with hashed two-way classical
      communication,'' in {\em Proc.~QIT 16}, Atsugi, Japan, May, 2006.

\item S.~Watanabe, R.~Matsumoto, and T.~Uyematsu, ''Tomography increases
      key rate of quantum-key-distribution protocols,'' presented at
 recent result session in {\em 2008 IEEE Int.~Symp.~Inform.~Theory},
 Toronto, Canada, July, 2008.

\item S.~Watanabe, R.~Matsumoto, and T.~Uyematsu, ''Tomography increases
      key rate of quantum-key-distribution protocols,'' in {\em
 Proc.~SITA 2008}, Kinugawa, Japan, Oct., 2008.

\item S.~Watanabe, R.~Matsumoto, and T.~Uyematsu, ''Tomography increases
      key rate of quantum-key-distribution protocols,'' presented at
      GSIS Workshop on Quantum Information Theory, Sendai, Japan,
      November 2008.
\end{itemize}

\newcommand{\etalchar}[1]{$^{#1}$}


\end{document}